\documentclass[aps,prl,groupedaddress,superscriptaddress,onecolumn,notitlepage,nofootinbib]{revtex4-2}

\usepackage[latin1]{inputenc}
\usepackage{amsthm}
\usepackage{amssymb}
\usepackage{amsmath}
\usepackage{bbold}
\usepackage{bbm}
\usepackage[pdftex]{hyperref}
\usepackage{braket}
\usepackage{dsfont}
\usepackage{mathdots}
\usepackage{mathtools}
\usepackage{enumerate}
\usepackage[shortlabels]{enumitem}
\usepackage{csquotes}
\usepackage{stmaryrd}
\usepackage[cal=boondox]{mathalfa}
\usepackage{graphicx}
\usepackage{stackengine}
\usepackage{scalerel}
\usepackage{array}
\usepackage{makecell}
\newcolumntype{x}[1]{>{\centering\arraybackslash}p{#1}}
\usepackage{tikz}
\usepackage{pgfplots}
\usetikzlibrary{shapes.geometric, shapes.misc, positioning, arrows, arrows.meta, decorations.pathreplacing, angles, quotes}
\usepackage{booktabs}
\usepackage{xfrac}
\usepackage{siunitx}
\usepackage{centernot}
\usepackage{comment}
\usepackage{chngcntr}
\usepackage{caption}
\usepackage{subcaption}

\newtheorem{thm}{Theorem}
\newtheorem*{thm*}{Theorem}
\newtheorem{prop}[thm]{Proposition}
\newtheorem*{prop*}{Proposition}
\newtheorem{lemma}[thm]{Lemma}
\newtheorem*{lemma*}{Lemma}
\newtheorem{cor}[thm]{Corollary}
\newtheorem*{cor*}{Corollary}

\newtheorem*{cj*}{Conjecture}
\newtheorem{Def}[thm]{Definition}
\newtheorem*{Def*}{Definition}

\makeatletter
\def\thmhead@plain#1#2#3{%
  \thmname{#1}\thmnumber{\@ifnotempty{#1}{ }\@upn{#2}}%
  \thmnote{ {\the\thm@notefont#3}}}
\let\thmhead\thmhead@plain
\makeatother

\theoremstyle{definition}
\newtheorem{rem}[thm]{Remark}
\newtheorem*{note}{Note}
\newtheorem{ex}[thm]{Example}

\newtheorem{manuallemmainner}{Lemma}
\newenvironment{manuallemma}[1]{%
  \renewcommand\themanuallemmainner{#1}%
  \manuallemmainner \it
}{\endmanuallemmainner}

\newcommand{\bb}{\begin{equation}\begin{aligned}\hspace{0pt}}
\newcommand{\bbb}{\begin{equation*}\begin{aligned*}}
\newcommand{\ee}{\end{aligned}\end{equation}}
\newcommand{\eee}{\end{aligned*}\end{equation*}}
\newcommand*{\coloneqq}{\mathrel{\vcenter{\baselineskip0.5ex \lineskiplimit0pt \hbox{\scriptsize.}\hbox{\scriptsize.}}} =}
\newcommand*{\eqqcolon}{= \mathrel{\vcenter{\baselineskip0.5ex \lineskiplimit0pt \hbox{\scriptsize.}\hbox{\scriptsize.}}}}

\newcommand{\eqt}[1]{\stackrel{\mathclap{\scriptsize \mbox{#1}}}{=}}
\newcommand{\leqt}[1]{\stackrel{\mathclap{\scriptsize \mbox{#1}}}{\leq}}
\newcommand{\geqt}[1]{\stackrel{\mathclap{\scriptsize \mbox{#1}}}{\geq}}
\newcommand{\ketbra}[1]{\ket{#1}\!\!\bra{#1}}
\newcommand{\ketbraa}[2]{\ket{#1}\!\!\bra{#2}}
\newcommand{\sumno}{\sum\nolimits}

\newcommand{\G}{\mathrm{\scriptscriptstyle G}}

\newcommand{\tcb}[1]{{\color{blue!85!black} #1}}

\newcommand{\id}{\mathds{1}}
\newcommand{\R}{\mathds{R}}
\newcommand{\N}{\mathds{N}}
\newcommand{\C}{\mathds{C}}

\DeclareMathOperator{\Tr}{Tr}
\DeclareMathOperator{\rk}{rk}
\DeclareMathOperator{\cl}{cl}
\DeclareMathOperator{\co}{conv}
\DeclareMathOperator{\cone}{cone}

\DeclareMathAlphabet{\pazocal}{OMS}{zplm}{m}{n}

\DeclareMathOperator{\supp}{supp}
\DeclareMathOperator{\spec}{spec}

\DeclareMathOperator{\dom}{dom}

\newcommand{\HH}{\pazocal{H}}
\newcommand{\T}{\mathcal{T}}
\newcommand{\D}{\mathcal{D}}
\newcommand{\K}{\mathcal{K}}
\newcommand{\B}{\mathcal{B}}
\newcommand{\FF}{\mathcal{F}}

\newcommand{\GG}{\pazocal{G}}
\newcommand{\MM}{\pazocal{M}}
\newcommand{\PP}{\pazocal{P}}
\newcommand{\CC}{\pazocal{C}}

\newcommand{\SEP}{\pazocal{S}}
\newcommand{\PPT}{\pazocal{P\!P\!T}}

\newcommand{\lsmatrix}{\left(\begin{smallmatrix}}
\newcommand{\rsmatrix}{\end{smallmatrix}\right)}

\stackMath
\newcommand\xxrightarrow[2][]{\mathrel{%
  \setbox2=\hbox{\stackon{\scriptstyle#1}{\scriptstyle#2}}%
  \stackunder[5pt]{%
    \xrightarrow{\makebox[\dimexpr\wd2\relax]{$\scriptstyle#2$}}%
  }{%
   \scriptstyle#1\,%
  }%
}}

\newcommand{\tends}[2]{\xxrightarrow[\raisebox{1.5pt}{\!$\scriptstyle #2$\!}]{\raisebox{-0.5pt}{$\scriptstyle \mathrm{#1}$}}}
\newcommand{\tendsn}[1]{\xxrightarrow[\! n\rightarrow \infty\!]{\mathrm{#1}}}

\newcommand{\ctends}[3]{\xxrightarrow[\raisebox{#3}{$\scriptstyle #2$}]{\raisebox{-0.7pt}{$\scriptstyle #1$}}}

\stackMath

\makeatletter
\newcommand*\rel@kern[1]{\kern#1\dimexpr\macc@kerna}
\newcommand*\widebar[1]{%
  \begingroup
  \def\mathaccent##1##2{%
    \rel@kern{0.8}%
    \overline{\rel@kern{-0.8}\macc@nucleus\rel@kern{0.2}}%
    \rel@kern{-0.2}%
  }%
  \macc@depth\@ne
  \let\math@bgroup\@empty \let\math@egroup\macc@set@skewchar
  \mathsurround\z@ \frozen@everymath{\mathgroup\macc@group\relax}%
  \macc@set@skewchar\relax
  \let\mathaccentV\macc@nested@a
  \macc@nested@a\relax111{#1}%
  \endgroup
}

\counterwithin*{equation}{part}
\counterwithin*{thm}{part}
\counterwithin*{figure}{part}

\tikzset{meter/.append style={draw, inner sep=10, rectangle, font=\vphantom{A}, minimum width=30, line width=.8, path picture={\draw[black] ([shift={(.1,.3)}]path picture bounding box.south west) to[bend left=50] ([shift={(-.1,.3)}]path picture bounding box.south east);\draw[black,-latex] ([shift={(0,.1)}]path picture bounding box.south) -- ([shift={(.3,-.1)}]path picture bounding box.north);}}}
\tikzset{roundnode/.append style={circle, draw=black, fill=gray!20, thick, minimum size=10mm}}
\tikzset{squarenode/.style={rectangle, draw=black, fill=none, thick, minimum size=10mm}}

\definecolor{Blues5seq1}{RGB}{239,243,255}
\definecolor{Blues5seq2}{RGB}{189,215,231}
\definecolor{Blues5seq3}{RGB}{107,174,214}
\definecolor{Blues5seq4}{RGB}{49,130,189}
\definecolor{Blues5seq5}{RGB}{8,81,156}

\definecolor{Greens5seq1}{RGB}{237,248,233}
\definecolor{Greens5seq2}{RGB}{186,228,179}
\definecolor{Greens5seq3}{RGB}{116,196,118}
\definecolor{Greens5seq4}{RGB}{49,163,84}
\definecolor{Greens5seq5}{RGB}{0,109,44}

\definecolor{Reds5seq1}{RGB}{254,229,217}
\definecolor{Reds5seq2}{RGB}{252,174,145}
\definecolor{Reds5seq3}{RGB}{251,106,74}
\definecolor{Reds5seq4}{RGB}{222,45,38}
\definecolor{Reds5seq5}{RGB}{165,15,21}

\makeatletter
\newcommand{\raisemath}[1]{\mathpalette{\raisem@th{#1}}}
\newcommand{\raisem@th}[3]{\raisebox{#1}{$#2#3$}}
\makeatother

\newcommand{\deff}[1]{\textbf{#1}}
\renewcommand{\id}{I}
\newcommand{\W}{\,\pazocal{W}}
\newcommand{\shs}{\hspace{1pt}}
\newcommand{\nonG}{\delta_R}
\newcommand{\As}{A_{[s]}}
\newcommand{\Hs}{H_{[s]}}

\let\tcb\relax

\begin{document}

\title{Attainability and lower semi-continuity of the relative entropy of entanglement, \\ and variations on the theme
}

\author{Ludovico Lami}
\email{ludovico.lami@gmail.com}
\affiliation{Institut f\"{u}r Theoretische Physik und IQST, Universit\"{a}t Ulm, Albert-Einstein-Allee 11, D-89069 Ulm, Germany}
\affiliation{QuSoft, Science Park 123, 1098 XG Amsterdam, The Netherlands}
\affiliation{Korteweg--de Vries Institute for Mathematics, University of Amsterdam, Science Park 105-107, 1098 XG Amsterdam, The Netherlands}
\affiliation{Institute for Theoretical Physics, University of Amsterdam, Science Park 904, 1098 XH Amsterdam, The Netherlands}

\author{Maksim E. Shirokov}
\email{msh@mi-ras.ru}
\affiliation{Steklov Mathematical Institute, Moscow, Russia}


\begin{abstract}
The relative entropy of entanglement $E_R$ is defined as the distance of a multi-partite quantum state from the set of separable states as measured by the quantum relative entropy. We show that this optimisation is \emph{always} achieved, i.e.\ any state admits a closest separable state, even in infinite dimensions; also, $E_R$ is everywhere lower semi-continuous. We use this to derive a dual variational expression for $E_R$ in terms of an external supremum instead of infimum. These results, which seem to have gone unnoticed so far, hold not only for the relative entropy of entanglement and its multi-partite generalisations, but also for many other similar resource quantifiers, such as the relative entropy of non-Gaussianity, of non-classicality, of Wigner negativity --- more generally, all relative entropy distances from the sets of states with non-negative $\lambda$-quasi-probability distribution. The crucial hypothesis underpinning all these applications is the weak*-closedness of the cone generated by free states, and for this reason the techniques we develop involve a bouquet of classical results from functional analysis. We complement our analysis by giving explicit and asymptotically tight continuity estimates for $E_R$ and closely related quantities in the presence of an energy constraint.
\end{abstract}

\maketitle

\tableofcontents

\section{Introduction}

In its early days, almost a century ago~\cite{Planck1901, Heisenberg1925, Born-Jordan, Schroedinger1926, VONNEUMANN}, quantum mechanics was mostly regarded as a bizarre physical theory whose exotic mathematics was needed to explain the behaviour of atomic spectra. As decades passed and physicists grew accustomed to the strangeness of the quantum world, they became interested not only in the question of \emph{what we can do for it}, that is, of how to explain or interpret it, but also in the more operationally-oriented question of what \emph{it can do for us}, namely, of what feats can be achieved by exploiting genuinely quantum effects. Far from descending from a purely practically-oriented mindset, this attitude reflects a general belief, originally stemming from information theory~\cite{Shannon}, that exploring the ultimate operational potential of a resource tells us something about its nature.

The latest incarnation of this philosophy is the formalism of quantum resource theories~\cite{Bennett-RT, Coecke2016, RT-review}. In this very general framework, one identifies two sets of objects that can be accessed at will in an inexpensive fashion: a set of quantum states (`free states'), and a set of quantum operations (`free operations'). The main theoretical concern in this context is how to characterise in an operationally meaningful way the resource content of an arbitrary state $\rho$ that is not free. In accordance with the above philosophy, one typically considers two complementary tasks: evaluate the number of `golden units' of pure resource that can be drawn from $\rho$ by manipulating it via free operations, or --- conversely --- the number of golden units that are needed to prepare $\rho$ via free operations in the first place.

These two complementary ideas lead, in the asymptotic limit of many copies, to the notions of distillable resource and of resource cost~\cite{RT-review}. Albeit somewhat opposite to each other, these two approaches allow, under appropriate assumptions, to single out a particular resource quantifier as the unique function determining the rates at which resources can be inter-converted by means of free operations~\cite{Brandao-Gour}. This is the regularised version of the relative entropy of resource, defined for an arbitrary state $\rho$ by $D_\FF(\rho) \coloneqq \inf_{\sigma\in \FF} D(\rho\|\sigma)$, where $\FF$ is the set of free states, and $D(\cdot\|\cdot)$ is Umegaki's relative entropy~\cite{Umegaki1962, Hiai1991}.

Historically, the first embodiment of the relative entropy of resource arose in the context of entanglement quantification for states of a composite finite-dimensional quantum system~\cite{Vedral1997, Vedral1998, Horodecki2000, Donald1999, Donald2002}. In this case, the role of free states is played by the set of separable (a.k.a.\ un-entangled) states, defined as classical mixtures of product states (describing un-correlated subsystems)~\cite{Werner, Holevo2005}. The resulting entanglement monotone, called the relative entropy of entanglement, turned out to be a very successful entanglement measure. Its regularisation can be endowed with multiple operational interpretations, either in the context of entanglement manipulation~\cite{BrandaoPlenio1, BrandaoPlenio2, irreversibility} or in that of hypothesis testing~\cite{Brandao2010}. More generally, it plays a fundamental role in the analysis of composite quantum systems and quantum channels~\cite{KHATRI}.

For all these reasons, it is desirable to have a general method to calculate or estimate the relative entropy of resource. Since it is naturally defined as a minimisation, upper bounds are easily computed by making ansatzes as to what the closest free state may be. A systematic technique to construct lower bounds, instead, has been put forth by Berta, Fawzi, and Tomamichel~\cite[\S~5]{Berta2017}, who found a dual variational formula involving an external maximisation instead of a minimisation. Their proof, however, rests on an application of Sion's theorem, and we argue that a simple generalisation to infinite-dimensional systems does not work. This is a serious problem, as many resource theories of interest are infinite-dimensional, and one could even make the case that almost all fundamental quantum systems, i.e.\ the quantum fields that form the basis of our most successful theoretical models, are intrinsically infinite-dimensional.

On a different note, it is useful to observe that although the relative entropy is not a metric, we can intuitively interpret the quantity $D_\FF(\rho)$ as a distance of a state $\rho$ from the set $\FF$. The natural question that arises in connection with this interpretation concerns the existence of a nearest free state $\sigma$ to a given state $\rho$ --- in other words, we are asking whether the infimum in the definition of $D_\FF(\rho)$ is attainable. This problem has a simple positive solution in the finite-dimensional setting, provided that the set $\FF$ is closed, due to the trace norm compactness of the state space. In infinite dimensions, the lower semi-continuity of the relative entropy still allows to prove the existence of the nearest free state $\sigma$ if the set $\FF$ is trace norm compact, but this observation covers basically none of the physically interesting cases --- typically, $\FF$ is closed but not compact, and attainability of the infimum in the definition of $D_\FF$ is a non-trivial open question. The experience with many other quantifiers defined in a similar way, through infima (and suprema), seems to suggest that this question may have, in general, a negative answer.

In this paper we will show that this is however not the case. Namely, we leverage the special  properties of the quantum relative entropy to prove that for many practically important non-compact sets $\FF$ of free states the infimum in the definition of $D_\FF$ is in fact always achieved --- uniquely for faithful states, if $\FF$ is convex --- and that the resulting function $D_\FF$ is lower semi-continuous. Our first two main results establish a general condition on $\FF$ for these conclusions to hold (Theorem~\ref{achievable_relent_thm}) and an effective criterion to check whether such condition is met in almost all cases of practical interest (Theorem~\ref{w*_closed_condition_thm}). Exploiting Theorem~\ref{achievable_relent_thm}, for the physically interesting finite-entropy states we derive a dual variational formula for $D_\FF(\rho)$ involving an external supremum instead of an infimum (Theorem~\ref{variational_thm}). This formula, which is our final main result, generalises to infinite-dimensional resource theories \tcb{the one} discussed in~\cite{Berta2017} for the finite-dimensional case, and can be used to generate lower bounds to $D_\FF$ in a systematic way, thus addressing the pressing problem discussed before.

We also study the continuity of $D_\FF$ (Proposition~\ref{RE-LS-c+_1}) and how the closest free state to $\rho$ depends on $\rho$ itself (Proposition~\ref{RE-LS-c+_2}). The key assumption on $\FF$ underpinning all these results is the closedness of the cone composed of all non-negative multiples of states in $\FF$ with respect to the weak*-topology --- the topology induced on the Banach space of trace class operators acting on a certain Hilbert space by its pre-dual, the space of compact operators. 

On the technical side, the proofs of our general results, and especially of Theorem~\ref{w*_closed_condition_thm}, constitute a systematic and somewhat gratifying application of many of the cardinal results of functional analysis --- inter alia, the uniform boundedness principle, the Banach--Alaoglu theorem, and the Krein--\v{S}mulian theorem --- to the framework of quantum resource theories. Although some functional analytic tools have been applied to the study of entanglement in quantum field theories before~\cite{HOLLANDS}, to the extent of our knowledge it is the first time that this is done in such a systematic way, further extending these ideas to general quantum resources.
For the proof of Theorem~\ref{variational_thm}, which is technically rather involved, we employ an original strategy that gives as by-products generalisations of Petz's variational formulae~\cite{Petz1988} to the case of unfaithful states and of Lieb's three-matrix inequality~\cite[Theorem~7]{lieb73c} to infinite dimensions.

The main advantage of our approach is its universality, which we demonstrate by applying it to a wide range of examples. First and foremost, we establish that the relative entropy of entanglement~\cite{Vedral1997, Vedral1998} is (a)~always achieved, (b)~lower semi-continuous, and (c)~that it can be expressed as a dual maximisation for finite-entropy states (Corollary~\ref{relent_entanglement_cor}). The same conclusions~(a),~(b), and~(c) hold true for the multi-partite generalisations of the relative entropy of entanglement, where the role of $\FF$ is played by states that are separable according to a prescribed set of partitions (Corollary~\ref{relent_alpha_entanglement_cor}), and for the relative entropy distance to the set of states with a positive partial transpose (Corollary~\ref{relent_NPT_entanglement_cor}). Moving on to the continuous variable setting, we are able to prove~(a),~(b), and~(c) for all quantifiers of the form $D_\FF$, where $\FF$ is composed by those multi-mode states with non-negative $\lambda$-quasi-probability distribution, $\lambda\in [-1,1]$ being a fixed parameter (Corollary~\ref{relent_lambda_negativity_cor}). Of special physical interest are the case $\lambda=+1$, in which case $\FF$ comprises all so-called `classical' states, i.e.\ all convex mixtures of coherent states, as well as the case $\lambda=0$, in which case $\FF$ includes the states with non-negative Wigner function. The last application of our results is to the resource theory of non-Gaussianity. In this case, $\FF$ is taken to be the non-convex set of Gaussian states, and $D_\FF$ is known to have a closed-form expression --- namely, it can be computed for a state $\rho$ as the difference $S(\rho_\G)-S(\rho)$, where $S$ is the von Neumann entropy and $\rho_\G$ is the Gaussian state with the same first and second moments as those of $\rho$. Although this function is known to be generally discontinuous in $\rho$~\cite{Kuroiwa2021}, our results imply that it is at least lower semi-continuous (Corollary~\ref{relent_non_Gaussianity_cor}).

For the special cases of the relative entropy of entanglement, its multi-partite generalisations, the Rains bound, and regularisations thereof, we complement our findings by establishing quantitative continuity estimates that are valid under appropriate energy constraints (Propositions~\ref{CB} and~\ref{nREE-CB}). Our bounds, which we prove by further refining the techniques in Ref.~\cite{Shirokov-AFW-2,Shirokov-AFW-3}, turn out to be asymptotically tight in many physically interesting cases, and imply that the aforementioned quantifiers, whose importance for the study of entanglement theory can hardly be overestimated, are uniformly continuous on energy-bounded sets of states.

The rest of the paper is organised as follows. Section~\ref{notation_sec} contains basic notation and definitions, as well as a brief introduction to some functional analytic techniques that play an essential role in this article. In Section~\ref{main_results_sec} we state and prove our main results in a completely general form (Theorems~\ref{achievable_relent_thm},~\ref{w*_closed_condition_thm}, and~\ref{variational_thm}). These tools are then applied in Section~\ref{applications_sec} to several concrete examples of the relative entropy of resource, starting from the bipartite relative entropy of entanglement and ending with the relative entropy of non-Gaussianity. Section~\ref{tight_uniform_sec} is devoted to a quantitative continuity analysis of the relative entropy of resource under energy constraints.

\section{Notation} \label{notation_sec}

\subsection{Topologies for quantum systems} \label{topologies_subsec}

In what follows, $\HH$ will denote an arbitrary (not necessarily finite-dimensional) separable Hilbert space.\footnote{A Hilbert space, or, more generally, a Banach space, is called separable if it admits a countable dense subset.} The Banach space of trace class operators on $\HH$, endowed with the trace norm $\|T\|_1\coloneqq \Tr \sqrt{T^\dag T}$, will be denoted with $\T(\HH)$. It can be thought of as the dual of the Banach space $\K(\HH)$ of compact operators on $\HH$ equipped with the operator norm, in formula $\T(\HH) = \K(\HH)^*$. In turn, the dual of $\T(\HH)$ --- and hence the bi-dual of $\K(\HH)$ --- is the Banach space of bounded operators on $\HH$, denoted with $\B(\HH)=\T(\HH)^* = \K(\HH)^{**}$~\cite[Chapter~VI]{REED}. Importantly, from the separability assumption for $\HH$ it follows that both $\K(\HH)$ and $\T(\HH)$ --- but \emph{not} $\B(\HH)$ --- are separable as Banach spaces.\footnote{The separability of $\T(\HH)$ is easy to establish by hand: indeed, the countable set of finite linear combinations of rank-one operators of the form $\ketbra{v}$, where $\ket{v}\in \HH$ has a finite expansion with rational coefficients in a fixed orthonormal basis of $\HH$ is easily seen to be trace norm dense in $\T(\HH)$~\cite{separability-trace-class}. It follows either by an analogous construction or by general arguments~\cite[Theorem~III.7]{REED} (see also~\cite{separability-compact}) that $\K(\HH)$ is also separable. On the contrary, it turns out that $\B(\HH)$ is never separable when $\HH$ is infinite dimensional~\cite{inseparability-bounded}.}

The cone of positive semi-definite trace class operators is defined by $\T_+(\HH) \coloneqq \left\{ X\in \T(\HH):\, \braket{\psi|X|\psi}\geq 0\ \forall\, \ket{\psi}\in \HH\right\}$. Quantum states on $\HH$ are represented by density operators, i.e.\ positive semi-definite trace class operators with unit trace; they form the set $\D(\HH)\coloneqq \left\{ X\in \T_+(\HH):\, \Tr X=1\right\}$. States $\rho\in \D(\HH)$ for which $\rho>0$ are said to be \deff{faithful}. For an arbitrary set $S\subseteq \T(\HH)$, we will denote the cone it generates with $\cone(S)\coloneqq \left\{\lambda X:\, \lambda\in [0,\infty),\, X\in S\right\}$. For example, we have that $\T_+(\HH) = \cone\left(\D(\HH)\right)$.

We will consider essentially two topologies on $\T(\HH)$. The first one is induced by the native norm of $\T(\HH)$: a sequence of operators $(T_n)_{n\in \N}$ in $\T(\HH)$ is said to converge with respect to the \deff{trace norm topology} to some $T\in \T(\HH)$, denoted $T_n\tends{tn}{n\to\infty} T$, if $\left\|T_n - T\right\|_1 \tends{}{n\to\infty} 0$. The second one is the \deff{weak*-topology} induced on $\T(\HH)$ by its pre-dual $\K(\HH)$, which can be defined as the coarsest topology making all functions of the form $\T(\HH)\ni T\mapsto \Tr TK$ continuous, where $K\in \K(\HH)$ is an arbitrary compact operator. Convergence of a sequence $(T_n)_{n\in \N}$ to $T\in \T(\HH)$ with respect to the weak*-topology, denoted $T_n \tends{w*}{n\to\infty} T$, is therefore equivalent to the condition that $\Tr T_n K\tends{}{n\to\infty} \Tr TK$ for all $K\in \K(\HH)$.\footnote{Defining convergence for sequences only does not suffice in this case, because the weak*-topology is not `metrisable' (see the discussion at the end of this section). To investigate the true nature of the weak*-topology one should instead consider general nets. However, we will see that this can be avoided in most cases of practical interest.} 
It is not difficult to see that any sequence of operators converging in the trace norm topology is also convergent (to the same limit) with respect to the weak*-topology. This is usually expressed by saying that the weak*-topology is coarser than the trace norm topology. An immediate consequence is that weak*-closed sets are also closed with respect to the trace norm topology.

The fact that these two topologies are genuinely different in infinite dimensions can be illustrated by showing that the two associated notions of convergence are different. For example, in a Hilbert space with orthonormal basis $\{\ket{n}\}_{n\in \N}$, the sequence of pure states $(\ketbra{n})_{n\in \N}$ does not converge at all with respect to the trace norm topology --- it is not even a Cauchy sequence --- but it tends to $0$ with respect to the weak*-topology, i.e.\ $\ketbra{n} \tends{w*}{n\to\infty} 0$. Intuitively, the weak*-topology treats any component that `escapes to infinity' as converging to $0$, while the trace norm topology takes it into account nonetheless.

Curiously, these two topologies agree on the very special set of quantum states. Namely, given a sequence $(\rho_n)_{n\in \N}$ of density operators $\rho_n\in \D(\HH)$ and another state $\rho\in \D(\HH)$, we have that $\rho_n \tendsn{tn} \rho$ if and only if $\rho_n \tendsn{w*} \rho$~\cite[Lemma~4.3]{Davies1969}. In order for this surprising conclusion to hold, it is crucial that both $\rho_n$ \emph{and} $\rho$ are \emph{normalised} density operators. In the above example, the weak*-limit was $0$, hence not a normalised density operator.

The cone of positive semi-definite trace class operators is well known to be closed with respect to the trace norm topology. It is an easy yet instructive exercise to verify that it is also weak*-closed, a fact that we will employ multiple times throughout the paper. To see this, it suffices to remember that $X\in \T(\HH)$ satisfies $X\geq 0$ if and only if $\Tr X\psi = \braket{\psi|X|\psi}\geq 0$ for all $\ket{\psi}\in \HH$. The claim then follows because each rank-one projector $\psi$ is a compact operator.

Unlike the trace norm topology, which is \emph{induced} by a metric (in the sense that a metric --- namely, the norm --- determines the convergence of all nets), the weak*-topology is \emph{not} `metrisable', i.e.\ it is not induced by \emph{any} metric~\cite[Proposition~2.6.12]{MEGGINSON}. The reader could wonder why to introduce the complicated weak*-topology alongside the more intuitive trace norm one. The fundamental reason to do so is that a general result in Banach space theory, the Banach--Alaoglu theorem~\cite[Theorem~2.6.18]{MEGGINSON}, guarantees that \emph{the dual unit ball is always compact in the weak*-topology.} In the present context, this implies that the unit ball 
\bb
B_1\coloneqq \{X\in \T(\HH):\, \|X\|_1\leq 1\}
\label{B_1}
\ee
of $\T(\HH)$ is weak*-compact. The acute reader will correctly suspect that the achievability results in our work rests crucially on this compactness property.

\begin{rem} \label{bounded_w*_metrisable_rem}
An immediate consequence of the Banach--Alaoglu theorem is that, provided that $\HH$ is separable (and hence both $\K(\HH)$ and $\T(\HH)$ are), the weak*-topology, albeit globally non-metrisable, is indeed metrisable on norm-bounded subsets~\cite[Corollary~2.6.20]{MEGGINSON}. Since weak*-compact sets of $\T(\HH)$ are always norm-bounded~\cite[Corollary~2.6.9]{MEGGINSON}, and compactness and sequential compactness are equivalent on metrisable spaces, we deduce the following handy fact: \emph{if $\HH$ is separable then any weak*-compact subset of $\T(\HH)$ is also sequentially weak*-compact}.
\end{rem}

\begin{rem} \label{vN_weak*_rem}
Several notions of weak topologies are commonly employed in the von Neumann algebra approach to quantum information. The purpose of this remark is to clarify some possible issues deriving from confusion in the terminology. Given a von Neumann algebra $\MM$ and its pre-dual $\MM_*$ spanned by all positive normal functionals on $\MM$ (also called normal states), the weak* topology on $\MM_*$ in the von Neumann algebra sense, or vN-weak* topology for short, is induced by the semi-norms $\MM_*\ni \varphi \mapsto |\varphi(x)|$, indexed by $x\in \MM$.
How does this topology compare to the weak*-topology we employ? To make a comparison, we look at the simplest case where $\MM=\B(\HH)$ is the von Neumann algebra of all bounded operators on some Hilbert space $\HH$, so that $\MM_*\simeq \T(\HH)$ is essentially the space of trace class operators on $\HH$. It is not difficult to see that the vN-weak* topology on $\T(\HH)$ is the coarsest topology that makes all functions of the form $\T(\HH)\ni T\mapsto \Tr TX$ continuous, where $X\in \B(\HH)$ is an arbitrary bounded operator. This is clearly a stronger topology than our weak*-topology, whose definition, despite the apparent similarity, only requires $X$ to vary over all \emph{compact} operators. This is a key difference: the weak*-converging sequence $\ketbra{n} \tends{w*}{n\to\infty} 0$ we looked at before does not converge in the vN-weak* topology, as can be seen swiftly by picking $X=\id$.
\end{rem}

\subsection{Relative entropy of resource}

The \deff{von Neumann entropy}, or simply the entropy, of an arbitrary positive semi-definite trace class operator $X\in \T_+(\HH)$ whose spectral decomposition reads $X=\sum_i x_i \ketbra{e_i}\geq 0$ is defined by
\bb
S(X) \coloneqq - \Tr \left[ X\ln X\right] = \sum_i (-x_i \ln x_i)\, ,
\label{von_Neumann}
\ee
where the sum on the right-hand side is well defined (although possibly infinite) because $x\mapsto -x\ln x$ is non-negative for all $x\in [0,1]$, and all but a finite number of eigenvalues of $X$ belong to the interval $[0,1]$, due to the fact that $X$ is of trace class.

Given two positive semi-definite trace class operators $X,Y\in \T(\HH)$, $X,Y\geq 0$, with spectral decomposition $X = \sum_i x_i \ketbra{e_i}$ and $Y = \sum_j y_j \ketbra{f_j}$, $x_i,y_j>0$, one defines their \deff{relative entropy} by~\cite{Umegaki1962, Lindblad1973}
\bb
D(X\|Y) \coloneqq&\ \Tr \left[ X \left(\ln X - \ln Y\right) + Y - X \right] \\
\coloneqq&\ \sum_{i,j} \left|\braket{e_i|f_j}\right|^2 \left( x_i \ln x_i - x_i \ln y_j + y_j-x_i \right) .
\label{Umegaki}
\ee
This is a function taking on values in the set of extended real numbers $\R\cup\{ +\infty\}$, and we adopt the convention according to which $D(0\|0)=0$ and $D(X\|0)=+\infty$ if $X\neq 0$. The expression on the first line of~\eqref{Umegaki} is to be interpreted as specified in the second. Thanks to the convexity of the function $x\mapsto x\ln x$, each term of the series is non-negative; hence, the value of the series is well defined, albeit possibly infinite. Note that a necessary condition for $D(X\|Y)<\infty$ to hold is that the support of the first argument is contained into the support of the second, in formula $\supp X\subseteq \supp Y$.

On a different note, it is useful to observe that whenever $Y\leq \id$ (an assumption that can be made without loss of generality) the expression~\eqref{Umegaki} can be recast as
\bb
D(X\|Y) = \sum_i \left( x_i \ln x_i + x_i \big\| \sqrt{-\ln Y} \ket{e_i}\big\|^2 \right) + \Tr[Y-X]\, ,
\label{Umegaki_grown-ups}
\ee
where we convene that $\left\| \sqrt{-\ln Y} \ket{e_i}\right\|^2=+\infty$ if $\ket{e_i}\notin \dom\left( \sqrt{-\ln Y} \right)$. If moreover $S(X)<\infty$, then~\eqref{Umegaki_grown-ups} can be further reduced to the more familiar expression
\bb
D(X\|Y) \coloneqq -S(X) - \Tr\left[ X\ln Y\right] + \Tr[Y-X]\, ,
\label{Umegaki_splitting}
\ee
which in this case is well defined because the only term that can possibly diverge on the right-hand side is the second, to be interpreted as the series $-\Tr\left[ X\ln Y\right] \coloneqq -\sum_{i,j} \left|\braket{e_i|f_j}\right|^2 x_i \ln y_j$.

A fundamental result establishes that the relative entropy is jointly convex, i.e.~\cite{lieb73a, lieb73b} (see also~\cite[Theorem~5.4]{PETZ-ENTROPY})
\bb
D\left( \sumno_i p_i X_i\Big\| \sumno_i p_i Y_i\right) \leq \sum_i p_i D\left(X_i\|Y_i\right)
\label{joint_convexity}
\ee
for all finite probability distributions $p_1,\ldots p_k$ (satisfying $p_i\geq 0$ for all $i$ and $\sum_i p_i=1$) and all collections of positive semi-definite trace class operators $X_i,Y_j\in \T_+(\HH)$. We now establish the special property of the quantum relative entropy which is key for our approach, namely, its weak* lower semi-continuity.

\begin{rem}
The absence of a similar property for other quantum divergences prevents the immediate extension of our results to the corresponding generalised divergences of resource. Indeed, remember from Remark~\ref{vN_weak*_rem} that the weak*-topology employed here is \emph{not} the one commonly encountered in the study of von Neumann algebras.
\end{rem}

\begin{lemma} \label{weak*_lsc_relent_lemma}
The relative entropy
\bb
\T_+(\HH)\times \T_+(\HH) \ni (X,Y)\mapsto D(X\|Y)\in \R\cup\{+\infty\}
\label{Umegaki_function}
\ee
is lower semi-continuous with respect to the product weak*-topology.
\end{lemma}

\begin{proof}
An elementary yet important fact that we shall use multiple times without further comments is the following: if $\{f_\alpha\}_{\alpha}$ is a (possibly infinite) family of lower semi-continuous real-valued functions $f_\alpha:\pazocal{X}\to \R$ on the topological space $\pazocal{X}$, then $F:\pazocal{X}\to \R$ defined by $F(x)\coloneqq \sup_{\alpha} f_\alpha(x)$, referred to as the point-wise supremum of the family $\{f_\alpha\}_{\alpha}$, is also lower semi-continuous. 

Now, a well-known result of Lindblad~\cite[Lemmata~3 and~4]{Lindblad1974} states that
\begin{equation}
D(X\|Y) = \sup_P D\left(PX P\big\| PY P\right) ,
\label{finite_dim_reduction_relent}
\end{equation}
where the supremum is over all finite-dimensional projectors $P$. By the above criterion, if we establish that $(X,Y)\mapsto D\left(PX P\big\| PY P\right)$ is lower semi-continuous for all $P$ then we are done. Write an arbitrary finite-dimensional projector $P$ as $P=\sum_{i=1}^N \ketbra{\psi_i}$. The map $X\mapsto PXP = \sum_{i,j=1}^N \braket{\psi_i|X|\psi_j} \ketbraa{\psi_i}{\psi_j}$, being a finite sum of weak*-continuous functions, is clearly continuous with respect to the weak*-topology on the input space. Since the output space is finite dimensional, all Hausdorff linear topologies are equivalent there. Thanks to the fact that $(X,Y)\mapsto D(X\|Y)$ is lower semi-continuous whenever $X,Y$ are finite-dimensional positive semi-definite matrices --- indeed, 
\tcb{due} to the operator monotonicity of the logarithm~\cite{BHATIA-MATRIX} 
\tcb{the above function} can be written as $D(X\|Y) = \sup_{c>0} \left\{-S(X) - \Tr X \ln(Y+c\id) + \Tr[Y-X] \right\}$, where $\id$ is the identity operator, $S$ is the von Neumann entropy~\eqref{von_Neumann}, \tcb{and the functions inside this latter supremum are continuous in $X,Y$} --- we deduce immediately that $\T_+(\HH)\times \T_+(\HH) \ni (X,Y)\mapsto D\left(PXP\big\| PYP\right) $ is lower semi-continuous with respect to the product weak*-topology. The same is then true of the fully-fledged relative entropy~\eqref{Umegaki_function}, thanks to the representation~\eqref{finite_dim_reduction_relent}.
\end{proof}

The framework of quantum resource theories~\cite{RT-review} models all those situations where operational or experimental constraints limit the set of states of a quantum system one can access in practice. Let a quantum system with Hilbert space $\HH$ be given, and denote with $\FF\subseteq \D(\HH)$ the set of states that \emph{are} accessible at no cost, hereafter called \deff{free states}. To quantify the resource cost of an arbitrary state $\rho\in \D(\HH)$, one employs functions known as resource quantifiers. One of the simplest and most important such functions is the so-called \deff{relative entropy of resource}, defined for $\rho\in \D(\HH)$ by
\bb
D_\FF(\rho) \coloneqq \inf_{\sigma\in \FF} D(\rho\|\sigma)\, .
\label{relent_F}
\ee
Several specific examples of this construction are explored in detail in the forthcoming Section~\ref{applications_sec}. Importantly, the above function $D_\FF$ is a resource monotone, i.e.\ it is non-increasing under any (completely) positive trace preserving map that sends $\FF$ into itself. As a matter of fact, it is perhaps the most important resource monotone, as its regularisation governs the inter-conversion rates between states under asymptotically resource non-generating operations~\cite{Brandao-Gour}.

As explained in the Introduction, we are concerned here with the general properties of the relative entropy of resource~\eqref{relent_F}, for either finite- or infinite-dimensional systems. For instance, a pressing problem is how to compute $D_\FF$ in practice. While~\eqref{relent_F} allows us to find upper bounds rather easily by simply making ansatzes for $\sigma$, it is not at all obvious how to calculate lower bounds systematically. In the finite-dimensional case, Berta, Fawzi, and Tomamichel~\cite{Berta2017} managed to tackle this problem by employing the variational expression for the relative entropy found by Petz~\cite{Petz1988}; plugging it into~\eqref{relent_F} and using Sion's theorem to exchange the infimum and supremum, one obtains a formula for $D_\FF$ that involves an external supremum instead of an infimum (cf.~\eqref{relent_F}). Such a formula is an ideal tool to compute lower bounds on $D_\FF$ in a systematic way, as well as to prove general properties of $D_\FF$. 
And indeed, one of our main results provides an extension of it to infinite-dimensional resource theories (Theorem~\ref{variational_thm}).

However, to arrive there we need to start from simpler questions. A particularly immediate one is whether the infimum in~\eqref{relent_F} is always achievable. In other words, does there always exist a closest free state? Is it unique? One could also wonder whether the resulting function $D_\FF$ is lower semi-continuous, which is one of the strongest forms of regularity we can hope for in general, as in infinite dimensions useful (i.e.\ extensive, or, more formally, additive) resource monotones are typically everywhere discontinuous --- so, for instance, is the von Neumann entropy~\cite{Wehrl}.

The compactness (and convexity) of $\FF$ with respect to the trace norm topology is a sufficient condition to ensure both the existence of a closest free state to any given state and also the lower semi-continuity of $D_\FF$. If $\dim\HH<\infty$, then this amounts to requiring that $\FF$ is closed: under this hypothesis, which is typically met in many important cases, in the finite-dimensional case we can replace the infimum in~\eqref{relent_F} with a minimum. The problem is that the above compactness assumption is almost never met for infinite-dimensional resource theories: indeed, density operators form themselves a non-compact set! Just to name a few examples of interesting $\FF$, separable states~\cite{Werner, Holevo2005} or states with positive partial transpose~\cite{PeresPPT} in a bipartite system, classical states~\cite{Bach1986, Yadin2018, nonclassicality}, states with a non-negative Wigner function~\cite{Hudson1974, Hudson-thm-multimode, Broecker1995}, or even Gaussian states~\cite{BUCCO} in continuous variable multi-mode systems all give rise to non-compact sets. For these cases, prior to our paper none of the above properties of $D_\FF$ was known~\cite{Eisert2002}. In Section~\ref{applications_sec} we will see how to apply our main results (Theorems~\ref{achievable_relent_thm},~\ref{w*_closed_condition_thm}, and~\ref{variational_thm}) to establish the achievability and lower semi-continuity of $D_\FF$ in all of these cases and in even greater generality, and to derive dual variational expressions for it.

\section{Main results} \label{main_results_sec}

\subsection{The statements}

Throughout this section we will present the statements of our main results. All proofs can be found in the forthcoming Section~\ref{proofs_subsec}.

Our first main result establishes a sufficient condition that makes the relative entropy of resource~\eqref{relent_F} achieved, meaning that the infimum in~\eqref{relent_F} is in fact a minimum, and lower semi-continuous as a function of the state. As we will see in the next section, this condition covers virtually all quantum resource theories of practical interest.

\begin{thm} \label{achievable_relent_thm}
Let $\HH$ be a (possibly infinite-dimensional) separable Hilbert space, and let $\FF\subseteq \D(\HH)$ be a (not necessarily convex) set of density operators on $\HH$. If $\cone(\FF)= \left\{\lambda \sigma:\, \lambda\in [0,\infty),\, \sigma\in \FF\right\}$ is weak*-closed, then the relative entropy distance from $\FF$, defined by~\eqref{relent_F}, is:
\begin{enumerate}[(a)]
    \item always achieved, meaning that for all $\rho$ there exists $\sigma\in \FF$ such that $D_\FF(\rho) = D(\rho\|\sigma)$; and
    \item lower semi-continuous in $\rho$ with respect to the trace norm topology, i.e.\ such that $\rho_n \tends{tn}{n\to\infty} \rho$ implies that
\begin{equation}
    D_\FF(\rho) \leq \liminf_{n\to\infty} D_\FF(\rho_n)\, ,
\end{equation}
    for all sequences of density operators $(\rho_n)_{n\in \N}$.
\end{enumerate}
Moreover, if $\rho\in \D(\HH)$ is such that $D_\FF(\rho)<\infty$ then the non-empty set $\Sigma_\FF(\rho)\coloneqq \left\{ \sigma\in \FF:\, D(\rho\|\sigma)=D_\FF(\rho)\right\}$ of minimisers:
\begin{itemize}
\item[(c)] is trace norm compact, and furthermore convex if $\FF$ itself is convex; and
\item[(d)] contains a unique state if $\FF$ is convex and $\rho$ is faithful (i.e.\ $\rho>0$).
\end{itemize}
\end{thm}

\begin{rem} The faithfulness condition in claim~(d) of Theorem~\ref{achievable_relent_thm} cannot be omitted even in the case of finite-dimensional space $\HH$. This is confirmed by Example~\ref{non-un} in Section~\ref{entanglement_subsec}.
\end{rem}

The main difficulty in applying Theorem~\ref{achievable_relent_thm} lies in verifying the weak*-closedness of the cone generated by the set of free states $\FF$. In fact, although it is often the case that the set $\FF$ itself is trace norm closed, since the weak*-topology is coarser than the trace norm topology this fact cannot be used to deduce the sought property of $\cone(\FF)$.
To make our life easier, we will now equip ourselves with a technical tool that turns out to cover almost all interesting cases. The following result 
ought to be thought of as instrumental to the application of Theorem~\ref{achievable_relent_thm}.

\begin{thm} \label{w*_closed_condition_thm}
Let $\HH$ be a (possibly infinite-dimensional) separable Hilbert space, and let $(M_n)_{n\in \N}$ be a sequence of compact operators on $\HH$ that converges to the identity in the strong operator topology, i.e.\ such that $\left\|M_n\ket{\psi} - \ket{\psi}\right\|\tends{}{n\to\infty} 0$ for all $\ket{\psi}\in \HH$. Define $\MM_n(\cdot)\coloneqq M_n(\cdot) M_n^\dag$. Consider a set of states $\FF\subseteq \D(\HH)$, and assume that:
\begin{enumerate}[(i)]
    \item $\FF$ is convex;
    \item $\FF$ is trace norm closed;
    \item $\MM_n$ preserves free states (up to normalisation), i.e.\ $\MM_n\left( \cone(\FF)\right)\subseteq \cone(\FF)$ for all $n\in \N$.
\end{enumerate}
Then $\cone(\FF)$ is weak*-closed, and in particular claims~(a)--(b) of Theorem~\ref{achievable_relent_thm} hold true.
\end{thm}

Remarkably, hypothesis~(iii) in Theorem~\ref{w*_closed_condition_thm} cannot be removed. To see why, and to better illustrate the nature of the assumption of weak*-closedness in Theorem~\ref{achievable_relent_thm}, we now present a simple yet instructive example.

\begin{ex}
Let $\HH$ be a separable Hilbert space with orthonormal basis $\{\ket{n}\}_{n\in \N}$. Define the unbounded, densely defined operator $H\coloneqq \sum_{n=0}^\infty n \ketbra{n}$, which can be interpreted physically as a Hamiltonian (cf.\ Section~\ref{energy_constraints_subsec}). For some $E\geq 0$, construct the set of states
\bb
\FF_{H,E}\coloneqq \left\{ \rho_{AB}\in \D(\HH\otimes \HH):\, \Tr \left[\rho_{AB}\, H_A\otimes \id_B\right]\leq E\right\}
\ee
on the Hilbert space $\HH\otimes \HH$, where it is understood that $\Tr \left[\rho_{AB} H_A\otimes \id_B\right] \coloneqq \sup_{N\in\N} \sum_{k=0}^N k\braket{k|\rho_A|k}$, and $\rho_A = \Tr_B \rho_{AB}$.

Then $\FF_{H,E}$ is (i)~convex, (ii)~trace norm closed, but $\cone(\FF_{H,E})$ is \emph{not} weak*-closed. Claim~(i) is obvious, while~(ii) follows from the trace norm lower semi-continuity of $\Tr \rho H$, which is a point-wise supremum of continuous functions. To verify that $\cone(\FF_{H,E})$ is \emph{not} weak*-closed, consider the sequence of states $\big(\sigma^{(n)}_{AB}\big)_{n\in \N}$ with $\sigma_{AB}^{(n)} \coloneqq \frac12 \ketbra{0}_A\otimes \ketbra{n}_B + \frac12 \ketbra{2}_A\otimes \ketbra{0}_B$. Note that $\sigma_{AB}^{(n)}\in \FF_{H,1}$ for all $n$. However, $\sigma_{AB}^{(n)}\tends{w*}{n\to\infty} \frac12 \ketbra{2}_A\otimes \ketbra{0}_B$, and the right-hand side does not belong to $\cone(\FF_{H,1})=\left\{ X_{AB}\in \T_+(\HH\otimes \HH):\, \Tr \left[X_{AB}\, H_A\otimes \id_B\right]\leq \Tr X_{AB}\right\}$.
\end{ex}

Our final main result is a general variational expression for the relative entropy of resource $D_\FF$ that is dual to~\eqref{relent_F}, in that it features an external maximisation instead of a minimisation. Its proof leverages in a key way Theorem~\ref{achievable_relent_thm}.

\begin{thm} \label{variational_thm}
Let $\HH$ be a (possibly infinite-dimensional) separable Hilbert space, and let $\FF\subseteq \D(\HH)$ be a convex set of states such that $\cone(\FF)$ is weak*-closed. For any state $\rho\in \D(\HH)$ with finite entropy $S(\rho)<\infty$, it holds that
\bb
D_\FF(\rho) = \sup_{X = X^\dag\in \B(\HH)} \left\{ \Tr \rho X - \sup_{\sigma\in \FF} \ln \Tr e^{\ln \sigma + X} \right\} .
\label{variational}
\ee
\end{thm}

\begin{note}
If the state $\sigma$ appearing in~\eqref{variational} fails to be faithful the interpretation of its logarithm may pose some problems. In this case, we convene that 
\bb
e^{\ln \sigma + X} \coloneqq P^\dag e^{\ln \left(P \sigma P^\dag\right) + P X P^\dag} P\, ,
\label{interpretation}
\ee
where $P:\HH\to \supp(\sigma)$ is the orthogonal projector onto the support of $\sigma$.
\end{note}

The variational approach to the study of quantum relative entropy has a long and illustrious tradition~\cite{Kosaki1986, Petz1988, PETZ-ENTROPY, Berta2017}. As for the relative entropy of resource, variational expressions of the above dual kind have proved to be very useful in establishing general properties of $D_\FF$ and related quantifiers, for instance super-additivity~\cite{Berta2017, nonclassicality}. Furthermore, as mentioned before they can be of immense help computationally, as they provide a valuable tool to generate lower bounds on $D_\FF$ systematically.

As for the proof technique, the one in~\cite[Section~V.A]{Berta2017} rests on Sion's theorem and does not carry over to infinite dimensions, because the compactness hypothesis typically breaks down. In Section~\ref{proofs_subsec} we show how to overcome this difficulty by making use of Theorem~\ref{achievable_relent_thm}; as a by-product of our derivation, exploiting some new results on multivariate trace inequalities~\cite{Sutter2017, Junge2021, Hollands2021} we extend the celebrated Lieb's three-matrix inequality~\cite{lieb73c} to the infinite-dimensional case (Appendix~\ref{Lieb_3_app}).

\subsection{Proofs} \label{proofs_subsec}

\begin{proof}[Proof of Theorem~\ref{achievable_relent_thm}]
Consider the set
\bb
\widetilde{\FF}\coloneqq \left\{\lambda \sigma:\, \lambda\in [0,1],\, \sigma\in \FF\right\} = \cone(\FF)\cap B_1\, ,
\label{F_tilde}
\ee
where $B_1$, defined by~\eqref{B_1}, is the unit ball of the trace norm. Note that $B_1$ is weak*-compact thanks to the Banach--Alaoglu theorem~\cite[Theorem~2.6.18]{MEGGINSON}, while $\cone(\FF)$ is weak*-closed by hypothesis. Being the intersection of a weak*-compact and a weak*-closed set, $\widetilde{\FF}$ is also weak*-compact.

Now, let us prove claim~(a) for a fixed state $\rho\in \D(\HH)$. Clearly, we can assume without loss of generality that $D_\FF(\rho)<\infty$, otherwise any state $\sigma\in \FF$ achieves~\eqref{relent_F}. For all $\lambda\in [0,1]$ and $\rho,\sigma\in \D(\HH)$, we have that
\bb
D(\rho\|\lambda\sigma) = D(\rho\|\sigma) + \lambda - 1 - \ln \lambda \geq D(\rho\|\sigma)\, ,
\ee
thanks to the fact that $f(\lambda)\coloneqq \lambda - 1 - \ln \lambda \geq 0$ for all $\lambda\in [0,1]$, with the convention that $f(0)=+\infty$. We can then write
\bb
D_\FF(\rho) = \inf_{\substack{\sigma\in \FF, \\ \lambda\in [0,1]}} \left\{ D(\rho\|\sigma) + \lambda - 1 - \ln \lambda \right\} = \inf_{\substack{\sigma\in \FF, \\ \lambda\in [0,1]}} D(\rho\|\lambda\sigma) = \inf_{\eta\in \widetilde{\FF}} D(\rho\|\eta)\, .
\label{D_F_eta}
\ee
Since the function $\eta\mapsto D(\rho\|\eta)$ is lower semi-continuous with respect to the weak*-topology, it achieves its minimum on the weak*-compact set $\widetilde{\FF}$. Let $\eta_0\in \widetilde{\FF}$ be such that $D_\FF(\rho) = D(\rho\|\eta_0)$. It must be that $\eta_0\neq 0$, otherwise $D_\FF(\rho) = +\infty$, contradicting our working assumptions. Due to the fact that $\eta_0\geq 0$, we deduce that in fact $\Tr \eta_0>0$. Then, $\sigma_0\coloneqq \left(\Tr\eta_0\right)^{-1} \eta_0\in \FF$ satisfies that
\bb
D_\FF(\rho) &\leq D(\rho\|\sigma_0) = D(\rho\|\eta_0) + \ln \Tr\eta_0 -\Tr\eta_0+1 \leq D(\rho\|\eta_0) = D_\FF(\rho)\, ,
\label{throw_away_normalisation}
\ee
where the inequality follows from the fact that $\ln \Tr\eta_0-\Tr\eta_0+1\leq 0$ as $\Tr \eta_0\in (0,1]$. We infer immediately that in fact $\Tr\eta_0=1$, i.e.\ $\eta_0=\sigma_0\in \FF$.

We now move on to claim~(b). Let $(\rho_n)_{n\in \N}$ be a sequence such that $\rho_n \tends{tn}{n\to\infty} \rho$. Up to extracting a subsequence, we can assume without loss of generality that $\left(D_\FF(\rho_n)\right)_{n\in \N}$ converges (if $\lim_{n\to\infty} D_\FF(\rho_n) = +\infty$ there is nothing to prove). We now construct states $\sigma_n\in \FF$ such that $D_\FF(\rho_n) = D(\rho_n\|\sigma_n)$. Thanks to the fact that $\widetilde{\FF}$ is weak*-compact and hence sequentially weak*-compact by Remark~\ref{bounded_w*_metrisable_rem}, we can extract a weak*-converging subsequence $\sigma_{n_k}\tends{w*}{k\to\infty} \eta\in \widetilde{\FF}$. Note that
\bb
D_\FF(\rho) \leqt{(i)} D(\rho\|\eta) \leqt{(ii)} \liminf_{k\to\infty} D\left(\rho_{n_k}\big\|\sigma_{n_k}\right) = \liminf_{k\to\infty} D_\FF(\rho_{n_k}) \eqt{(iii)} \lim_{n\to\infty} D_\FF(\rho_n) \, ,
\ee
where (i)~holds due to~\eqref{D_F_eta}, (ii)~thanks to Lemma~\ref{weak*_lsc_relent_lemma}, and (iii)~because the sequence $\left(D_\FF(\rho_n)\right)_{n\in \N}$ converges.

We now set out to prove claim~(c). The convexity of $\Sigma_\FF(\rho)$ follows immediately from that of $\FF$, once we remember that the relative entropy is a jointly convex function and hence in particular convex in its second argument (cf.~\eqref{joint_convexity}). To prove the trace norm compactness, pick an arbitrary sequence $(\sigma_n)_{n\in\N}$ with $\sigma_n\in \Sigma_\FF(\rho)$ for all $n$. Since $\sigma_n\in \widetilde{\FF}$, where $\widetilde{\FF}$ defined by~\eqref{F_tilde} is weak*-compact and hence sequentially weak*-compact by Remark~\ref{bounded_w*_metrisable_rem}, we can extract a weak*-converging sub-sequence $(\sigma_{n_k})_{k\in \N}$, so that $\sigma_{n_k}\tends{w*}{k\to\infty}\eta_*\in \widetilde{\FF}$. Now,
\bb
D_\FF(\rho) &\eqt{(i)} \lim_{k\to\infty} D\left(\rho \| \sigma_{n_k}\right) \\
&\geqt{(ii)} D\left(\rho\|\eta_{*}\right) \\
&\eqt{(iii)} D\left(\rho\|\sigma_*\right) + \Tr \eta_* - 1 - \ln \Tr \eta_* \\
&\geqt{(iv)} D\left(\rho\|\sigma_*\right) ,
\label{chain_inequalities}
\ee
where~(i) is because $\sigma_n\in \Sigma_\FF(\rho)$, (ii)~holds thanks to the lower semi-continuity of the relative entropy (Lemma~\ref{weak*_lsc_relent_lemma}), and in~(iii)--(iv) we introduced the state $\sigma_*\coloneqq (\Tr \eta_*)^{-1} \eta_* \in \FF$ and proceeded as in~\eqref{throw_away_normalisation}. Note that~(ii) ensures that $\Tr \eta_*>0$, otherwise we would have that $\eta_*=0$ and hence $D(\rho\|\eta_*)=+\infty$. Since $\sigma_*\in \FF$, the above inequality implies that $\sigma_*\in \Sigma_\FF(\rho)$ and moreover $\Tr \eta^*=1$, so that in fact $\eta^*=\sigma^*\in \Sigma_\FF(\rho)$. To conclude, note that since weak* and trace norm topology coincide on the set of density operators (cf.\ the discussion in Section~\ref{topologies_subsec}), from $\sigma_{n_k}\tends{w*}{k\to\infty}\eta_*=\sigma_*$ we infer that $\sigma_{n_k}\tends{tn}{k\to\infty}\sigma_*$ with respect to the (much stronger) trace norm topology. Therefore, $\Sigma_\FF(\rho)$ is trace norm compact.

Finally, assume as in~(d) that $\FF\subseteq \D(\HH)$ is convex and that $\rho>0$ is such that $D_\FF(\rho)<\infty$. To show that the minimiser is unique, assume that there exist $\sigma_1,\sigma_2\in \FF$ such that $D(\rho\|\sigma_1) = D(\rho\|\sigma_2) = D_\FF(\rho)<\infty$, and let us show that $\sigma_1=\sigma_2$. Note that the states $\sigma_1$ and $\sigma_2$ are also faithful, because a necessary condition in order for $D(\rho\|\sigma)$ to be finite is that $\supp\rho\subseteq \supp\sigma$, where $\supp$ denotes the support, and in this case $\supp\rho=\HH$. Set $\sigma_0\coloneqq \frac12 (\sigma_1+\sigma_2)$. Then the convexity of the relative entropy implies that $D_\FF(\rho)\leq D(\rho\|\sigma_0) \leq \frac12 \left( D(\rho\|\sigma_1) + D(\rho\|\sigma_2)\right) = D_\FF(\rho)$, and hence in particular
\bb
D\left(\rho\,\Big\|\,\frac{\sigma_1+\sigma_2}{2}\right) = \frac12 \left( D(\rho\|\sigma_1) + D(\rho\|\sigma_2)\right) .
\label{joint_convexity_saturation}
\ee
Since the relative entropy is strictly concave in each of the entries provided that they are faithful states~\cite{Donald-further, Petz-old}, this can happen only if $\sigma_1=\sigma_2$. To be more explicit, by a result of Petz~\cite[p.~130]{Petz-old} we can infer from~\eqref{joint_convexity_saturation} that $\sigma_1^{it} = \sigma_2^{it}$ holds for all $t\in \R$. By applying the spectral theorem, it is easy to see that for given $\ket{\psi}\in \HH$ and $\omega\in \D(\HH)$ we have that $\left|\braket{\psi|\omega^{it}|\psi}\right|=1$ if and only if $\ket{\psi}$ is an eigenvector of $\omega$. Thus, each eigenvector of $\sigma_1$ must be an eigenvector of $\sigma_2$ as well, implying that $\sigma_1$ and $\sigma_2$ can be diagonalised simultaneously. Since for $p,q>0$ it holds that $p^{it}=q^{it}$ for all $t\in R$ if and only if $p=q$, it must be that $\sigma_1=\sigma_2$, as claimed.
\end{proof}

\begin{rem}
We note in passing that the conditions for equality in the joint convexity of the relative entropy (generalising~\eqref{joint_convexity_saturation}) have been studied thoroughly in the specialised literature. See e.g.~\cite[Theorem~8]{Jencova2010} or~\cite[Corollary~5.3]{Hiai2011}.
\end{rem}

\begin{proof}[Proof of Theorem~\ref{w*_closed_condition_thm}]
Let us break the argument into several elementary steps.
\begin{enumerate}[(I)]
\item We start by showing that
\bb
\cone(\FF) = \bigcap_{n\in \N} \MM_n^{-1}\left(\cone(\FF)\right) .
\label{intersection_identity}
\ee
\begin{enumerate}
\item[(I.a)] First, observe that since $\MM_n\left( \cone(\FF)\right)\subseteq \cone (\FF)$ we have that 
\bb
\cone(\FF) \subseteq \MM_n^{-1}\left( \MM_n\left( \cone(\FF)\right) \right) \subseteq \MM_n^{-1}\left( \cone(\FF)\right) .
\ee
This holds for all $n\in \N$, hence we infer that $\cone(\FF) \subseteq \bigcap_{n\in \N} \MM_n^{-1}\left( \cone(\FF)\right)$.

\item[(I.b)] To show the converse inclusion, take a trace class $X\in \T(\HH)$, $X\neq 0$, such that $X\in \MM_n^{-1}\left(\cone(\FF)\right)$ for all $n\in \N$.
An important observation to make now is that the sequence of operator norms $\left(\|M_n\|_\infty \right)_{n\in \N}$ is bounded, i.e.
\bb
\sup_{n\in \N} \|M_n\|_\infty \leq L < \infty\, .
\label{UBP_1}
\ee
This non-trivial fact follows from the uniform boundedness principle~\cite[Section~1.6.9]{MEGGINSON} combined with the observation that $\sup_{n\in \N} \left\|M_n \ket{\psi}\right\| < \infty$ for all fixed $\ket{\psi}\in \HH$ because $\left\|M_n \ket{\psi}\right\| \tends{}{n\to\infty} \left\|\ket{\psi}\right\|$.

An immediate consequence is that
\bb
\left\|\MM_n(X) - X\right\|_1 \tends{}{n\to\infty} 0\qquad \text{and}\qquad \Tr \MM_n(X) \tends{}{n\to\infty} \Tr X\, .
\label{convergence_MM_n}
\ee
To see why this is the case, start by observing that the second identity follows from the first upon taking the trace. To prove the first, then, pick an $\varepsilon>0$, and construct a finite-rank operator $X'$ such that $\|X'-X\|_1\leq \varepsilon$. Since $\left\|\MM_n(X')-X'\right\|_1\tends{}{n\to\infty} 0$ by strong operator convergence, noticing that
\bb
\left\|\MM_n(X) - X\right\|_1 &\leq \left\|\MM_n(X) - \MM_n(X')\right\|_1 + \left\|\MM_n(X') - X'\right\|_1 + \left\|X' - X\right\|_1 \\
&\leq \|M_n\|_\infty^2 \|X-X'\|_1 + \left\|\MM_n(X') - X'\right\|_1 + \left\|X' - X\right\|_1
\ee
we see that
\bb
\limsup_{n\to\infty} \left\|\MM_n(X) - X\right\|_1 \leq \left(L^2+1\right)\varepsilon
\ee
thanks to~\eqref{UBP_1}. Since $\varepsilon>0$ is arbitrary, it must be that $\left\|\MM_n(X) - X\right\|_1\tends{}{n\to\infty} 0$, as claimed. This proves~\eqref{convergence_MM_n}. Now, observe that $\MM_n(X)\in \cone(\FF)$ is positive semi-definite for all $n$; since the cone of positive semi-definite operators is trace norm closed, we deduce that also $X\geq 0$; since $X\neq 0$, it must be that $\Tr X>0$ and hence, by~\eqref{convergence_MM_n}, we have that $\Tr \MM_n(X)>0$ for sufficiently large $n$, too.

For $n$ large enough, define $\sigma_n \coloneqq \frac{\MM_n(X)}{\Tr \MM_n(X)}$ and $\sigma\coloneqq \frac{X}{\Tr X}$. Note that
\bb
\left\|\sigma_n - \sigma\right\|_1 &= \left\| \frac{\MM_n(X)}{\Tr \MM_n(X)} - \frac{X}{\Tr X} \right\|_1 \leq \frac{\left\|\MM_n(X) - X\right\|_1}{\left|\Tr \MM_n(X)\right|} + \|X\|_1 \left| \frac{1}{\Tr \MM_n(X)} - \frac{1}{\Tr X}\right| \tends{}{n\to\infty} 0\, .
\ee
Remembering that $\sigma_n\in \FF$ because $X\in \MM_n^{-1}\left(\cone(\FF)\right)$ and invoking the trace norm closedness of $\FF$, hypothesis~(ii), we infer that $\sigma\in \FF$ and therefore $X\in \cone (\FF)$, as claimed. This proves~\eqref{intersection_identity}.
\end{enumerate}

\item We now argue that for all $n\in \N$ the map $\MM_n$ is sequentially continuous with respect to the weak*-topology on the input space and the trace norm topology on the output. This means that for all sequences $(X_p)_{p\in \N}$ of trace class operators,
\bb
X_p\tends{w*}{p\to\infty} X\in \T(\HH)\qquad \Longrightarrow \qquad \MM_n(X_p)\tends{tn}{p\to\infty} \MM_n(X)\qquad \text{$\forall$ fixed $n\in \N$.}
\label{sequential_continuity}
\ee
To see why this is the case, notice first that
\bb
\|X\|_1\leq \sup_{p\in \N} \|X_p\|_1 \leq K < \infty\, .
\label{UBP_2}
\ee
The first inequality in~\eqref{UBP_2} follows from the relation $\|X\|_1\leq \liminf_{p\to\infty}\|X_p\|_1$, in turn a consequence of the general fact the dual norm is lower semi-continuous with respect to the weak*-topology~\cite[Theorem~2.6.14]{MEGGINSON}. The second inequality, instead, is proved in a similar manner to~\eqref{UBP_1}. Namely: (i)~the space of trace class operators is the dual to the Banach space of compact operators; and (ii)~$\sup_{p\in \N} \Tr X_p S <\infty$ for all compact $S$ because $\Tr X_p S\tends{}{p\to\infty} \Tr XS$, the uniform boundedness principle~\cite[Section~1.6.9]{MEGGINSON} immediately implies~\eqref{UBP_2}.

Now, for every $\varepsilon>0$ consider a finite-rank $\varepsilon$-approximation of $M_n$, i.e.\ an operator $M'_n$ with $\rk M'_n<\infty$ and $\|M_n-M'_n\|_\infty \leq \varepsilon$. For each $p\in \N$, we have that
\bb
\,\! &\left\|\MM_n(X_p) \!-\! \MM_n(X)\right\|_1 \\
&\qquad = \left\|M_n(X_p\!-\!X)M_n^\dag\right\|_1 \\
&\qquad \leq \|M_n\!-\!M'_n\|_\infty \left\|(X_p\!-\!X)M_n^\dag\right\|_1 + \left\|M'_n(X_p\!-\!X)M_n^\dag\right\|_1 \\
&\qquad \leq \|M_n\!-\!M'_n\|_\infty \left\|X_p\!-\!X\right\|_1 \|M_n\|_\infty + \|M'_n\|_\infty \|X_p\!-\!X\|_1 \left\|M_n\!-\!M'_n\right\|_\infty + \left\|M'_n(X_p\!-\!X)(M'_n)^\dag\right\|_1 \\
&\qquad \leq \|M_n\!-\!M'_n\|_\infty \left\|X_p\!-\!X\right\|_1 \|M_n\|_\infty + \left(\|M_n\|_\infty + \|M_n\!-\!M'_n\|_\infty \right) \|X_p\!-\!X\|_1 \left\|M_n\!-\!M'_n\right\|_\infty \\
&\hspace{8ex} + \left\|M'_n(X_p\!-\!X)(M'_n)^\dag\right\|_1 \\
&\qquad \leq 2K \varepsilon \|M_n\|_\infty + 2K\varepsilon \left(\|M_n\|_\infty + \varepsilon \right) + \left\|M'_n(X_p\!-\!X)(M'_n)^\dag\right\|_1\, .
\label{long_chain}
\ee
Since $M'_n = \sum_{j=1}^N \ketbraa{v_j}{w_j}$ is of finite rank and $X_p\tends{w*}{p\to\infty} X$, we have that
\bb
\left\|M'_n(X_p\!-\!X)(M'_n)^\dag\right\|_1 \leq \sum_{j,k=1}^N \|\ket{v_j}\|\|\ket{v_k}\| \left|\braket{w_j|(X_p\!-\!X)|w_k}\right| \tends{}{p\to\infty} 0\, .
\ee
Thus, from~\eqref{long_chain} we infer that
\bb
\limsup_{p\to\infty} \left\|\MM_n(X_p) \!-\! \MM_n(X)\right\|_1 \leq 2K \varepsilon \|M_n\|_\infty + 2K\varepsilon \left(\|M_n\|_\infty + \varepsilon \right) .
\ee
Given that $K<\infty$ and $\|M_n\|_\infty < \infty$ are fixed constants and that $\varepsilon>0$ is arbitrary, we conclude that in fact $\left\|\MM_n(X_p) - \MM_n(X)\right\|_1 \tends{}{p\to\infty} 0$, finally proving~\eqref{sequential_continuity}.

\item The last preliminary observation we need to make is that $\cone(\FF)$ is trace norm closed thanks to hypothesis~(ii). In fact, if for a sequence $(X_p)_{p\in \N}$ in $\cone(\FF)$ we have that $X_p\tends{tn}{p\to\infty} X\in \T(\HH)$, then either $X=0$, hence there is nothing to prove, or else $\Tr X>0$. In this latter case, since $\Tr X_p\tends{}{p\to\infty} \Tr X$ we have that $X_p/\Tr X_p \tends{tn}{p\to\infty} X/\Tr X$. Since the operators on the left-hand side (well defined for sufficiently large $p$) belong to $\FF$ and this is trace norm closed, it must be that also $X/\Tr X\in \FF$, i.e.\ $X\in \cone(\FF)$, as claimed.

\item We now conclude the argument by proving that $\cone (\FF)$ is weak*-closed. Thanks to hypothesis~(i), $\cone (\FF)$ is convex; also, we have seen in Section~\ref{topologies_subsec} that the Banach space $\K(\HH)$ of compact operators on a separable Hilbert space $\HH$ is itself separable.
In this situation, an immediate corollary of the Krein--\v{S}mulian theorem~\cite[Corollary~2.7.13]{MEGGINSON} ensures that $\cone (\FF)$ is weak*-closed if and only if it is weak*-sequentially-closed, i.e.\ if and only if for every sequence $(X_p)_{p\in N}$ in $\cone(\FF)$, $X_p\tends{w*}{p\to\infty} X\in \T(\HH)$ implies that also $X\in \cone(\FF)$. Using~\eqref{intersection_identity}, it suffices to show that $\MM_n(X)\in \cone(\FF)$ for all $n\in \N$. This follows straightforwardly by combining the weak*-to-trace-norm sequential continuity of $\MM_n$, proved in point~(II) above, and the trace norm closedness of $\cone(\FF)$, proved in~(III).
\end{enumerate}
\end{proof}

We now set out to prove our final main result, Theorem~\ref{variational_thm}. The key idea of the proof is contained in Lemma~\ref{variational_technical_lemma} below, which essentially tackles the simpler case where the relative entropy of resource is finite. Before presenting it, we will make sure that a result by Petz~\cite[Corollary~2]{Petz1988} holds in a slightly more general sense than originally stated.

\begin{lemma} \label{Petz_amended_lemma}
Let $\xi\in \T_+(\HH)$ be a positive semi-definite trace class operator, and let $X=X^\dag \in \B(\HH)$ be bounded. Then
\bb
\ln \Tr e^{\ln \xi + X} = \sup_{\omega\in \D(\HH)} \left\{ \Tr \omega X - D(\omega\|\xi) + \Tr \xi -1 \right\} ,
\label{Petz_corollary}
\ee
where the left-hand side is interpreted as in~\eqref{interpretation} when $\xi$ is not strictly positive definite. In particular, for a fixed (but arbitrary) $X$ as above the function
\bb
\T_+(\HH)\ni \xi \longmapsto \ln \Tr e^{\ln \xi + X}
\label{trace_perturbed_function}
\ee
is concave and monotonically non-decreasing.
\end{lemma}

\begin{rem}
Equation~\ref{Petz_corollary} is essentially~\cite[Corollary~2]{Petz1988} but without any faithfulness assumption on either $\omega$ or $\xi$.
\end{rem}


\begin{proof}[Proof of Lemma~\ref{Petz_amended_lemma}]
The famous variational formula due to Petz~\cite{Petz1988}, re-adapted to our notation, states that
\bb
D(\omega\|\xi) &= \sup_{X = X^\dag \in \B(\HH)} \left\{ \Tr \omega X - \ln \Tr e^{\ln \xi + X} + \Tr \xi -1 \right\}
\label{Petz variational 1}
\ee
whenever the state $\omega\in \D(\HH)$ and the operator $\xi\in \T_+(\HH)$ are both faithful, i.e.\ $\omega,\xi>0$. These further assumptions, which will turn out to be superfluous, mean that we cannot use~\eqref{Petz variational 1} directly in our proof. Fortunately, here we need only the `easy' inequality contained in~\eqref{Petz variational 1}, established by Petz in~\cite[Proposition~1]{Petz1988}. Let us make sure that this holds irrespectively of any faithfulness assumption, by writing
\bb
-\ln \Tr e^{\ln \xi + X} + \Tr e^{\ln \xi + X} - 1 &\eqt{(i)} D\left( 1 \big\| \Tr e^{\ln \xi + X} \right) \\
&\leqt{(ii)} D\left( \omega\, \big\|\, e^{\ln \xi + X} \right) \\
&\eqt{(iii)} D \left( \omega\| \xi\right) - \Tr \omega X + \Tr e^{\ln \xi + X} - \Tr \xi\, .
\label{intermediate_inequality_variational_technical}
\ee
Here, (i)~is just~\eqref{Umegaki} evaluated between two numbers, (ii)~is the data processing inequality for the relative entropy~\cite[Proposition~5.1(ii)]{PETZ-ENTROPY}, and finally (iii)~is Araki's identity~\cite[Theorem~3.10]{Araki-II} (see also~\cite[Corollary~12.8]{PETZ-ENTROPY} for how to remove the faithfulness assumption). Note that the equality in~(iii) is trivially true (as $+\infty=+\infty$) also when $\supp(\omega) \not\subseteq \supp(\xi) = \supp\left(e^{\ln \xi +X}\right)$. Now, from the above inequality we deduce that
\bb
\ln \Tr e^{\ln \xi + X} \geq \Tr \omega X - D \left( \omega\| \xi\right) + \Tr \xi - 1\, .
\label{Petz_corollary_eq1}
\ee
This proves that the left-hand side of~\eqref{Petz_corollary} is no smaller than the right-hand side. To establish the converse statement, one can observe that the inequality in~(ii) of~\eqref{intermediate_inequality_variational_technical}, and hence that in~\eqref{Petz_corollary_eq1}, is saturated for the choice $\omega=\frac{e^{\ln \xi +X}}{\Tr e^{\ln \xi+X}}$, simply because $D(\omega \| c\,\omega) = D(1\|c)$ holds for all $\omega\in \D(\HH)$ and $c>0$. This concludes the proof of~\eqref{Petz_corollary}.

We now prove that~\eqref{trace_perturbed_function} is monotonically non-decreasing. This will follows from~\eqref{Petz_corollary} once we show that for each $\omega\in \D(\HH)$ the function $\varphi_\omega:\T_+(\HH)\to \R \cup \{-\infty\}$ defined by $\varphi_\omega(\xi) \coloneqq \Tr \xi - D(\omega\|\xi)$ is non-decreasing. To prove this latter claim, pick $\xi_1\leq \xi_2$; we can take $\xi_2\leq \id$ without loss of generality, as otherwise we write $\xi_i = a \xi'_i$ with $a>0$ and $\xi'_i\leq \id$, and assuming that $\varphi_\omega(\xi'_1)\leq \varphi_\omega(\xi'_2)$ we deduce that
\bb
\varphi_\omega(\xi_1) = \varphi_\omega(a\xi'_1) = \varphi_\omega(\xi'_1) + \ln a\leq \varphi_\omega(\xi'_2) + \ln a = \varphi_\omega(a\xi'_2) = \varphi_\omega(\xi_2)\, .
\ee

When $\xi_1\leq \xi_2\leq \id$, denoting with $\omega = \sum_i \mu_i \ketbra{e_i}$ a spectral decomposition of $\omega$ one can use~\eqref{Umegaki_grown-ups} to write
\bb
\varphi_\omega(\xi_1) = 1 - \sum_i \left( \mu_i \ln \mu_i + \mu_i \big\| \sqrt{-\ln \xi_1} \ket{e_i}\big\|^2 \right) \leq 1 - \sum_i \left( \mu_i \ln \mu_i + \mu_i \big\| \sqrt{-\ln \xi_2} \ket{e_i}\big\|^2 \right) = \varphi_\omega(\xi_2)\, ,
\ee
where the inequality comes from the operator monotonicity of the logarithm (in the sense of Schm\"{u}dgen~\cite[Exercise~10.8.7.c, p.~249]{SCHMUEDGEN}, see also~\cite[Exercise~51, Ch.~VIII, p.~317]{REED}).

Finally, the concavity of the function in~\eqref{trace_perturbed_function} descends once again from~\eqref{Petz_corollary} and from the generally valid fact that the supremum over a convex set of a jointly concave function (see~\eqref{joint_convexity}) is itself concave.
\end{proof}

\begin{lemma} \label{variational_technical_lemma}
Let $\CC\subseteq \T_+(\HH)\cap B_1$ be a convex and weak*-compact subset of positive semidefinite trace class operators with trace at most $1$. Let $\rho\in \D(\HH)$ be a state with finite entropy $S(\rho)<\infty$ and such that $D_\CC(\rho) \coloneqq \inf_{\xi\in \CC} D(\rho\|\xi)<\infty$. Then
\bb
D_\CC(\rho) = \sup_{X=X^\dag \in \B(\HH)} \left\{ \Tr \rho X - \sup_{\xi\in \CC} \left\{ \ln \Tr e^{\ln \xi + X} - \Tr \xi +1 \right\} \right\} ,
\label{variational_technical}
\ee
where again we use~\eqref{interpretation} to interpret the right-hand side when $\xi$ is not strictly positive definite.
\end{lemma}

\begin{proof}
By the same reasoning we encountered in the proof of Theorem~\ref{achievable_relent_thm}, since thanks to Lemma~\ref{weak*_lsc_relent_lemma} the relative entropy is weak* lower semi-continuous in both arguments, and hence in particular with respect to the second, we can find $0\neq \xi\in \CC$ such that $D_\CC(\rho) = D(\rho\|\xi)$. Now, consider an arbitrary $\xi'\in \CC$ and some $\lambda\in [0,1]$. By convexity of $\CC$, we have that $(1-\lambda)\xi+\lambda\xi'\in \CC$, so that
\bb
0 \leq D\left( \rho\, \big\|\, (1-\lambda)\xi+\lambda\xi'\right) - D(\rho\|\xi) = - \Tr \rho \ln \left( (1-\lambda)\xi+\lambda\xi'\right) + \Tr \rho \ln \xi + \lambda\Tr \left[ \xi'- \xi\right] ,
\label{inequality_lambda}
\ee
where the last equality follows from~\eqref{Umegaki_splitting}. Note also that the expression at the rightmost side is well defined and finite, because $-\Tr \rho \ln \xi = D(\rho\|\xi) + S(\rho) + 1 - \Tr \xi <\infty$ and moreover
\bb
- \Tr \rho \ln \left( (1-\lambda)\xi+\lambda\xi'\right) \leq - \Tr \rho \ln \left( (1-\lambda)\xi \right) = -\Tr \rho \ln \xi - \ln(1-\lambda) <\infty
\label{boundedness_for_all_lambda}
\ee
thanks to the operator monotonicity of the logarithm.

We can now divide both sides of~\eqref{inequality_lambda} by $\lambda$ and take the limit $\lambda\to 0^+$. To carry out this computation we use a well-known representation of the Fr\'echet differential of the operator logarithm reported (without proof) e.g.\ in~\cite[Lemma~3.4]{Sutter2017}, obtaining that
\bb
0\leq \Tr \rho\, \Gamma_\xi (\xi - \xi') + \Tr\left[\xi'- \xi\right] = 1 - \Tr \rho\, \Gamma_\xi (\xi') + \Tr\left[\xi'- \xi\right] , \label{inequality_differentiated}
\ee
where
\bb
\Tr \rho\, \Gamma_\xi(X) \coloneqq&\ \int_0^{+\infty} ds\, \Tr \rho\, \frac{1}{\xi+s\id}\, X\, \frac{1}{\xi+s\id}\, , \label{Gamma_A}
\ee
and the integral on the right-hand side converges absolutely for $X=\xi'-\xi$ and $X=\xi'$. Since we were not able to deduce from the existing literature a completely rigorous proof of~\eqref{inequality_differentiated} and~\eqref{Gamma_A} that works in the infinite-dimensional case as well, we provide one in Appendix~\ref{app_differential_log}.

We now leave~\eqref{inequality_differentiated} for a moment, and return to the definition of $D_\CC$. 
By making the ansatz $\omega=\rho$ in~\ref{Petz_corollary}, for every bounded $X=X^\dag \in \B(\HH)$ and every $\xi'\in \CC$ we deduce that
\bb
D \left( \rho\| \xi'\right) \geq \Tr \rho X - \ln \Tr e^{\ln \xi' + X} + \Tr \xi' - 1\, ,
\ee
so that naturally
\bb
D_\CC(\rho) = \inf_{\xi'\in \CC} D(\rho\|\xi') &\geq \inf_{\xi'\in \CC} \sup_{X=X^\dag \in \B(\HH)} \left\{ \Tr \rho X - \ln \Tr e^{\ln \xi' + X} + \Tr \xi' - 1\right\} \\
&\geq \sup_{X=X^\dag \in \B(\HH)} \inf_{\xi'\in \CC} \left\{ \Tr \rho X - \ln \Tr e^{\ln \xi' + X} + \Tr \xi'- 1\right\} \\
&= \sup_{X=X^\dag \in \B(\HH)} \left\{ \Tr \rho X - \sup_{\xi'\in \CC} \left\{ \ln \Tr e^{\ln \xi' + X} - \Tr \xi' +1 \right\} \right\} ,
\label{first_inequality_variational_technical}
\ee
where in the second line we have used the elementary fact that $\inf_{a\in A} \sup_{b\in B} f(a,b) \geq \sup_{b\in B} \inf_{a\in A} f(a,b)$ holds for an arbitrary function $f:A\times B\to \R$ on any product set $A\times B$. This proves the first inequality needed to establish~\eqref{variational_technical}.


As for the second, consider that
\bb
&\sup_{X=X^\dag \in \B(\HH)} \left\{ \Tr \rho X - \sup_{\xi'\in \CC} \left\{ \ln \Tr e^{\ln \xi' + X} - \Tr \xi' + 1 \right\} \right\} \\
&\hspace{6ex} \geqt{(iv)} \limsup_{\epsilon\to 0^+} \bigg( \Tr \rho \left( \ln (\rho + \epsilon^2\id) - \ln (\xi+\epsilon \id) + \left( \Tr \xi - 1 \right) \id\right) \\
&\hspace{18ex} - \sup_{\xi'\in \CC} \left\{ \ln \Tr e^{\ln \xi' + \ln (\rho + \epsilon^2\id) - \ln (\xi+\epsilon \id) + \left( \Tr \xi - 1 \right) \id} - \Tr \xi' +1 \right\} \bigg) \\
&\hspace{6ex} \eqt{(v)} D(\rho\|\xi) - \liminf_{\epsilon\to 0^+} \sup_{\xi'\in \CC} \left\{ \ln \Tr e^{\ln \xi' + \ln (\rho + \epsilon^2\id) - \ln (\xi+\epsilon \id)} +\Tr [\xi-\xi'] \right\} \\
&\hspace{6ex} \geqt{(vi)} D(\rho\|\xi) - \liminf_{\epsilon\to 0^+} \sup_{\xi'\in \CC} \left\{ \ln \Tr \left[(\rho + \epsilon^2\id)\, \Gamma_{\xi+\epsilon\id}(\xi')\right] +\Tr [\xi-\xi'] \right\} \\
&\hspace{6ex}\eqt{(vii)} D(\rho\|\xi) - \liminf_{\epsilon\to 0^+} \sup_{\xi'\in \CC} \left\{ \ln \left(\Tr \left[\rho\, \Gamma_{\xi+\epsilon\id}(\xi')\right] + \epsilon^2 \Tr \left[ \frac{1}{\xi+\epsilon\id}\, \xi'\right] \right) + \Tr [\xi-\xi'] \right\} \\
&\hspace{6ex} \geqt{(viii)} D(\rho\|\xi) - \liminf_{\epsilon\to 0^+} \sup_{\xi'\in \CC} \left\{ \ln \big(\Tr \left[\rho\, \Gamma_\xi (\xi')\right] + \epsilon \Tr \xi' \big) +\Tr [\xi-\xi'] \right\} \\
&\hspace{6ex} \geqt{(ix)} D(\rho\|\xi) - \liminf_{\epsilon\to 0^+} \sup_{\xi'\in \CC} \left\{ \ln \big(1 + (1+\epsilon) \Tr \xi' - \Tr \xi \big) +\Tr [\xi-\xi'] \right\} \\
&\hspace{6ex} \geqt{(x)} D(\rho\|\xi) - \liminf_{\epsilon\to 0^+} \left\{ \ln (1+\epsilon) - \left(1-\Tr \xi\right) \frac{\epsilon}{1+\epsilon} \right\} \\
&\hspace{6ex} = D(\rho\|\xi) = D_\CC(\rho)\, .
\label{variational_technical_key}
\ee
Here, in~(iv) we made the ansatz $X=\ln(\rho+\epsilon^2\id) - \ln(\xi+\epsilon\id)+ (\Tr \xi-1)\id$, where later we take $\epsilon\to 0^+$; (v)~holds because denoting with $\rho=\sum_i p_i \ketbra{e_i}$ and $\xi=\sum_j \mu_j \ketbra{f_j}$ the spectral decompositions of $\rho$ and $\xi$, we have that
\bb
&\lim_{\epsilon\to 0^+} \Tr \rho \left( \ln (\rho+\epsilon^2\id) - \ln (\xi+\epsilon\id)\right) + \Tr \xi -1 \\
&\qquad = \lim_{\epsilon\to 0^+} \left( \sumno_i p_i \ln \left(p_i +\epsilon^2\right) + \sumno_{i,j} p_i |\braket{e_i|f_j}|^2 (-\ln (\mu_j+\epsilon)) \right) + \Tr \xi -1 \\
&\qquad = -S(\rho) - \Tr \rho \ln \xi + \Tr \xi -1 = D(\rho\|\xi)
\ee
by the dominated convergence theorem applied to each series (remember that $S(\rho)<\infty$ and also $-\Tr\rho\ln\xi = D(\rho\|\xi) + S(\rho) + 1-\Tr\xi<\infty$); (vi)~follows from Lieb's three-operator inequality, proved in~\cite[Theorem~7]{lieb73c} for the finite-dimensional case and extended to infinite dimensions in Lemma~\ref{Lieb_3_matrix_lemma}; in~(vii) we noted that for a bounded $0<\epsilon\id\leq A \coloneqq \xi + \epsilon\id \in \B(\HH)$ with spectral decomposition $A=\int_\epsilon^\infty dP(\lambda)\, \lambda$ and a trace class $0\leq B \coloneqq \xi' \in \T(\HH)$
\bb
\Tr \Gamma_{A}(B) &= \int_0^\infty ds\, \Tr \frac{1}{A+s\id}\, B\, \frac{1}{A+s\id} \\
&= \int_0^\infty ds\, \Tr \frac{1}{(A+s\id)^2}\, B \\
&= \int_0^\infty ds\, \int_\epsilon^\infty \Tr\left[dP(\lambda)\, B \right] \frac{1}{(\lambda+s)^2} \\
&= \int_\epsilon^\infty \Tr\left[dP(\lambda)\, B \right] \int_0^\infty ds\, \frac{1}{(\lambda+s)^2} \\
&= \int_\epsilon^\infty \Tr\left[dP(\lambda)\, B \right] \frac{1}{\lambda} \\
&= \Tr \frac{1}{A}\, B\, ,
\ee
where we used Tonelli's theorem to exchange the integrals; (viii)~comes from observing that on the one hand $\Tr \frac{1}{\xi+\epsilon\id}\, \xi' \leq \frac1\epsilon \Tr \xi'$, and on the other
\bb
\Tr \rho\, \Gamma_{\xi+\epsilon\id}(\xi') &= \int_0^\infty ds\, \Tr \rho\, \frac{1}{\xi+\epsilon\id+s\id}\, \xi' \frac{1}{\xi+\epsilon\id+s\id} \\
&= \int_\epsilon^\infty dt\, \Tr \rho\, \frac{1}{\xi+t\id}\, \xi' \frac{1}{\xi+t\id} \\
&\leq \int_0^\infty dt\, \Tr \rho\, \frac{1}{\xi+t\id}\, \xi' \frac{1}{\xi+t\id} \\
&= \Tr \rho\, \Gamma_\xi(\xi')\, ;
\ee
in~(ix) we leveraged~\eqref{inequality_differentiated}; finally, in~(x) we remembered that $\Tr \xi'\in (0,1]$ and $\Tr \xi\in [0,1]$ (the case where $\Tr \xi'=0$ and thus $\xi'=0$ is trivial and can be excluded) and observed that
\bb
\sup_{x\in (0,1]}\left\{ \ln \left( (1+\epsilon) x+1-a\right) +a -x \right\} = \ln(1+\epsilon) - \frac{\epsilon}{1+\epsilon} (1-a) 
\ee
for all $\epsilon>0$ and $a\in [0,1]$. This concludes the justification of~\eqref{variational_technical_key} and thus the proof.
\end{proof}

We are finally ready to present the proof of our last general result.

\begin{proof}[Proof of Theorem~\ref{variational_thm}]
Let $\rho\in \D(\HH)$ be a state with finite entropy $S(\rho)<\infty$. If $D_\FF(\rho)<\infty$ then we can use directly Lemma~\ref{variational_technical_lemma} and complete the proof. Therefore, what we set out to do now is to devise an argument that tackles also the case where $D_\FF(\rho)=+\infty$. To this end, start by observing that the general inequality
\bb
D_\FF(\rho) \geq \sup_{X = X^\dag\in \B(\HH)} \left\{ \Tr \rho X - \sup_{\sigma\in \FF} \ln \Tr e^{\ln \sigma + X} \right\}
\label{first_variational}
\ee
can be proved exactly as~\eqref{first_inequality_variational_technical}, whose derivation, in fact, does not rely on any further assumption. The problem is to establish the converse inequality.

To this end, for an arbitrary parameter $p\in (0,1]$ let us construct the convex set
\bb
\widetilde{\FF}_p \coloneqq (1-p) \widetilde{\FF} + p\{\rho\} = \left\{ (1-p)\eta + p\rho:\, \eta\in \widetilde{\FF}\right\} ,
\ee
where $\widetilde{\FF}=\cone(\FF)\cap B_1$ is as in~\eqref{F_tilde}. Since $\widetilde{\FF}_p$ is the sum of two weak*-compact sets, it is itself weak*-compact.

Now, define $D_\FF^{(p)}(\rho) \coloneqq \inf_{\xi \in \widetilde{\FF}_p} D(\rho\|\xi)$. Since
\bb
D_\FF^{(p)}(\rho) = \inf_{\eta \in \widetilde{\FF}} D\left( \rho \big\| (1-p)\eta+p\rho\right) \leq D\left( \rho \big\| p\rho\right) = -\ln p < \infty\, ,
\ee
we see that
\bb
D_\FF^{(p)}(\rho) &\eqt{(i)} \sup_{X=X^\dag \in \B(\HH)} \left\{ \Tr \rho X - \sup_{\xi\in \widetilde{\FF}_p} \left\{ \ln \Tr e^{\ln \xi + X} - \Tr \xi + 1 \right\} \right\} \\
&= \sup_{X=X^\dag \in \B(\HH)} \left\{ \Tr \rho X - \sup_{\eta\in \widetilde{\FF}} \left\{ \ln \Tr e^{\ln \left((1-p)\eta + p\rho\right) + X} + (1-p) \left(1-\Tr \eta\right) \right\} \right\} \\
&\leqt{(ii)} \sup_{X=X^\dag \in \B(\HH)} \left\{ \Tr \rho X - \sup_{\eta\in \widetilde{\FF}} \left\{ (1-p) \ln \Tr e^{\ln \eta + X} + p \ln \Tr e^{\ln \rho + X} + (1-p) \left(1-\Tr \eta\right) \right\} \right\} \\
&\leqt{(iii)} \sup_{X=X^\dag \in \B(\HH)} \left\{ \Tr \rho X - \sup_{\eta\in \widetilde{\FF}} \left\{ (1-p) \ln \Tr e^{\ln \eta + X} + p \Tr \rho X + (1-p) \left(1-\Tr \eta\right) \right\} \right\} \\
&= (1-p) \sup_{X=X^\dag \in \B(\HH)} \left\{ \Tr \rho X - \sup_{\eta\in \widetilde{\FF}} \left\{ \ln \Tr e^{\ln \eta + X} -\Tr \eta + 1 \right\} \right\} \\
&= (1-p) \sup_{X=X^\dag \in \B(\HH)} \left\{ \Tr \rho X - \sup_{\sigma\in \FF}\sup_{\lambda\in [0,1]} \left\{ \ln \Tr e^{\ln \sigma + X} + \ln \lambda -\lambda + 1 \right\} \right\} \\
&= (1-p) \sup_{X=X^\dag \in \B(\HH)} \left\{ \Tr \rho X - \sup_{\sigma\in \FF} \ln \Tr e^{\ln \sigma + X} \right\} .
\ee
Here, (i)~holds thanks to Lemma~\ref{variational_technical_lemma}; in~(ii) we exploited the concavity of the function $\xi \mapsto \ln \Tr e^{\ln \xi+X}$, as established by Lemma~\ref{Petz_amended_lemma}; and in~(iii) we applied the Peierls--Bogoliubov inequality~\cite[Theorem~7]{Ruskai1972}. Taking the limit $p\to 0^+$ we therefore deduce that
\bb
\limsup_{p\to 0^+} D_\FF^{(p)}(\rho) \leq \sup_{X=X^\dag \in \B(\HH)} \left\{ \Tr \rho X - \sup_{\sigma\in \FF} \ln \Tr e^{\ln \sigma + X} \right\} .
\label{second_1_variational}
\ee

Now, as in the proof of Lemma~\ref{variational_technical_lemma}, by weak*-compactness for all $p$ we can find some $\eta_p\in \widetilde{\FF}$ such that $D_\FF^{(p)} = D\left(\rho\,\|\, (1-p)\eta_p + p\rho\right)$. By the same reason, when taking the limit $p\to 0^+$ we can assume without loss of generality that $\eta_p \ctends{w*}{p\to 0^+}{0pt} \eta$, for some $\eta\in \widetilde{\FF}$. This naturally implies that also $(1-p) \eta_p + p\rho \ctends{w*}{p\to 0^+}{0pt} \eta$; using the lower semi-continuity of the relative entropy, we see that
\bb
\liminf_{p\to 0^+} D_\FF^{(p)}(\rho) = \liminf_{p\to 0^+} D\left(\rho\,\|\, (1-p)\eta_p + p\rho\right) \geq D(\rho \| \eta) \geq D_\FF(\rho)\, .
\ee
Combining this with~\eqref{second_1_variational} yields that
\bb
D_\FF(\rho) \leq \sup_{X=X^\dag \in \B(\HH)} \left\{ \Tr \rho X - \sup_{\sigma\in \FF} \ln \Tr e^{\ln \sigma + X} \right\} .
\ee
Together with~\eqref{first_variational}, this proves~\eqref{variational}.


\end{proof}

\begin{rem}
In the special case where $\FF =\{\sigma\}$ contains a single state (and $S(\rho)<\infty$), Theorem~\ref{variational_thm} yields immediately the identity
\bb
D(\rho\|\sigma) = \sup_{X=X^\dag \in \B(\HH)} \left\{ \Tr \rho X - \ln \Tr e^{\ln \sigma +X} \right\} .
\label{Petz_variational}
\ee
This is naturally just Petz's variational formula~\cite{Petz1988}, without any faithfulness assumption on either $\rho$ or $\sigma$. In this sense, Theorem~\ref{variational_thm} can also be seen as a generalisation of Petz's result.
\end{rem}

\begin{rem} \label{Sion?_rem}
One could wonder whether there may be a more direct approach to prove Theorem~\ref{variational_thm}. Assuming that we have established~\eqref{Petz_variational}, or equivalently~\eqref{Petz variational 1}, we could plug this into~\eqref{D_F_eta}, obtaining
\bb
D_\FF(\rho) = \inf_{\eta\in \widetilde{F}} \sup_{X = X^\dag \in \B(\HH)} \left\{ \Tr \rho X - \ln \Tr e^{\ln \eta + X} + \Tr \eta -1 \right\} .
\ee
Now, knowledge of~\cite{Berta2017} would suggest that we try to exchange the supremum and infimum using Sion's theorem. Note that the weak*-topology makes $\widetilde{\FF}$ compact, which is encouraging. Since the other conditions can be shown to be met, we need only to ask ourselves whether $\widetilde{\FF}\ni \eta \mapsto f_X(\eta) \coloneqq - \ln \Tr e^{\ln \eta + X} + \Tr \eta-1$ is lower semi-continuous for all fixed $X=X^\dag\in \B(\HH)$. Unfortunately, that is not the case. To see why, it suffices to take $X=0$ and compute $f_0(\eta) = -\ln \Tr \eta + \Tr \eta -1$. Evaluating this on a sequence $(\eta_n)_{n\in \N}$ with constant trace $\Tr \eta_n\equiv 1$ but such that $\eta_n \tendsn{w*} 0$ (such as the one constructed in Section~\ref{topologies_subsec}) shows that $\lim_{n\to \infty} f_0(\eta_n) = 1 < \infty = f_0(0)$.
\end{rem}

\subsection{On the continuity of the relative entropy of resource}

The lower semi-continuity of $D_\FF$ implies the following sufficient conditions for its local continuity.

\begin{prop} \label{RE-LS-c+_1}
Let $(\rho_n)_{n\in \N}$ be a sequence of states in $\D(\HH)$ converging to a state $\rho$ in trace norm topology, i.e.\ $\rho_n \tends{tn}{n\to\infty} \rho$. Then the relation
\begin{equation}\label{CS-r+}
\lim_{n\to+\infty} D_\FF(\rho_n)=D_\FF(\rho)\leq+\infty
\end{equation}
holds provided that one of the following conditions is valid:
\begin{enumerate}[(a)]
\item $\rho_n=\Phi_n(\rho)/\Tr\Phi_n(\rho)$, where $\Phi_n$ is a positive trace-non-increasing linear transformation of $\T(\HH)$ such that 
\begin{equation}\label{Phi-cond}
\Phi_n(\FF)\subseteq \widetilde{\FF} = \cone(\FF)\cap B_1
\end{equation}
for each $n$, and $\Tr\Phi_n(\rho)\tends{}{n\to\infty} 1$;
\item $c_n\rho_n\leq \sigma_n$ for all $n$, where $c_n\tends{}{n\to\infty} 1$ are real numbers, and $(\sigma_n)_{n\in \N}$ is a sequence of states in $\D(\HH)$ such that 
\begin{equation}
\lim_{n\to+\infty} D_\FF(\sigma_n)=D_\FF(\rho)\, .
\end{equation}
\end{enumerate}
\end{prop}

\begin{proof}
Let us prove the claims one at a time.
\begin{enumerate}[(a)]
\item A previous result by one of us~\cite[Lemma~1]{Shirokov-relent} implies that $\left(\Tr\Phi_n(\rho)\right) D_\FF(\rho_n)\leq D_\FF(\rho)$ for all $n$. So, in this case~\eqref{CS-r+} follows directly from Theorem~\ref{achievable_relent_thm}(b).

\item For an arbitrary set $\FF$ the function $D_\FF$ satisfies the inequality
\begin{equation}\label{F-p-1}
D_\FF(p\rho+(1-p)\sigma)\geq pD_\FF(\rho)+(1-p)D_\FF(\sigma)-h_2(p)
\end{equation}
valid for any states $\rho$ and $\sigma$ in $\D(\HH)$ and any $p\in(0,1)$, where
\bb
h_2(p)\coloneqq -p\ln p-(1-p)\ln(1-p)
\label{h_2}
\ee
is the binary entropy. Inequality~\eqref{F-p-1} follows directly from the inequality~\cite[Proposition~5.24]{PETZ-ENTROPY}
\begin{equation*}
D(p\rho+(1-p)\sigma\| \omega)\geq pD(\rho\| \omega)+(1-p)D(\sigma\| \omega)-h_2(p)\, .
\end{equation*}
Now, we argue that under the hypotheses in~(b) we have that $\sigma_n\tends{tn}{n\to\infty}\rho$. In fact, since $\sigma_n - \rho_n = \sigma_n - c_n\rho_n + \left(c_n-1\right)\rho_n$, using the triangle inequality and the fact that $\|X\|_1=\Tr X$ for $X\geq 0$ we arrive at $\left\|\sigma_n - \rho_n\right\|_1\leq 2\left|c_n-1\right|$. The right-hand side tends to $0$ by hypothesis; hence, so does the left-hand side.
The condition $c_n\rho_n\leq \sigma_n$ and inequality~\eqref{F-p-1} also show that since $\sigma_n = (1-c_n)\frac{\sigma_n - c_n\rho_n}{1-c_n} + c_n \rho_n$, we have that
\begin{equation}
D_\FF(\sigma_n)\geq c_n D_\FF(\rho_n)-h_2(c_n)\qquad \forall n\, .
\end{equation}
It follows that
\begin{equation}
\limsup_{n\to+\infty} D_\FF(\rho_n)\leq \lim_{k\to+\infty} D_\FF(\sigma_n)=D_\FF(\rho)\, .
\end{equation}
This relation and Theorem~\ref{achievable_relent_thm}(b) imply~\eqref{CS-r+}.
\end{enumerate}
\end{proof}

Condition~(b) in Proposition~\ref{RE-LS-c+_1} shows that  
\begin{equation}
  \lim_{n\to+\infty} D_\FF(\rho_n)=D_\FF(\rho)\leq+\infty
\end{equation}
for any sequence $(\rho_n)_{n\in \N}$ converging to an arbitrary state $\rho$ provided that $c_n\rho_n\leq \rho$ for all $n$, where $(c_n)_{n\in \N}$ is a sequence of positive numbers tending to $1$.\footnote{A similar property holds for the von Neumann entropy (due to its concavity and lower semi-continuity). It is widely used in analysis of infinite dimensional quantum systems.} This holds, in particular, for the sequence $(\rho_n)_{n\in \N}$ of finite-rank states obtained by truncation of the spectral decomposition of $\rho$. We deduce the following: 

\begin{cor}
The function $D_\FF$ is completely determined by its values on the set of finite-rank states in $\D(\HH)$.
\end{cor}

\subsection{On the continuity of the minimiser(s)}

Let $\FF$ be a convex subset of $\D(\HH)$ such that $\cone(\FF)$ is weak*-closed. For any state $\rho$ denote by $\Sigma_\FF(\rho)$ the \deff{minimiser set} of $\rho$ defined in Theorem~\ref{achievable_relent_thm}, i.e.\ the set of all states $\sigma\in \FF$ such that $D_\FF(\rho)=D\left(\rho\|\sigma\right)$. By Theorem~\ref{achievable_relent_thm}, $\Sigma_\FF(\rho)$ is always nonempty, convex, trace norm compact, and consists of a single state if $\rho$ is faithful. 

Here we will push forward the investigation of the properties of the function $\rho\mapsto \Sigma_\FF(\rho)$ by considering its continuity in a neighbourhood of a faithful state. A variation on the argument used in the proof of Theorem~\ref{achievable_relent_thm} yields the following result.

\begin{prop}\label{RE-LS-c+_2}
Let $\FF\subseteq \D(\HH)$ be convex and such that $\cone(\FF)$ is weak*-closed. Let $(\rho_n)_{n\in \N}$ be any sequence of states in $\D(\HH)$ converging to a faithful state $\rho$ in trace norm, i.e.\ $\rho_n \tends{tn}{n\to\infty}\rho$ such that 
\begin{equation}\label{CS-r++}
\lim_{n\to+\infty} D_\FF(\rho_n)=D_\FF(\rho)<+\infty.
\end{equation}
Then any sequence $(\sigma_n)_{n\in \N}$ such that $\sigma_n\in\Sigma_\FF(\rho_n)$ for all $n$ converges in trace norm to the unique state in $\Sigma_\FF(\rho)$.
\end{prop}

\begin{proof}
We can assume without loss of generality that $D_\FF(\rho_n)<\infty$ for all $n\in \N$. Proceeding by contradiction, and calling $\sigma$ the unique state in $\Sigma_\FF(\rho)$, we can also posit, up to selecting a subsequence of $(\sigma_n)_{n\in \N}$, that $\liminf_{n\to\infty}\left\|\sigma_n-\sigma\right\|_1>0$. Since the set $\widetilde{\FF} = \cone(\FF)\cap B_1$ defined by~\eqref{F_tilde} is weak*-compact and therefore sequentially weak*-compact (see Remark~\ref{bounded_w*_metrisable_rem}), we can extract from the sequence $(\sigma_n)_{n\in \N}$ a weak*-converging sub-sequence $(\sigma_{n_k})_{k\in \N}$, so that $\sigma_{n_k}\tends{w*}{k\to\infty}\eta_*\in \widetilde{\FF}$. 
In analogy with~\eqref{chain_inequalities}, we now have that
\bb
D\left( \rho\,\|\sigma\right) &= D_\FF(\rho) \\
&\eqt{(i)} \lim_{k\to\infty} D\left(\rho_{n_k}\|\sigma_{n_k}\right) \\
&\geqt{(ii)} D\left(\rho\|\eta_{*}\right) \\
&\eqt{(iii)} D\left(\rho\|\sigma_*\right) + \Tr \eta_* - 1 - \ln \Tr \eta_* \\
&\geqt{(iv)} D\left(\rho\|\sigma_*\right) ,
\ee
where~(i) is thanks to~\eqref{CS-r++}, (ii)--(iv) are justified as the corresponding inequalities in~\eqref{chain_inequalities}, and in~(iii) we introduced the state $\sigma_*\coloneqq (\Tr \eta_*)^{-1} \eta_* \in \FF$. Note that $\Tr \eta_*>0$ thanks to~(ii). The above relation together with the faithfulness of $\rho$ guarantees that $\sigma_*=\sigma$ is the unique state in $\Sigma_\FF(\rho)$, and moreover $\Tr \eta_*=1$, so that in fact $\eta_*=\sigma_*=\sigma$. Since weak* and trace norm topology coincide on the set of density operators (see Section~\ref{topologies_subsec}), from $\sigma_{n_k}\tends{w*}{k\to\infty}\eta_*$ we infer that $\sigma_{n_k}\tends{tn}{k\to\infty}\sigma$. Hence, $\liminf_{n\to\infty} \left\|\sigma_n- \sigma\right\|_1=0$, and we have reached a contradiction.
\end{proof}

\begin{rem} If the limit state $\rho$ in Proposition~\ref{RE-LS-c+_2} is not faithful, the arguments from the above proof show that the sequence $(\sigma_n)_{n\in \N}$ is relatively compact with respect to the trace norm and all the limit points of this sequence are contained in $\Sigma_\FF(\rho)$. 
\end{rem}


\section{Applications} \label{applications_sec}

\subsection{Relative entropy of entanglement} \label{entanglement_subsec}

Let $\HH_A,\HH_B$ be two separable Hilbert spaces.\footnote{The notion of separability of Hilbert spaces has nothing to do with that of separability of quantum states explained here.} The set of \deff{separable states} on the bipartite system with Hilbert space $\HH_{AB}\coloneqq \HH_A\otimes \HH_B$ is defined as the closed convex hull of product states, i.e.
\begin{equation}
    \SEP_{AB} \coloneqq \cl_{\mathrm{tn}} \left( \co\left\{ \ketbra{\psi}_A \otimes \ketbra{\phi}_B:\, \ket{\psi}_A\in \HH_A,\, \ket{\phi}_B\in \HH_B,\, \braket{\psi|\psi}=1=\braket{\phi|\phi} \right\} \right) .
    \label{SEP}
\end{equation}
Here, the closure is taken with respect to the trace norm topology. As it turns out, a state $\sigma_{AB}$ is separable if and only if it can be decomposed as $\sigma_{AB} = \int \ketbra{\psi}_A \otimes \ketbra{\phi}_B\, d\mu(\psi,\phi)$ for some Borel probability measure $\mu$ defined on the product of the sets of local normalised pure states~\cite{Holevo2005}. A state that is not separable is called \deff{entangled}.

The \deff{relative entropy of entanglement} is nothing but the relative entropy of resource associated with the set of free states $\SEP$ defined by~\eqref{SEP}. In formula, it is defined by~\cite{Vedral1997}
\bb
E_R(\rho_{AB}) \coloneqq D_{\SEP_{AB}}(\rho_{AB}) = \inf_{\sigma_{AB}\in \SEP_{AB}} D\left(\rho_{AB}\|\sigma_{AB}\right)\, .
\label{relent_entanglement}
\ee
Its central importance in entanglement theory stems from the fact that its regularisation bounds from above the distillable entanglement and from below the entanglement cost~\cite{Vedral1998, Horodecki2000, Donald1999, Donald2002}.

To set the stage for the application of Theorem~\ref{achievable_relent_thm}, we need to ask ourselves whether the cone generated by separable states is weak*-closed. This has been proved already in Ref.~\cite[Lemma~25]{taming-PRA}; an alternative and significantly more general proof resting on Theorem~\ref{w*_closed_condition_thm} will be presented below (cf.\ Corollary~\ref{alpha_separable_w*_closed_cor}). A straightforward application of Theorems~\ref{achievable_relent_thm} and~\ref{variational_thm} then yields:

\begin{cor} \label{relent_entanglement_cor}
For an arbitrary bipartite system with separable Hilbert space $\HH_{AB}$, the relative entropy of entanglement is:
\begin{enumerate}[(a)]
    \item always achieved, meaning that for all $\rho_{AB}\in \D(\HH_{AB})$ there exists a state $\sigma_{AB}\in \SEP_{AB}$ such that $E_R(\rho_{AB}) = D\left(\rho_{AB}\|\sigma_{AB}\right)$; and
    \item lower semi-continuous with respect to the trace norm topology.
\end{enumerate}
Moreover, for every state $\rho=\rho_{AB}\in \D(\HH_{AB})$ with finite entropy $S(\rho)<\infty$, it holds that
\bb
E_R(\rho_{AB}) = \sup_{X = X^\dag\in \B(\HH_{AB})} \left\{ \Tr \rho\, X - \sup_{\sigma \in \SEP_{\!AB}} \ln \Tr e^{\ln \sigma + X} \right\} .
\label{variational_relent}
\ee
\end{cor}

The following example shows that the state $\sigma_{AB}\in \SEP_{AB}$ such that $E_R(\rho_{AB}) = D\left(\rho_{AB}\|\sigma_{AB}\right)$ may not be unique, even in the finite-dimensional case.

\begin{ex} \label{non-un}
Let $A$ and $B$ be qubit systems, and define the maximally entangled state $\Phi_{AB}\coloneqq \ketbra{\Phi}_{AB}$, where $\ket{\Phi}_{AB}\coloneqq \frac{1}{\sqrt{2}}(\ket{00}+\ket{11})$. It is well known that its relative entropy of entanglement is $\ln 2$~\cite[Proposition~1]{Vedral1997}. Now, setting
\bb
\sigma^1_{AB} \coloneqq \frac{1}{2}\Phi_{AB}+\frac{1}{6}\left(I_{AB}-\Phi_{AB}\right) ,\qquad \sigma^2_{AB} \coloneqq \frac{1}{2}\Phi_{AB}+\frac{1}{4}\left(\ketbra{01}+\ketbra{10}\right) ,
\ee
we see that
\bb
D\left(\Phi_{AB}\|\sigma^1_{AB}\right)=D\left(\rho_{AB}\|\sigma^2_{AB}\right)=\ln 2 = E_R(\Phi_{AB})\, .
\ee
Since $\sigma_{AB}^1\neq \sigma_{AB}^2$, the relative entropy of entanglement is not uniquely achieved in this case.
\end{ex}

\subsection{Relative entropy of multi-partite entanglement} \label{multipartite_entanglement_subsec}

The above Corollary~\ref{relent_entanglement_cor} can be generalised to the multi-partite setting. To do so, let us start by fixing some terminology. We follow the particularly clear exposition of Szalay~\cite{Szalay2015}. For a given positive integer $m$, representing the total number of parties, let $[m]\coloneqq \{1,\ldots, m\}$. Of particular interest to us are the partitions of this set. A partition $\pi=(\pi(j))_j$ is a finite collection of non-empty sets $\pi(j)\subseteq [m]$ that do not intersect, i.e.\ $\pi(j)\cap \pi(j')=\emptyset$ for all $j\neq j'$, and together cover the whole $[m]$, i.e.\ $\cup_j \pi(j)=[m]$. We will denote the set of all partitions of $[m]$ with $P(m)$.

In the context of multi-partite entanglement, a partition $\pi\in P(m)$ can represent the allowed quantum interactions between some systems $A_1,\ldots, A_m$: two parties $A_\ell, A_{\ell'}$ can exchange quantum messages and thus establish entanglement if and only if $\ell$ and $\ell'$ belong to the same element of the partition, i.e.\ $\ell,\ell'\in \pi(j)$ for some $j$. One can however imagine a setting where not one but several partitions are allowed in this sense. Let $\pi =\{\pi_k\}_k\subseteq P(m)$ be the (non-empty) set of allowed partitions. We can then define the set of \deff{$\boldsymbol{\pi}$-separable states} by~\cite{Szalay2015}
\bb
\SEP_{A_1\ldots A_m}^\pi \coloneqq \cl_{\mathrm{tn}} \left( \co\left( \bigcup\nolimits_k \left\{ \bigotimes\nolimits_j \Psi^{(j)}_{A_{\pi_k(j)}}:\ \ket{\Psi^{(j)}}_{A_{\pi_k(j)}}\!\in \HH_{A_{\pi_k(j)}},\ \braket{\Psi^{(j)} | \Psi^{(j)}} = 1 \right\} \right) \right) .
\label{alpha_separable}
\ee
Here, $A_{\pi_k(j)}$ is the system obtained by joining those $A_i$ such that $i\in \pi_k(j)$, and we used the shorthand notation $\Psi^{(j)}\coloneqq \ketbra{\Psi^{(j)}}$.

The notions of multi-partite separability most commonly employed in the literature correspond to special choices of $\pi$ in~\eqref{alpha_separable}. Namely, we can pick $\pi$ to be:
\begin{itemize}
    \item The set containing only the finest partition $\left\{\{1\},\ldots, \{m\}\right\}$: the states obtained in~\eqref{alpha_separable} are then called \deff{totally} (or \deff{fully}) \deff{separable}.
    \item More generally, the set of partitions of $[m]$ into at least $k$ subsets: the corresponding states in~\eqref{alpha_separable} are usually referred to as \deff{$\boldsymbol{k}$-separable}~\cite{Vedral1997, Acin-3-qubits, Szalay2015} (of particular interest is the case of bi-separability).
    \item Taking a different angle, we can consider also the set of partitions of $[m]$ involving subsets of at most $k$ elements: the states obtained in~\eqref{alpha_separable} are then called \deff{$\boldsymbol{k}$-producible}~\cite{Seevinck2001, Szalay2015}.
\end{itemize}

Let $\rho_{A_1\ldots A_m}$ be a state of an $m$-partite quantum system $A_1\ldots A_m$. For a generic non-empty $\pi\subseteq P(m)$, we can define its \deff{relative entropy of $\boldsymbol{\pi}$-entanglement} by
\bb
E_{R,\pi}(\rho_{A_1\ldots A_m}) \coloneqq D_{\SEP^\pi_{A_1\ldots A_m}}(\rho_{A_1\ldots A_m}) = \inf_{\sigma_{A_1\ldots A_m}\in \SEP_{A_1\ldots A_m}^{\pi}} D\left(\rho_{A_1\ldots A_m}\|\sigma_{A_1\ldots A_m}\right)\, .
\label{m-E-r}
\ee
Special cases of the quantity in~\eqref{m-E-r} have been considered by several authors in different contexts~\cite{BrandaoPlenio2, Piani2009, Wei2008, Zhu2010, Friedland2011, Das2020}.

To apply Theorems~\ref{achievable_relent_thm} and~\ref{variational_thm} to the relative entropy of $\pi$-entanglement we need to first extend the result of Ref.~\cite[Lemma~25]{taming-PRA} on the weak*-closedness of the cone of bipartite separable states. This can be done thanks to a swift application of Theorem~\ref{w*_closed_condition_thm}.




\begin{cor} \label{alpha_separable_w*_closed_cor}
For a positive integer $m$, let $\pi\subseteq P(m)$ be a non-empty subset of partitions of $[m]$. Then, for arbitrary separable Hilbert spaces $\HH_{A_1},\ldots, \HH_{A_m}$, the cone generated by the set of $\pi$-separable states~\eqref{alpha_separable} is weak*-closed.

In particular, the cones generated by $k$-separable and $k$-producible states are weak*-closed, for all positive integers $k\leq m$.
\end{cor}


\begin{proof}[Proof of Corollary~\ref{alpha_separable_w*_closed_cor}]
For $\ell=1,\ldots, m$, let $\{\ket{p}_{A_\ell}\}_{p\in \N}$ be an orthonormal basis of the Hilbert space $\HH_{A_\ell}$. Defining the projectors $Q^{(n)}_{A_\ell}\coloneqq \sum_{p=0}^{n-1}\ketbra{p}_{A_\ell}$, set $M_n=\bigotimes_{\ell=1}^m Q^{(n)}_{A_\ell}$ and $\MM_n(\cdot)=M_n(\cdot) M_n^\dag$. Note that $M_n$ is of finite rank and hence compact for all $n$. Also, using the fact that the coefficients of any multi-partite pure state with respect to the product basis $\left\{\bigotimes_{\ell=1}^m \ket{k_\ell}_{A_\ell} \right\}_{k_1,\ldots, k_m\in \N}$ form a square-summable sequence, one sees that $(M_n)_{n\in \N}$ converges to the identity in the strong operator topology in the sense explained in the statement of Theorem~\ref{w*_closed_condition_thm}. Since it is straightforward to verify that $\MM_n\left( \cone (\FF)\right)\subseteq \cone (\FF)$ for $\FF=\SEP^\pi_{A_1\ldots A_m}$, and given that this set is trace norm closed and convex by construction, we can apply Theorem~\ref{w*_closed_condition_thm} and conclude.
\end{proof}

We are now ready to apply Theorems~\ref{achievable_relent_thm} and~\ref{variational_thm}.

\begin{cor} \label{relent_alpha_entanglement_cor}
For a positive integer $m$, let $\pi\subseteq P(m)$ be a non-empty subset of partitions of $[m]$. Then, for arbitrary separable Hilbert spaces $\HH_{A_1},\ldots, \HH_{A_m}$, the relative entropy of $\pi$-entanglement is:
\begin{enumerate}[(a)]
    \item always achieved, meaning that for all states $\rho_{A_1\ldots A_m}$ there exists a state $\sigma_{A_1\ldots A_m}\in \SEP^\pi_{A_1\ldots A_m}$ such that $E_{R,\, \pi}(\rho_{A_1\ldots A_m}) = D\left(\rho_{A_1\ldots A_m}\|\sigma_{A_1\ldots A_m}\right)$; and
    \item lower semi-continuous with respect to the trace norm topology.
\end{enumerate}
Moreover, for every state $\rho=\rho_{A_1\ldots A_m}\in \D(\HH_{A_1\ldots A_m})$ with finite entropy $S(\rho)<\infty$, it holds that
\bb
E_{R,\, \pi} (\rho_{A_1\ldots A_m}) = \sup_{X = X^\dag\in \B(\HH_{A_1\ldots A_m})} \left\{ \Tr \rho\, X - \sup_{\sigma\in \SEP_{\!A_1\ldots A_m}^\pi} \ln \Tr e^{\ln \sigma + X} \right\} .
\label{variational_relent_alpha}
\ee
\end{cor}

\subsection{Relative entropy of NPT entanglement} \label{NPT_entanglement_subsec}

Another historically important upper bound to the distillable entanglement is the relative entropy distance to the set of states with a positive partial transpose (PPT). Let $\HH_{AB}=\HH_A\otimes \HH_B$ be a bipartite Hilbert space, and select a preferred basis on $\HH_B$ with respect to which one considers the transposition $\intercal$. Then, the \deff{partial transposition} on $\HH_{AB}$ is the operation $\Gamma:\T(\HH_{AB})\to \B(\HH_{AB})$, where $\B(\HH_{AB})$ is the space of bounded operators on $\HH_{AB}$, defined by~\cite{PeresPPT}
\bb
\Gamma (X_A\otimes Y_B) = \left( X_A\otimes Y_B\right)^\Gamma \coloneqq X_A\otimes Y_B^\intercal\, ,
\ee
and extended by linearity and continuity to the whole $\T(\HH_{AB})$.\footnote{The partial transpose of a trace class operator may not be of trace class anymore. However, the map $\Gamma:\T(\HH_A)\otimes \T(\HH_B)\to \B(\HH_{AB})$ turns out to be bounded with respect to the trace norm on the input space $\T(\HH_A)\otimes \T(\HH_B)\subseteq \T(\HH_{AB})$ of finite linear combinations of simple tensors and to the operator norm on the output space --- in fact, the norm of $\Gamma$ is $1$. Hence, it admits a continuous extension to the whole $\T(\HH_{AB})$.} The set of \deff{PPT states} on $AB$ is defined by
\bb
\PPT_{\!AB} \coloneqq \left\{\sigma_{AB}\in \D(\HH_{AB}):\, \sigma_{AB}^\Gamma\geq 0 \right\} .
\label{PPT_states}
\ee
PPT entangled states that are not separable have been the first --- and, to date, only~\cite{Horodecki-open-problems} --- examples of entangled states that are undistillable, meaning that their distillable entanglement vanishes~\cite{Horodecki-PPT-entangled, HorodeckiBound, Bruss2000}. This is not to say that they are easy to generate; in fact, they have positive entanglement cost~\cite{faithful-EC}, become distillable if the set of protocols available is enlarged to include all `PPT operations'~\cite{Eggeling2001}, and can even have an almost maximal `Schmidt number'~\cite{SchmidtNumber, PPT-high-SN, Cariello2020}.

The \deff{relative entropy of NPT entanglement} of an arbitrary state $\rho_{AB}\in \D(\HH_{AB})$ is simply the relative entropy of resource associated to the set~\eqref{PPT_states}, in formula~\cite{Rains1999, Audenaert2001}
\bb
E_{R,\, \PPT}(\rho_{AB})\coloneqq D_{\PPT_{\!AB}}(\rho_{AB}) = \inf_{\sigma_{AB}\in \D(\HH_{AB}),\, \sigma_{AB}^\Gamma\geq 0} D(\rho_{AB}\|\sigma_{AB})\, .
\label{relent_NPT_entanglement}
\ee
This quantity is a generally sharper upper bound to the distillable entanglement than the standard relative entropy of entanglement~\cite{Rains1999}, and it has the advantage of being often easier to compute, even upon regularisation~\cite{Audenaert2001}.

The cone generated by PPT states turns out to obey the hypotheses of Theorem~\ref{achievable_relent_thm}.

\begin{lemma} \label{PPT_w*_closed}
For any two quantum systems $A,B$ with separable Hilbert spaces, $\cone(\PPT_{\!AB}) = \left\{ X:\, X\geq 0,\, X^\Gamma\geq 0\right\}$ is weak*-closed.
\end{lemma}

\begin{proof}
This can be seen as a corollary to Theorem~\ref{w*_closed_condition_thm}. However, it is even easier to give a direct proof of this fact. Taking local orthonormal bases $\{\ket{p}_A\}_{p\in \N}$ and $\{\ket{q}_B\}_{q\in \N}$ of $\HH_A$ and $\HH_B$, respectively, let us say that a bounded operator $Y$ on $\HH_{AB}=\HH_A\otimes \HH_B$ has a finite expansion if $\braket{p,q|Y|p',q'}\neq 0$ for a finite number of quadruples $(p,p',q,q')\in \N^4$. It is not difficult to verify that a trace class operator $X$ satisfies $X^\Gamma\geq 0$ if and only if $\Tr X Y^\Gamma = \Tr X^\Gamma Y\geq 0$ for all $Y\geq 0$ with a finite expansion. For an arbitrary such $Y$, by definition of $\Gamma$ we see that also $Y^\Gamma$ has a finite expansion and is thus compact. We deduce straight away that the cone of operators $X$ with $X^\Gamma\geq 0$ is weak*-closed. Taking the intersection with the cone of positive semi-definite trace class operators, which is also weak*-closed, we obtain precisely $\cone(\PPT_{\!AB})$, which is then weak*-closed as well.
\end{proof}

Thanks to Lemma~\ref{PPT_w*_closed} we can immediately apply Theorems~\ref{achievable_relent_thm} and~\ref{variational_thm}, which give:

\begin{cor} \label{relent_NPT_entanglement_cor}
For an arbitrary bipartite system with separable Hilbert space, the relative entropy of NPT entanglement is:
\begin{enumerate}[(a)]
    \item always achieved, meaning that for all $\rho_{AB}\in \D(\HH_{AB})$ there exists a state $\sigma_{AB}\in \PPT_{\!AB}$ such that $E_{R,\, \PPT}(\rho_{AB}) = D\left(\rho_{AB}\|\sigma_{AB}\right)$; and
    \item lower semi-continuous with respect to the trace norm topology.
\end{enumerate}
Moreover, for every state $\rho=\rho_{AB}\in \D(\HH_{AB})$ with finite entropy $S(\rho)<\infty$, it holds that
\bb
E_{R,\, \PPT}(\rho_{AB}) = \sup_{X = X^\dag\in \B(\HH_{AB})} \left\{ \Tr \rho\, X - \sup_{\sigma \in \PPT_{\!AB}} \ln \Tr e^{\ln \sigma + X} \right\} .
\label{variational_relent_NPT}
\ee
\end{cor}

As it turns out, one can modify~\eqref{relent_NPT_entanglement} to give an even better upper bound to the distillable entanglement, namely, the \deff{Rains bound}. For an arbitrary bipartite state $\rho_{AB}$, this is defined by~\cite{Rains1999, Rains2001, Audenaert2002, nonadditivity-Rains}
\bb
R(\rho_{AB}) \coloneqq \inf_{\substack{\sigma_{AB}\geq 0, \\ \left\|\sigma_{AB}^\Gamma\right\|_1\leq 1}} \left\{ D(\rho_{AB}\|\sigma_{AB}) + 1 - \Tr\sigma_{AB} \right\} .
\label{Rains_bound}
\ee
Note that in the right-hand side of~\eqref{Rains_bound} the operator $\sigma_{AB}$ need not be normalised, although it is implicitly assumed to be of trace class. One can anyway see that it must be at least sub-normalised, meaning that $\Tr \sigma_{AB} = \Tr \sigma_{AB}^\Gamma\leq \left\|\sigma_{AB}^\Gamma\right\|_1\leq 1$.

The Rains bound is not of the general form in~\eqref{relent_F}. However, it is sufficiently alike to it that we can hope to employ similar techniques to show that it is achieved and lower semi-continuous. In fact, the set of operators $\sigma_{AB}$ over which the optimisation in~\eqref{Rains_bound} runs can be shown to be weak*-compact. And yet, remarkably, we \emph{cannot} deduce a statement similar to Theorem~\ref{achievable_relent_thm} in this case, because the function to be minimised in~\eqref{Rains_bound} fails to be lower semi-continuous, due to the negative trace term.\footnote{The trace function is known to be lower semi-continuous but generally not continuous with respect to the weak*-topology on the cone of positive semi-definite operators.}

\subsection{Relative entropy of non-classicality, Wigner non-positivity, and generalisations thereof} \label{lambda_negativity_subsec}

We now move on to resource theories specific to continuous variable systems. An \deff{$\boldsymbol{m}$-mode continuous variable system}~\cite{HOLEVO, HOLEVO-CHANNELS-2, BARNETT-RADMORE, BUCCO} is just a finite collection of $m$ harmonic oscillators with canonical operators $x_1,p_1,\ldots, x_m,p_m$. These obey the canonical commutation relations $[x_j,x_k]=0=[p_j,p_k]$ and $[x_j,p_k]=i\delta_{jk} \id$ on the underlying Hilbert space $\HH_m\coloneqq L^2(\R^m)$ of square-integrable complex-valued functions on $\R^m$. Defining the `vector of operators' $r\coloneqq (x_1,\ldots, x_m,p_1,\ldots, p_m)^\intercal$, we can rephrase them in the matrix form 
\begin{equation}
[r,r^\intercal] = i\Omega\, ,\qquad \Omega \coloneqq \begin{pmatrix} 0_m & \id_m \\ -\id_m & 0_m \end{pmatrix} .
\label{CCR}
\end{equation}
For a vector $\xi \in \R^{2m}$, the associated \deff{Weyl operator} is the unitary defined by $\W_\xi \coloneqq e^{- i \xi^\intercal \Omega r}$, while the corresponding \deff{coherent state} is~\cite{Schroedinger1926-coherent}
\begin{equation} \label{coherent_state}
\ket{\xi} \coloneqq \W_\xi \ket{0}\, ,
\end{equation}
with $\ket{0}$ being the vacuum state. The operators $\W_\xi$ satisfy the important identity
\bb
\W_{\xi_1} \W_{\xi_2} = e^{-\frac{i}{2}\xi_1^\intercal \Omega \xi_2} \W_{\xi_1+\xi_2}\, ,
\label{Weyl_CCR}
\ee
known as the Weyl form of the canonical commutation relations. Taking the trace against $\W_\xi$ yields a representation of any trace class operator $T\in \T(\HH_m)$ as a `characteristic function'. Here we are especially interested in a slight generalisation of this notion. For a parameter $\lambda\in [-1,1]$, let us define the \deff{$\boldsymbol{\lambda}$-ordered characteristic function} $\chi_{T,\lambda}$ of an operator $T\in \T(\HH_m)$ by
\bb
\chi_{T,\lambda}(\xi) \coloneqq \Tr \left[T \W_\xi\right] e^{-\frac\lambda4 \|\xi\|^2} .
\label{chi_lambda}
\ee
An important property of $\chi_{T,\lambda}$ is that it characterises the trace class operator $T$ completely: namely, \emph{$\chi_{T_1,\lambda}(\xi) = \chi_{T_2,\lambda}(\xi)$ for all $\xi$ if and only if $T_1=T_2$}~\cite[p.~199]{HOLEVO}.

If $\lambda\geq 0$, the $\lambda$-ordered characteristic function is guaranteed to be square-integrable; however, it may or may not be such if $\lambda<0$. In any case, if $\chi_{T,\lambda}$ happens to be square-integrable then its Fourier transform gives the \deff{$\boldsymbol{\lambda}$-quasi-probability distribution} of $T$, in formula
\bb
W_{T,\lambda}(u) \coloneqq \frac{1}{2^m \pi^{2m}} \int d^{2m}\xi\, \chi_{T,\lambda}(\xi)\, e^{-i u^\intercal \Omega \xi}\, .
\label{W_lambda}
\ee

If $T=\rho$ is a quantum state and the function $W_{\rho,\lambda}$ is well defined, it is not difficult to show that it is also real-valued. It is usually referred to as a quasi-probability distribution because in general it can take on negative values. In fact, the negativity of $W_{\rho,\lambda}$ (for $\lambda<1$) is an unmistakable signature of quantumness. Accordingly, it turns out that substantial physical insight can be gained by considering the family of resource theories whose free states are those with a non-negative $W_{\rho,\lambda}$. Defining this notion rigorously may appear problematic at first sight, because for a given $\rho\in \D(\HH_m)$ the object $W_{\rho,\lambda}$ may be ill-defined as a function --- this is the case if $\chi_{\rho,\lambda}$ in~\eqref{chi_lambda} is not square-integrable --- and thus it may be unclear how to check that $W_{\rho,\lambda}\geq 0$. While this problem can be overcome by resorting to the theory of distributions, we prefer to take a simpler route here. All we need is the following notion.

\begin{Def}[{\cite[Chapter~5]{BHATIA}}] \label{positive_definite}
A function $f:\R^s \to \C$ is called \deff{positive definite} if for all positive integers $n\in \N_+$ and all choices of $\xi_1,\ldots, \xi_n\in \R^s$, it holds that
\bb
\sum_{\mu,\nu=1}^n f(\xi_\mu-\xi_\nu) \ketbraa{\mu}{\nu} \geq 0\, ,
\ee
i.e.\ if the $n\times n$ complex matrix on the left-hand side turns out to be positive semi-definite. Here, $\ket{\mu}$ denotes the $\mu^{\text{th}}$ vector of the canonical basis of $\C^n$.
\end{Def}

The connection between this notion and that of non-negativity of the $\lambda$-quasi-probability distribution is captured by Bochner's theorem~\cite[Theorem~1.8.9]{USHAKOV}: \emph{a function $f:\R^s\to \C$ is the Fourier transform of a measure on $\R^s$ if and only if it is continuous at $0$ and positive definite. In this case, the measure is in fact a probability measure if and only if $f(0)=1$.}

We are now ready to define rigorously our object of interest.

\begin{Def}
For $\lambda\in [-1,1]$ and a positive integer $m$, the set $\PP_{m,\lambda}$ of free states of the \deff{resource theory of $\boldsymbol{\lambda}$-negativity} on $m$ modes is defined by
\bb
\PP_{m,\lambda} \coloneqq \left\{ \sigma\in \D(\HH_m):\, \text{$\chi_{\sigma,\lambda}$ is positive definite} \right\} ,
\label{lambda_positive_states}
\ee
where $\chi_{\sigma,\lambda}$ is given by~\eqref{chi_lambda}, and the notion of positive definiteness for functions is explained in Definition~\ref{positive_definite}.
\end{Def}

\begin{rem} \label{explicit_condition_PP_m_lambda_rem}
An alternative way to rephrase~\eqref{lambda_positive_states} is as follows:
\bb
\PP_{m,\lambda} = \left\{ \sigma\in \D(\HH_m):\, \Tr \left[ \sigma \sum_{\mu,\nu=1}^n c_\mu^* c_\nu e^{-\frac{\lambda}{4}\|\xi_\mu-\xi_\nu\|^2} \W_{\xi_\mu-\xi_\nu} \right] \geq 0\quad \forall\, n\in \N,\, c\in \C^n,\, \xi_1,\ldots, \xi_n\in \R^{2m} \right\} .
\label{alternative_lambda_positive_states}
\ee
\end{rem}

\begin{rem} \label{Bochner_connection_rem}
Thanks to the aforementioned Bochner's theorem, $\PP_{m,\lambda}$ (respectively, $\cone(\PP_{m,\lambda})$) is easily seen to contain \emph{precisely} those quantum states (respectively, positive trace class operators) $\sigma$ for which $\chi_{\sigma,\lambda}$ is the Fourier transform of a probability measure (respectively, a measure) on $\R^{2m}$. In fact, due to the strong continuity of the map $\R^{2m}\ni \xi\mapsto \W_\xi$, in turn a consequence of Stone's theorem~\cite[Theorem~10.15]{HALL}, the $\lambda$-ordered characteristic function is always continuous, and moreover it satisfies $\chi_{\sigma,\lambda}(0)=\Tr \sigma =1$.
\end{rem}

Three special cases of particular physical significance, corresponding to as many values of $\lambda$, are as follows:
\begin{enumerate}
    \item $\lambda=1$: $W_{\rho,1}$ is called the \deff{Husimi function} of $\rho$~\cite{Husimi}; it has the special representation 
    \bb
    W_{\rho,1}(u)=\frac{1}{\pi^{m}}\, \braket{u|\rho|u}\, ,
    \label{Husimi}
    \ee
    where $\ket{u}$ is a coherent state~\eqref{coherent_state}. Hence, $W_{\rho,1}\geq 0$ for all quantum states, i.e.\ $\PP_{m,1} = \D(\HH_m)$.
    \item $\lambda=0$: $W_{\rho,0}$ is known as the \deff{Wigner function} of $\rho$~\cite{Wigner, Grossmann1976, Hillery1984}. The set $\PP_{m,0}$ of states with a non-negative Wigner function~\cite{Hudson1974, Hudson-thm-multimode, Broecker1995}, together with the associated quantum resource theory~\cite{Albarelli2018, Tan2020, CV-data-hiding, KK-VV-GOCC} have been the subject of intense study.
    \item $\lambda=-1$: $W_{\rho,-1}$ is referred to as the \deff{Glauber--Sudarshan $\boldsymbol{P}$-function} of $\rho$~\cite{Glauber1963, Sudarshan1963}. The set $\PP_{m,-1}$ in this case can be described more compactly as the (trace norm) closed convex hull of the set of coherent states~\eqref{coherent_state}, in formula~\cite{Bach1986}
    \bb
    \PP_{m,-1} = \cl_{\mathrm{tn}}\left( \co\left\{ \ketbra{u}:\, u\in \R^{2m} \right\} \right) .
    \label{classical_states}
    \ee
    The resource theory of optical non-classicality, in which the free states are those in~\eqref{classical_states}, has garnered a lot of attention in the last few decades~\cite{Sperling2015, Tan2017, Yadin2018, NC-review, taming-PRA, taming-PRL, nonclassicality}.
\end{enumerate}

In general, we define the \deff{relative entropy of $\boldsymbol{\lambda}$-negativity} of an arbitrary $m$-mode state $\rho\in \D(\HH_m)$ as the relative entropy of resource associated to the set $\PP_{m,\lambda}$, in formula
\bb
N_{R,\lambda}(\rho)\coloneqq D_{\PP_{m,\lambda}}(\rho) = \inf_{\sigma \in \PP_{m,\lambda}} D(\rho\|\sigma)\, .
\label{relent_lambda_negativity}
\ee
For the special case $\lambda=-1$, this quantity has been considered in~\cite{nonclassicality}, where it is employed to give upper bounds to transformation rates in the resource theory of non-classicality.

In order to apply Theorems~\ref{achievable_relent_thm} and~\ref{variational_thm} to the relative entropy of $\lambda$-negativity, we must first prove the weak*-closedness of the cone generated by $\PP_{m,\lambda}$. This has been established in Ref.~\cite[Lemma~38]{nonclassicality} for the case $\lambda=-1$, but all other cases (except for the trivial one $\lambda=1$) have not been studied elsewhere, to the best of our knowledge. In order to apply Theorem~\ref{w*_closed_condition_thm} to the case at hand, we need a preliminary result. Let us fix some terminology first. Denote with
\bb
N\coloneqq \sum_{j=1}^m \frac{x_j^2+p_j^2-1}{2}
\label{total_photon_number}
\ee
the `total photon number' Hamiltonian on $\HH_m=L^2(\R^m)$. This can be diagonalised as $N=\sum_{k_1,\ldots, k_m\in \N} (k_1+\ldots +k_m) \ketbra{k_1}_1\otimes \ldots \otimes \ketbra{k_m}_m$, where $\ket{k}_j$ is the $k^\text{th}$ `Fock state' on the $j^\text{th}$ mode. Accordingly, for some $\eta\in [-1,1]$ we will set
\bb
\eta^N \coloneqq \sum_{k_1,\ldots, k_m\in \N} \eta^{k_1+\ldots +k_m} \ketbra{k_1}_1\otimes \ldots \otimes \ketbra{k_m}_m = \left( \sumno_{k\in \N} \eta^k \ketbra{k}\right)^{\otimes m}\, ,
\ee
with the convention that $0^0=1$.

\begin{lemma} \label{sandwiching_lemma}
For an arbitrary $\eta\in [-1,1]$, $\lambda\in [-1,1]$, and $\sigma\in \PP_{m,\lambda}$, we have that
\bb
\frac{\eta^N\sigma\, \eta^N}{\Tr\left[\eta^N\sigma\, \eta^N\right]} \in \PP_{m,\lambda}\, .
\ee
\end{lemma}

\begin{proof}
The case $\eta=1$ is trivial, while the claim for $\eta=-1$ follows from the simple observation that --- since $(-1)^N$ is the parity operator --- we have that $\chi_{(-1)^N \sigma\, (-1)^N\!,\,\lambda}(\xi) = \chi_{\sigma,\lambda}(-\xi)$ is positive definite if $\chi_{\sigma,\lambda}$ is such. If $\eta=0$ then $\eta^N=\ketbra{0}$, where $\ket{0}=\ket{0}_1\otimes \ldots \otimes \ket{0}_m$ is the multi-mode `vacuum state'. Since $\ketbra{0}\in \PP_{m,\lambda}$ for every $\lambda\in [-1,1]$, also the case $\eta=0$ is easily dealt with.

From now on, we assume that $\eta\neq -1,0,1$. Employing the expression for the ($0$-ordered) characteristic function of a thermal state for the Hamiltonian $N$ and mean photon number $\frac{\eta}{1-\eta}$ (see e.g.~\cite[Eq.~(4.48)--(4.49)]{BUCCO}), one derives the representation
\bb
\eta^N = \frac{1}{(1-\eta)^m} \int \frac{d^{2m}\xi}{(2\pi)^m}\, e^{-\frac14 \frac{1+\eta}{1-\eta}\, \|\xi\|^2} \W_\xi\, ,
\label{sandwiching_proof_eq1}
\ee
where the integral is in the Bochner sense with respect to the operator norm. Using this twice and leveraging also the Weyl form of the canonical commutation relations~\eqref{Weyl_CCR}, for an arbitrary $\xi\in \R^{2m}$ we obtain that 
\bb
\,\! &\Tr \left[ \eta^N\! \sigma\, \eta^N \W_\xi\right] \\
&\qquad = \frac{1}{(1-\eta)^{2m}} \int \frac{d^{2m}u}{(2\pi)^m} \frac{d^{2m}v}{(2\pi)^m} \exp \left[ - \frac14 \frac{1+\eta}{1-\eta} \left( \|u\|^2 + \|v\|^2\right) - \frac{i}{2}(u-v)^\intercal \Omega \xi - \frac{i}{2} u^\intercal \Omega v \right] \chi_\sigma(u+v+\xi) \\
&\qquad = \frac{1}{(1-\eta)^{2m}} \int \frac{d^{2m}w}{(2\pi)^m} \frac{d^{2m}z}{(2\pi)^m} \exp \left[ - \frac14 \frac{1+\eta}{1-\eta} \left( \frac{\|w\|^2}{2} + 2\|z\|^2\right) + i \left(\xi +\frac{w}{2}\right)^\intercal \Omega z \right] \chi_\sigma(w+\xi) \\
&\qquad = \frac{1}{(1-\eta^2)^m} \int \frac{d^{2m}w}{(2\pi)^m} \exp \left[ - \frac18 \frac{1+\eta}{1-\eta} \|w\|^2 - \frac18 \frac{1-\eta}{1+\eta} \left\|w+2\xi\right\|^2 \right] \chi_\sigma(w+\xi)\, .
\label{sandwiching_proof_eq2}
\ee
In the above calculations, we implicitly used the Fubini--Tonelli theorem, applicable because $\chi_\sigma$ is bounded in modulus by $1$ and thanks to the absolute integrability of the Gaussians, and performed the change of variable $w\coloneqq u+v$ and $z\coloneqq \frac{u-v}{2}$.

Now, assume that $\sigma\in \PP_{m,\lambda}$. As per the above discussion, thanks to Bochner's theorem~\cite[Theorem~1.8.9]{USHAKOV} there exists a probability measure $\mu$ on $\R^{2m}$ such that
\bb
\chi_\sigma(x)\, e^{-\frac{\lambda}{4}\|x\|^2} = \chi_{\sigma,\lambda}(x) = \int d\mu(y)\, e^{-i x^\intercal \Omega y}\, .
\label{sandwiching_proof_eq3}
\ee
Plugging this representation into~\eqref{sandwiching_proof_eq2} and using once again the Fubini--Tonelli theorem to swap the integrals yields
\bb
\chi_{\eta^N\! \sigma\, \eta^N,\, \lambda}(\xi) &= \Tr \left[ \eta^N\! \sigma\, \eta^N \W_\xi\right] e^{-\frac{\lambda}{4}\, \|\xi\|^2} \\
&= \frac{1}{(1-\eta^2)^m} \int \frac{d^{2m}w}{(2\pi)^m} e^{- \frac18 \frac{1+\eta}{1-\eta} \|w\|^2 - \frac18 \frac{1-\eta}{1+\eta} \left\|w+2\xi\right\|^2} e^{\frac{\lambda}{4} \left(\|w+\xi\|^2 - \|\xi\|^2\right)} \int d\mu(y)\, e^{-i (w+\xi)^\intercal \Omega y} \\
&= \frac{1}{(1-\eta^2)^m} \int d\mu(y)\, e^{-i\xi^\intercal \Omega y} \int \frac{d^{2m}w}{(2\pi)^m} e^{- \frac18 \frac{1+\eta}{1-\eta} \|w\|^2 - \frac18 \frac{1-\eta}{1+\eta} \left\|w+2\xi\right\|^2} e^{\frac{\lambda}{4} \left(\|w+\xi\|^2 - \|\xi\|^2\right)} \, e^{-i w^\intercal \Omega y} \\
&= \frac{1}{(1-\eta^2)^m} \int d\mu(y)\, e^{-i\xi^\intercal \Omega y} \int \frac{d^{2m}w}{(2\pi)^m} e^{-\frac14 \left( \frac{1+\eta^2}{1-\eta^2} - \lambda \right) \|w\|^2 - w^\intercal \left( i \Omega y + \frac12 \left(\frac{1-\eta}{1+\eta} - \lambda \right) \xi \right) - \frac12 \frac{1-\eta}{1+\eta}\|\xi\|^2} \\
&= \frac{1}{(1-\eta^2)^m} \int d\mu(y)\, e^{-i\xi^\intercal \Omega y} \, \frac{1}{(2\pi)^m} \frac{(4\pi)^m}{\left(\frac{1+\eta^2}{1-\eta^2} - \lambda\right)^m} \\
&\hspace{4ex} \cdot \exp\left[-\frac12 \frac{1\!-\!\eta}{1\!+\!\eta}\|\xi\|^2 + \frac{1}{\frac{1+\eta^2}{1-\eta^2} - \lambda} \left( -\|y\|^2 + \frac14 \left( \frac{1\!-\!\eta}{1\!+\!\eta}-\lambda\right)^2 \|\xi\|^2 + i \left( \frac{1\!-\!\eta}{1\!+\!\eta}-\lambda\right) \xi^\intercal \Omega y \right) \right] \\
&= \frac{2^m}{\left( 1 - \lambda + (1+\lambda)\eta^2 \right)^m} \\
&\hspace{4ex}\cdot \int d\mu(y)\, \exp\left[ -\frac{1}{1 - \lambda + (1+\lambda)\eta^2} \left( 2i \eta\, \xi^\intercal \Omega y -\frac14 (1-\eta^2)(1-\lambda^2) \|\xi\|^2 + (1-\eta^2) \|y\|^2 \right) \right]
\label{sandwiching_proof_eq3.5}
\ee
\bb
&= e^{- \frac14 \frac{(1-\eta^2)(1-\lambda^2)}{1 - \lambda + (1+\lambda)\eta^2}\, \|\xi\|^2}  \left( \frac{1 - \lambda + (1+\lambda)\eta^2}{2\eta^2} \right)^m \int d\mu'(y')\, e^{-i\xi^\intercal \Omega y' - \frac{1}{16\eta^2}(1-\eta^2) \left( 1 - \lambda + (1+\lambda)\eta^2\right) \|y'\|^2}\, ,
\label{sandwiching_proof_eq4}
\ee
where in the last line we introduced the change of variables
\bb
y'\coloneqq \frac{2\eta}{1 - \lambda + (1+\lambda)\eta^2}\, y\, ,\qquad d\mu'(y') \coloneqq \left( \frac{2\eta}{1 - \lambda + (1+\lambda)\eta^2}\right)^{2m} d\mu \left( \frac{1 - \lambda + (1+\lambda)\eta^2}{2\eta}\, y \right) .
\ee
We now claim that from~\eqref{sandwiching_proof_eq3.5}--\eqref{sandwiching_proof_eq4} it is quite clear that $\chi_{\eta^N\! \sigma\, \eta^N\!,\, \lambda}$ is the Fourier transform of a measure on $\R^{2m}$, which allows us to conclude the proof thanks to Remark~\ref{Bochner_connection_rem}. Indeed, $\chi_{\eta^N\! \sigma\, \eta^N\!,\, \lambda}$ is written as a point-wise product of a Gaussian\footnote{Under the current assumption that $\eta\neq -1,0,1$, we have that $\frac{(1-\eta^2)(1-\lambda^2)}{1-\lambda+(1+\lambda)\eta^2}>0$ if $\lambda\neq \pm 1$. The cases where $\lambda=\pm 1$ are straightforward to address.} and the Fourier transform of a measure --- namely, $\mu'$ re-scaled by another Gaussian factor. Since the point-wise product of Fourier transforms is the Fourier transform of the convolution~\cite[Proposition~2.3.22(11)]{GRAFAKOS}, and a Gaussian is the Fourier transform of another Gaussian, we see that indeed $\chi_{\eta^N\! \sigma\, \eta^N\!,\, \lambda}$ is the Fourier transform of a measure --- namely, that obtained by convolving a Gaussian re-scaled version of $\mu'$ by another Gaussian.
\end{proof}

We are now ready to deduce:

\begin{cor} \label{P_lambda_weak*_closed}
For all positive integers $m$, the cone $\cone(\PP_{m,\lambda})$ generated by the set~\eqref{lambda_positive_states} of $\lambda$-positive states is weak*-closed.
\end{cor}

\begin{proof}
For $n\in \N$, let us construct the compact operator $M_n = \big(\frac{n}{n+1}\big)^N$, where $N$ is the total photon number Hamiltonian~\eqref{total_photon_number}. It is immediate to verify that $M_n$ converges to the identity in the strong operator topology, in the sense explained in the statement of Theorem~\ref{w*_closed_condition_thm}. Also, $\PP_{m,\lambda}$ is easily verified to be convex and trace norm closed --- this latter fact is apparent once one notices that the operators $\sum_{\mu,\nu=1}^n c_\mu^* c_\nu e^{-\frac{\lambda}{4}\|\xi_\mu-\xi_\nu\|^2} \W_{\xi_\mu-\xi_\nu}$ in~\eqref{alternative_lambda_positive_states} are finite linear combinations of unitary operators and hence bounded. The last missing condition needed to apply Theorem~\ref{w*_closed_condition_thm} with $\FF=\PP_{m,\lambda}$ is~(iii), i.e.\ that conjugation by $\big(\frac{n}{n+1}\big)^N$ preserves $\cone(\PP_{m,\lambda})$. This follows directly from Lemma~\ref{sandwiching_lemma}.
\end{proof}

Thanks to the weak*-closedness of the cone $\PP_{m,\lambda}$ we are now in position to apply Theorems~\ref{achievable_relent_thm} and~\ref{variational_thm}, yielding the following.

\begin{cor} \label{relent_lambda_negativity_cor}
The relative entropy of $\lambda$-negativity, defined in~\eqref{relent_lambda_negativity}, is:
\begin{enumerate}[(a)]
    \item always achieved, meaning that for all $\rho\in \D(\HH_m)$ there exists a state $\sigma\in \PP_{m,\lambda}$ such that $N_{R,\lambda} (\rho) = D\left(\rho\|\sigma\right)$; and
    \item lower semi-continuous with respect to the trace norm topology.
\end{enumerate}
Moreover, for every state $\rho\in \D(\HH_m)$ with finite entropy $S(\rho)<\infty$, it holds that
\bb
N_{R,\lambda}(\rho) = \sup_{X = X^\dag\in \B(\HH_m)} \left\{ \Tr \rho\, X - \sup_{\sigma \in \PP_{m,\lambda}} \ln \Tr e^{\ln \sigma + X} \right\} .
\label{variational_relent_lambda_negativity}
\ee
\end{cor}

\subsection{Relative entropy of non-Gaussianity} \label{non_Gaussianity_subsec}

In an $m$-mode continuous variable system, particularly simple yet experimentally relevant states are the so-called \deff{Gaussian states}. A state $\sigma\in \D(\HH)$ is said to be Gaussian if any of its $\lambda$-ordered characteristic functions $\chi_{\sigma,\lambda}$ is a Gaussian (and hence all are). Let us consider for instance the case $\lambda=0$. Since $\chi_{\sigma,0}$ achieves it maximum modulus at $0$ (see~\cite[Proposition~14]{QCLT} and~\cite[Lemma~10]{G-dilatable}), if it is a Gaussian it must be centred, i.e.~\cite[Eq.~(4.48)]{BUCCO}
\bb
\chi_{\sigma,0}(\xi) = e^{-\frac14 \xi^\intercal \Omega^\intercal V\Omega \xi + i s^\intercal \Omega \xi}\, ,
\label{chi_0_Gaussian}
\ee
where the real vector $s = \Tr [\sigma\, r] \in \R^{2m}$ and the $2m\times 2m$ real matrix $V = \Tr \left[ \sigma\, \{ r-s,(r-s)^\intercal\} \right]$ represent the first and second moments of the state, respectively. We will denote with $\GG_m$ the set of $m$-mode Gaussian states. Unlike all other sets of free states considered so far, \emph{$\GG_m$ is not convex.}

The \deff{relative entropy of non-Gaussianity} can be defined as the relative entropy of resource $D_\FF$ corresponding to the choice of free states $\FF=\GG_m$, in formula~\cite{Genoni2008, Marian2013}
\bb
\nonG(\rho) \coloneqq D_{\GG_m}(\rho) = \inf_{\sigma \in \GG_m} D(\rho\|\sigma)\, .
\label{relent_non_Gaussianity}
\ee
In Ref.~\cite{Marian2013} it was shown that \emph{if $\rho$ has well-defined second moments}, a condition that we equivalently rephrase by requiring that $\Tr \rho N<\infty$ for $N$ the total photon number Hamiltonian~\eqref{total_photon_number}, then the optimisation in~\eqref{relent_non_Gaussianity} is achieved at the Gaussian state $\sigma = \rho_\G$ with the same first\footnote{A state with well-defined second moments has also well-defined first moments and a finite entropy.} and second moments as those of $\rho$, hereafter called the \deff{Gaussification} of $\rho$, i.e.
\bb
\nonG(\rho) = D(\rho\|\rho_\G) = S(\rho_\G) - S(\rho)\, .
\label{relent_non_Gaussianity_closed_form}
\ee

To analyse this object more effectively, we now recall an alternative characterisation of $\GG_m$ that makes this set easier to work with. To this end, let us introduce the unitary modelling a 50:50 beam splitter acting on a bipartite quantum system $AB$, where $A,B$ are composed of $m$ mode each. This can be defined by
\bb
U \coloneqq \exp\left[\frac{\pi}{4}\sumno_{j=1}^m \left(x_j^A p_j^B - p_j^A x_j^B\right) \right] ,
\label{beam_splitter}
\ee
where $x_j^A,p_k^A$ are the canonical operators corresponding to system $A$, and analogously for $B$. The action of the beam splitter unitary on a tensor product of coherent states $\ket{u,v}=\ket{u}_A\otimes \ket{v}_B = \W_u \ket{0}_A \otimes \W_v\ket{0}_B$ (see~\eqref{coherent_state}) can be expressed as
\bb
U\ket{u,v} = \Ket{\frac{u+v}{\sqrt2},\, \frac{-u+v}{\sqrt2}}\, ,\qquad U^\dag\ket{u,v} = \Ket{\frac{u-v}{\sqrt2},\, \frac{u+v}{\sqrt2}}\, .
\label{beam_splitter_coherent}
\ee
Having established the notation, we now report a (marginally simplified) version of an interesting result by Cuesta~\cite{Cuesta2020}.

\begin{lemma}{{(Quantum Darmois--Skitovich theorem~\cite[Theorem~7]{Cuesta2020})}} \label{QDS_lemma}
A trace class operator $X\in \T(\HH_m)$ with $X\geq 0$ is a multiple of a Gaussian state if and only if
\bb
U(X_A\otimes X_B) U^\dag = X_A\otimes X_B\, ,
\ee
where $A,B$ stand for two $m$-mode systems, and $U$ is the beam splitter unitary~\eqref{beam_splitter}.
\end{lemma}

We are now ready to prove the following generalisation of the result in Ref.~\cite[Lemma~1 in Appendix~A]{G-resource-theories} establishing the trace norm closedness of $\GG_m$.

\begin{lemma} \label{w*_closed_cone_G_lemma}
The cone generated by $m$-mode Gaussian states, $\cone(\GG_m)$, is weak*-closed.
\end{lemma}

\begin{proof}
Since $\GG_m$ is not convex, we cannot hope to apply Theorem~\ref{w*_closed_condition_thm}. Hence, we have to proceed differently in this case. Consider a general net\footnote{A net on a set $\pazocal{X}$ is any function $f:\pazocal{A}\to \pazocal{X}$, where $\pazocal{A}$ is an arbitrary directed set, i.e.\ a set equipped with a pre-order $\leq$ such that any two elements $a,b\in \pazocal{A}$ admit a common upper bound.} $(X_\alpha)_\alpha$ on $\cone(\GG_m)$ --- hence, a generic $X_\alpha$ can be written as $X_\alpha = \lambda_\alpha \sigma_\alpha$, where $\lambda_\alpha\geq 0$ and $\sigma_\alpha$ is a Gaussian state. Assume that $X_\alpha\tends{w*}{\alpha} X\in \T(\HH)$, and let us show that $X\in \cone(\GG_m)$ as well.\footnote{Here, $X_\alpha\tends{w*}{\alpha} X$ just means that for all $\varepsilon>0$ and all compact operators $K\in \K(\HH_m)$, we can find $\alpha_0$ such that $\left| \Tr X_\alpha K - \Tr XK\right|\leq \varepsilon$ for all $\alpha\geq \alpha_0$.} For arbitrary $u,v,u',v'\in \R^{2m}$, denoting with $\ket{u}$, $\ket{v}$, etc.\ the coherent states~\eqref{coherent_state} and with $\ket{u,v}=\ket{u}\otimes \ket{v}$ their tensor products, we have that
\bb
\braket{u,v|(X\otimes X)|u',v'} \,&=\, \braket{u|X|u'} \braket{v|X|v'} \\
&\eqt{(i)} \left( \lim_\alpha \braket{u|X_\alpha|u'}\right) \left( \lim_\alpha \braket{v|X_\alpha|v'}\right) \\
&=\, \lim_\alpha \braket{u,v|(X_\alpha\otimes X_\alpha)|u',v'} \\
&\eqt{(ii)}\, \lim_\alpha \braket{u,v|U(X_\alpha\otimes X_\alpha)U^\dag |u',v'} \\
&\eqt{(iii)}\, \lim_\alpha \Braket{\frac{u-v}{\sqrt2},\, \frac{u+v}{\sqrt2}\bigg| \left(X_\alpha \otimes X_\alpha\right)\bigg| \frac{u'-v'}{\sqrt2},\, \frac{u'+v'}{\sqrt2}} \\
&=\, \lim_\alpha \Braket{\frac{u-v}{\sqrt2}\bigg| X_\alpha \bigg| \frac{u'-v'}{\sqrt2}} \Braket{\frac{u+v}{\sqrt2}\bigg| X_\alpha \bigg| \frac{u'+v'}{\sqrt2}} \\
&\eqt{(iv)}\, \Braket{\frac{u-v}{\sqrt2}\bigg| X \bigg| \frac{u'-v'}{\sqrt2}} \Braket{\frac{u+v}{\sqrt2}\bigg| X \bigg| \frac{u'+v'}{\sqrt2}} \\
&=\, \Braket{\frac{u-v}{\sqrt2},\, \frac{u+v}{\sqrt2}\bigg| (X\otimes X) \bigg| \frac{u'-v'}{\sqrt2},\, \frac{u'+v'}{\sqrt2}} \\
&\eqt{(v)}\, \braket{u,v|U(X\otimes X)U^\dag |u',v'}\, .
\ee
Here, (i)~and~(iv) hold by the definition of weak* convergence, (ii)~follows from the quantum Darmois--Skitovich theorem (Lemma~\ref{QDS_lemma}), while (iii)~and~(v) descend from the identity~\eqref{beam_splitter_coherent}. Since linear combinations of (tensor products of) coherent states are dense \tcb{in the topology induced by the Hilbert space norm}, from the above identity we deduce that in fact $X\otimes X = U(X\otimes X)U^\dag$. Applying Lemma~\ref{QDS_lemma} once again concludes the proof.  
\end{proof}

\begin{cor} \label{relent_non_Gaussianity_cor}
The relative entropy of non-Gaussianity~\eqref{relent_non_Gaussianity} is:
\begin{enumerate}[(a)]
    \item always achieved: if it is finite at $\rho$, then it is achieved at its Gaussification $\sigma=\rho_\G$, and therefore~\eqref{relent_non_Gaussianity_closed_form} holds; and
    \item lower semi-continuous with respect to the trace norm topology.
\end{enumerate}
\end{cor}

A notable aspect of the above result is claim~(b), which implies in particular that the map
\bb\label{v-int-map}
\left\{\rho\in \D(\HH_m):\, \Tr \rho N< \infty \right\} \ni \rho \longmapsto S(\rho_\G) - S(\rho)\, ,
\ee
where $\rho_\G$ is the Gaussification of $\rho$, is lower semi-continuous, something that is not at all obvious a priori. We know from Ref.~\cite{Kuroiwa2021} that such map can very well be discontinuous even on energy-bounded states, so that lower semi-continuity is really the strongest form of regularity that can reasonably be obeyed.

The lower semi-continuity of the map~\eqref{v-int-map} implies the lower semi-continuity of the map $\rho\mapsto S(\rho_\G)$ on the same set. This property is quite surprising in view of the discontinuity of the map $\rho\mapsto\rho_\G$ on the set of energy-bounded states, which follows from the aforementioned discontinuity of the map~\eqref{v-int-map} (established by Ref.~\cite{Kuroiwa2021}) and the continuity of the entropy on these sets (along with the equality $\Tr \rho_\G N=\Tr \rho N$).

\begin{proof}[Proof of Corollary~\ref{relent_non_Gaussianity_cor}]
Thanks to Lemma~\ref{w*_closed_cone_G_lemma}, claim~(b) follows immediately from Theorem~\ref{achievable_relent_thm}. As for claim~(a), let $\rho$ be an arbitrary $m$-mode state. If \tcb{$\nonG(\rho)=+\infty$} then there is nothing to prove. Otherwise, $D(\rho\|\sigma)<\infty$ for some Gaussian state $\sigma\in \GG_m$. We will argue that in fact $\Tr \rho N<\infty$, so that $\rho$ has well-defined (first and) second moments.

Represent $\sigma$ as $\sigma=\pazocal{V}\left(\ketbra{0} \otimes Z^{-1} e^{-H_q}\right)\pazocal{V}^\dag$, where $\ket{0}$ is the $k$-mode vacuum ($0\leq k\leq m$), $\pazocal{V}$ a certain `Gaussian unitary operator', $H_q=\sum_{j=1}^{m-k} \omega_j (x_j^2+p_j^2-1)/2$ with $\omega_j>0$, and $Z$ a normalising constant~\cite[Eq.~(3.60)]{BUCCO}. Given that $D(\rho\|\sigma)<\infty$, it must be that $\pazocal{V}^\dag\rho\pazocal{V} = \ketbra{0}\otimes \rho'$, so that
\bb
D(\rho\|\sigma) = D\left( \pazocal{V}\left(\ketbra{0}\otimes \rho' \right) \pazocal{V}^\dag \big\| \pazocal{V}\left(\ketbra{0}\otimes Z^{-1} e^{-H_q} \right) \pazocal{V}^\dag\right) = D\left(\rho'\big\| Z^{-1} e^{-H_q} \right) .
\label{relent_non_Gaussianity_proof_eq1}
\ee
Now, let us invoke the variational expression for the measured relative entropy in Ref.~\cite[Lemma~20]{nonclassicality}. Upon taking a straightforward limit, this implies that
\bb
D\left(\rho'\big\| Z^{-1} e^{-H_q} \right) \geq \sup_{L>0} \left\{ \Tr \rho' \ln L - \ln \Tr \left[ Z^{-1} e^{-H_q} L \right] \right\} \geq \frac12 \Tr \rho' H_q - \ln \Tr Z^{-1} e^{-H_q/2}\, .
\label{relent_non_Gaussianity_proof_eq2}
\ee
\tcb{Since both~\eqref{relent_non_Gaussianity_proof_eq1} as well as the last term on the rightmost side of~\eqref{relent_non_Gaussianity_proof_eq2} are finite,} putting together~\eqref{relent_non_Gaussianity_proof_eq1} and~\eqref{relent_non_Gaussianity_proof_eq2} we deduce that $\Tr \rho' H_q<\infty$, hence $\Tr \rho'N'<\infty$ and thus $\Tr \rho N<\infty$, i.e.\ $\rho$ has well-defined second moments. Applying the result of Ref.~\cite{Marian2013} then completes the proof of claim~(a).
\end{proof}

\section{Tight uniform continuity bounds for the relative entropy of entanglement and generalisations thereof} \label{tight_uniform_sec}

Until now we have considered the problem of establishing the lower semi-continuity of the relative entropy of resource $D_\FF$. At this point, the reader could wonder whether and under what conditions $D_\FF$ is fully-fledged continuous. Strictly speaking, this can only happen in finite dimension. In fact, as mentioned before, $D_\FF$, like many entropic quantities, is often everywhere discontinuous in infinite dimensions~\cite{Wehrl}. A remarkable example of this behaviour is offered, for instance, by the relative entropy of entanglement~\cite{Eisert2002}, but a similar reasoning holds in other cases as well. And yet, such a highly discontinuous behaviour does not represent a problem physically, because it typically involves infinite-energy states. Throughout this section we will see how it is possible to restore a (slightly weaker) form of continuity for the relative entropy of entanglement and related quantities by looking only at the physically meaningful energy-constrained states~\cite{tightuniform,Shirokov-AFW-1, Shirokov-AFW-2, Shirokov-AFW-3}.

\subsection{Energy constraints} \label{energy_constraints_subsec}

In order to model an energy constraint we introduce a \deff{Hamiltonian}, i.e.\ a densely defined, positive semi-definite operator $H$ whose spectrum $\spec(H)$ is bounded from below. Since the ground state energy can be re-defined without affecting the physics, we will hereafter take $H$ to be \deff{grounded}, that is, such that $\min\spec(H)=0$. In this context, another important assumption is convexity: if $\FF\subseteq \D(\HH)$ is a convex subset of states, then the function $D_\FF$ defined in~\eqref{relent_F} is itself convex and satisfies inequality~\eqref{F-p-1} --- informally, it is `not too convex'. If in addition $D_\FF$ does not increase too fast with respect to the energy, in the sense that
\bb
\sup_{\rho\in \D(\HH),\, \Tr\rho H\leq E} D_{\FF}(\rho) = o\left(\sqrt{E}\right)\qquad \textrm{as} \quad E\to+\infty\, ,
\ee
then Ref.~\cite[Proposition~3]{Shirokov-AFW-1} guarantees that $D_\FF$ is uniformly continuous on the set $\{\rho\in D(\HH):\,\Tr \rho H \leq E\}$ of energy-constrained states for any $E>0$. Moreover, it also gives an explicit (uniform) continuity bound for $D_{\FF}$, i.e.\ an upper bound on the quantity $|D_{\FF}(\rho)-D_{\FF}(\sigma)|$ for all pairs of states $\rho,\sigma$ with $\Tr\rho H, \Tr\sigma H\leq E$. Such a bound is faithful, meaning that it vanishes as $\|\rho-\sigma\|_1\to 0$, but it is not accurate when $\rho$ and $\sigma$ are very close --- indeed, in this regime it scales with $\sqrt{\|\rho-\sigma\|_1}$, which is generally not optimal.

Significantly more accurate continuity bounds for the function $D_{\FF}$ under an energy constraint can be obtained by using the methods proposed in Ref.~\cite{Shirokov-AFW-2,Shirokov-AFW-3}. To apply them, we shall also require that the function under examination be bounded by a multiple of the local entropy or the sum of local entropies. To fix ideas and establish more concrete statements,  we will consider in Section IV-B two examples of functions $D_\FF$, namely, the relative entropy of entanglement and the relative entropy of NPT entanglement, as well as another closely related quantity, the Rains bound. In all these cases, the infinite-dimensional generalisations of the Alicki--Fannes--Winter method~\cite{Alicki-Fannes, tightuniform} yield explicit and asymptotically tight uniform continuity bounds under energy constraint on one party of a bipartite system. In Section IV-C we will  obtain uniform continuity bounds for the relative entropy of $\pi$-entanglement in $m$-partite system defined in~\eqref{m-E-r} for any given set $\pi$ of partitions of $\{1,\ldots, m\}$ under different forms of energy constraints.

From now on, assume that $H$ is a grounded Hamiltonian $H$ on $\HH$. The methods of Ref.~\cite{Shirokov-AFW-2,Shirokov-AFW-3} need one more regularity assumption on $H$: informally, this captures the fact that its energy levels should not become too dense as the energy grows; formally, the requirement is that
\begin{equation}\label{H-cond+}
  \lim_{\lambda\rightarrow 0^+}\left(\Tr\, e^{-\lambda H}\right)^{\lambda} = 1\,.
\end{equation}
Note that already the finiteness of the trace on the left-hand side --- commonly referred to as the Gibbs hypothesis --- guarantees that $H$ has a purely discrete spectrum and that each eigenvalue has finite multiplicity.
As established in Ref.~\cite[Lemma~1]{Shirokov-AFW-1}, condition~\eqref{H-cond+} holds if and only if the function $F_H:[0,+\infty) \to \R$ defined by
\begin{equation}\label{F_H}
  F_{H}(E)\coloneqq \sup_{\rho\in \D(\HH),\, \Tr \rho H \leq E} S(\rho)
\end{equation}
where $S$ is the von Neumann entropy~\eqref{von_Neumann}, is finite and satisfies that
\begin{equation}\label{H-cond++}
  F_{H}(E) = o\left(\sqrt{E}\right)\qquad\textrm{as}\quad E\rightarrow+\infty\, .
\end{equation}
Importantly, condition~\eqref{H-cond+} holds for the Hamiltonians of many quantum systems of practical interest, including those made of finitely many harmonic oscillators~\cite{Shirokov-AFW-1, Simon-Nila}.\footnote{Also, Ref.~\cite[Theorem~3]{Simon-Nila} gives a sufficient condition for the logarithmic growth of the function $F_{H}$ in terms of the spectrum of $H$.}

The function $F_H$ defined in~\eqref{F_H} encodes some key information on the physics of the system. By Ref.~\cite[Proposition~1]{Shirokov-1} (see also Ref.~\cite[Proposition~11]{VV-diamond}) it is continuous, strictly increasing, and strictly concave. However, it often does not possess some other handy properties that turn out to be useful in computations. To enforce such properties, let us consider an auxiliary function $G:[0,+\infty) \to \R$ that should be thought of as a more `regular' version of $F_H$ itself. We will assume that:
\begin{enumerate}[(i)]
    \item $G$ is continuous;
    \item $G$ is non-decreasing;
    \item $G(E)\geq F_H(E)$ for all $E\geq 0$;
    \item $G(E) = o\left(\sqrt{E}\right)$ as $E\to +\infty$;
    \item $E\mapsto G(E)\big/\sqrt{E}$ is non-increasing.
\end{enumerate}
The existence of a function $G$ with these properties is proved in~\cite[Proposition~1]{Shirokov-AFW-2}, where it is also shown that the minimal such function is given by
\bb
G_{\min}(E) \coloneqq \sqrt{E} \sup_{E'\geq E} \frac{F_H\left(E'\right)}{\sqrt{E'}}\, ,
\label{G_min}
\ee
for all $E\geq 0$. 
If the operator $H$ satisfies also the condition in Ref.~\cite[Theorem~3]{Simon-Nila}, then one can find a function $G$ satisfying all the above conditions~(i)--(v) and moreover such that
\begin{enumerate}[(vi)]
    \item $G(E) = \left(1+o(1)\right) F_H(E)$ as $E\to\infty$.
\end{enumerate}
For example, consider an $\ell$-mode quantum oscillator with frequencies $\omega_1,\ldots ,\omega_{\ell}$. Denoting the annihilation and creation operators corresponding to the $k^\text{th}$ mode with $a_k,a_k^\dag$, respectively, the total Hamiltonian takes the form
\bb
H=\sum_{k=1}^\ell \omega_k a_k^\dag a_k\, .
\label{canonical_Hamiltonian}
\ee
The function
\bb \label{F-ub+}
G_{\ell,\omega}(E) \coloneqq \ell \left( \ln \frac{E+2E_0}{\ell E_*} + 1 \right) ,\quad E_0=\frac{1}{2}\sum_k \omega_k,\quad E_*=\sqrt[\ell]{\prod_k \omega_k},  
\ee
satisfies all the above conditions~(i)--(vi), as discussed in Ref.~\cite{Shirokov-AFW-2}.

\subsection{Uniform continuity on energy-constrained states: bipartite case}

Throughout this section, we will establish tight uniform continuity bounds on the relative entropy of entanglement and some closely related quantities that are valid on energy-constrained states. The functions we will consider are $E_R$ (defined by~\eqref{relent_entanglement}), $E_{R,\, \PPT}$ (defined by~\eqref{relent_NPT_entanglement}), the Rains bound $R$ (defined by~\eqref{Rains_bound}), as well as the corresponding regularised quantities, constructed as
\bb
f^\infty(\rho_{AB}) \coloneqq \lim_{n\to\infty} \frac1n f^\infty\left(\rho_{AB}^{\otimes n}\right),\qquad f=E_R,\, E_{R,\,\PPT},\, R\, ,
\label{regularisation}
\ee
where the limit exists thanks to Fekete's lemma~\cite{Fekete1923}, because the un-regularised quantities are all sub-additive on tensor products.

The above functions $E_R$, $E_{R,\,\PPT}$, and $R$ are non-negative, convex, and they satisfy inequality~\eqref{F-p-1}. Moreover, for any bipartite state $\rho_{AB}$ on $\HH_{AB}$ we have~\cite{Plenio-Virmani, Vedral1998}
\bb \label{E-R-UB}
R(\rho_{AB})\leq E_{R,\,\PPT}(\rho_{AB})\leq E_R(\rho_{AB})\leq S(\rho_A)\, .
\ee
Therefore, if $\,d=\min\{\dim \HH_A,\dim \HH_B\}<+\infty$ then by using~\cite[Lemma~7]{tightuniform} and the arguments from~\cite[proof of Corollary~8]{tightuniform} one can show that
\begin{equation}\label{W-CB}
|f(\rho)-f(\sigma)|\leq\varepsilon \ln d+g(\varepsilon),\quad f=E_R,\,E^{\infty}_R,\,E_{R,\,\PPT},\,E^{\infty}_{R,\,\PPT},\,R,\, R^{\infty}
\end{equation}
for any states $\rho$ and $\sigma$ in $\D(\HH_{AB})$ such that $\frac{1}{2}\|\rho-\sigma\|_1\leq\varepsilon$, where 
$g(x)\coloneqq (x+1)\ln(x+1)-x\ln x$. It is easy to see that all the continuity bounds in~\eqref{W-CB} are asymptotically tight\footnote{A continuity bound $\displaystyle\sup_{x,y\in X_a}|f(x)-f(y)|\leq M_a(x,y)$ depending on a parameter $a$ is called \emph{asymptotically tight} for large $a$ if $\displaystyle\limsup_{a\rightarrow+\infty}\sup_{x,y\in X_a}\frac{|f(x)-f(y)|}{M_a(x,y)}=1$.} for large $d$.

If both systems $A$ and $B$ are infinite dimensional, then asymptotically tight uniform continuity bounds for the functions $E_{R}$, $E_{R,\,\PPT}$, $R$ (as well as their regularisations) under the energy constraint on one of the systems $A$ and $B$ can be obtained by using Ref.~\cite[Theorem~1]{Shirokov-AFW-2} and its proof. 


\begin{prop}\label{CB}
Let $H$ be a positive operator on $\HH_A$ satisfying condition~\eqref{H-cond+}, and let $G:[0,+\infty)\to \R$ be any function on $\R_+$ satisfying conditions~(i)--(v) in Section~\ref{energy_constraints_subsec}. In what follows, $f$ will denote one of the functions $E_R$, $E^{\infty}_R$, $E_{R,\,\PPT}$, $E^{\infty}_{R,\,\PPT}$, $R$, $R^{\infty}$ defined by~\eqref{relent_entanglement},~\eqref{relent_NPT_entanglement},~\eqref{Rains_bound}, and~\eqref{regularisation}. Let $\varepsilon>0$, $E>0$, and $T \coloneqq \frac{1}{\varepsilon} \min\left\{1,\, \sqrt{\frac{E}{G^{-1}\left(\ln d_0\right)}}\right\}$, where $d_0$ is any positive integer such that $\ln d_0>G(0)$. For any two states $\rho_{AB},\sigma_{AB}\in \D(\HH_{AB})$ such that $\Tr H \rho_A,\, \Tr H \sigma_A \leq E$ and $\frac{1}{2}\|\rho-\sigma\|_1\leq \varepsilon$, it holds that
\bb \label{main-CB}
    \left|f(\rho)-f(\sigma)\right|\leq \inf_{t\in (0,T]} \left\{ \varepsilon(1+4t) \left(G\!\left(\!\frac{E}{(\varepsilon t)^2}\!\right) + \frac{1}{d_0} +\ln 2\right) + 2g(\varepsilon t) + g(\varepsilon(1+2t))\right\} ,
\ee
where $g(x) \coloneqq (x+1)\ln(x+1)-x\ln x$.

\tcb{If in addition we assume that the operator $H$ satisfies the Gibbs hypothesis, i.e.\ $\Tr e^{-\lambda H} <\infty$ for all $\lambda>0$, as well as the condition in~\cite[Theorem~3]{Simon-Nila}, expressed by the inequality\footnote{Eq.~\eqref{Becker_Datta_condition} implies in particular that the limit on the right-hand side is required to exist. Also, in~\eqref{Becker_Datta_condition} the spectrum $\spec(H)$ is intended as a multi-set, i.e.\ it is understood to contain each eigenvalue a number of times equal to its multiplicity.}
\bb
\lim_{E\to\infty} \frac{\sum_{\lambda,\lambda'\in \spec(H),\, \lambda+\lambda'\leq E} \lambda^2}{\sum_{\lambda,\lambda'\in \spec(H),\, \lambda+\lambda'\leq E} \lambda \lambda'} > 1\, ,
\label{Becker_Datta_condition}
\ee
} then for each function $f$ the continuity bound in~\eqref{main-CB} is asymptotically tight for large $E$. This is true, in particular, if $H$ is the canonical Hamiltonian~\eqref{canonical_Hamiltonian} of an $\ell$-mode quantum oscillator. In this case~\eqref{main-CB} holds with the choice~\eqref{F-ub+} of $G$, yielding
\bb\label{exp-form}
    \left|f(\rho)-f(\sigma)\right|\leq \inf_{t\in (0,T_*]} \left\{ \varepsilon(1+4t) \left( \ell \ln \frac{E/(\varepsilon t)^2 + 2E_0}{\ell E_*} + \ell + e^{-\ell} +\ln 2 \right) + 2 g(\varepsilon t) + g(\varepsilon(1+2t)) \right\},
\ee
where $T_* \coloneqq \frac{1}{\varepsilon} \min\left\{1,\, \sqrt{\frac{E}{E_0}}\right\}$ (the parameters $E_0$ and $E_*$ are determined in~\eqref{F-ub+} via the frequencies of the oscillator).
\end{prop}

\begin{rem} The right hand side of~\eqref{main-CB} tends to zero as $\,\varepsilon\to 0\,$ for any given $E>0$ due to the condition $G(E)=o\big(\sqrt{E}\big)$ as $E\to+\infty$.
\end{rem}

\begin{proof}[Proof of Proposition~\ref{CB}]
The validity of inequality (\ref{main-CB}) for $f=E_R,\, E_{R,\,\PPT},\, R$  follows directly from~\cite[Theorem~1]{Shirokov-AFW-2}. To prove inequality~\eqref{main-CB} for $f=E^{\infty}_R,\, E^{\infty}_{R,\,\PPT},\, R^{\infty}$ we will use the telescopic method from Ref.~\cite[proof of Corollary~8]{tightuniform}, with necessary modifications, combined with Ref.~\cite[proof of Theorem~1]{Shirokov-AFW-2}. We will assume that $f=E^{\infty}_R$, while the cases $f=E^{\infty}_{R,\,\PPT}$ and $f=R^{\infty}$ are addressed similarly. For a given positive integer $n$ we have
\bb
\left| E_R(\rho^{\otimes n})-E_R(\sigma^{\otimes n}) \right| \leq\sum_{k=1}^n \left|E_R\left(\rho^{\otimes k}\otimes\sigma^{\otimes(n-k)}\right)-E_R\left(\rho^{\otimes (k-1)}\otimes\sigma^{\otimes(n-k+1)}\right)\right|\displaystyle\leq\sum_{k=1}^n \left|E_R\left(\rho\otimes\omega_k\right)-E_R\left(\sigma\otimes\omega_k\right)\right| ,
\label{telescopic_eq1}
\ee
where the states $\rho$ and $\sigma$ in the last expression are states of the system $A_kB_k$ (the $k$-th copy of $AB$), while $\omega_k=\rho^{\otimes (k-1)}\otimes\sigma^{\otimes(n-k)}$ is a state of the system $A_1B_1\ldots A_nB_n\setminus A_kB_k$. The assumption $\Tr H\rho_A,\, \Tr H\sigma_A\leq E$ and the inequality~\eqref{E-R-UB} imply finiteness of all the terms in~\eqref{telescopic_eq1}.
So, to prove inequality~\eqref{main-CB} for $f=E^{\infty}_R$ it suffices to show that
\begin{equation}\label{RE-T}
\left|E_R\left(\rho\otimes\omega_k\right)-E_R\left(\sigma\otimes\omega_k\right)\right|\leq 
\Delta_t(E,\varepsilon)\, ,
\end{equation}
where 
$\Delta_t(E,\varepsilon)$ denotes the right-hand side of~\eqref{main-CB}, for each $k$ and any $t\in(0,T]$. This can be done thanks to the arguments from Ref.~\cite[proof of Theorem~1]{Shirokov-AFW-2}, which we summarise in the remaining part of the proof.

To simplify the notation, we will rename the system $A_kB_k$ as $AB$. For some $d\geq d_0$, let $\gamma(d)\coloneqq G^{-1}(\ln d)$. Thanks to Ref.~\cite[Lemma~1]{Shirokov-AFW-2}, for any $d>d_0$ such that $E\leq\gamma(d)$ there exist states $\rho',\sigma',\alpha_i,\beta_i\in\D(\HH_{AB})$ ($i=1,2$) and numbers $s,t\in (0,1)$ such that:
\begin{enumerate}[(a)]
    \item $\rk\rho'_A,\,\rk \sigma'_A\leq d$;
    \item $\Tr H \rho'_A, \Tr H \sigma'_A \leq E$;
    \item $\frac12 \|\rho-\rho'\|_1\leq s \leq \sqrt{E/\gamma(d)}$ and $\frac12 \|\sigma-\sigma'\|_1\leq t \leq\sqrt{E/\gamma(d)}$;
    \item $\Tr H \alpha_i^A\leq E/s^2$ and $\Tr H \beta_i^A\leq E/t^2$ ($i=1,2$), where $\alpha_i^A = \Tr_B \alpha_i^{AB}$, and analogously for $\beta_i$; furthermore,
    \item it holds that
    \begin{equation}\label{2-r-er}
    (1-s')\rho+s'\alpha_1=(1-s')\rho'+s'\alpha_2,\qquad (1-t')\sigma+t'\beta_1 =(1-t')\sigma'+t'\beta_2\, ,
    \end{equation}
    where $s'=\frac{s}{1+s}$ and $t'=\frac{t}{1+t}$.
\end{enumerate}
The function $E_R$ is well defined on all the states $\rho'\otimes\omega_k$, $\sigma'\otimes\omega_k$, $\alpha_i\otimes\omega_k$, $\beta_i\otimes\omega_k$ ($i=1,2$), since their marginal states corresponding to the subsystems $A_1,\ldots,A_n$ have finite energy. So, by using the first relation in~\eqref{2-r-er}, the convexity of $E_R$ and inequality~\eqref{F-p-1} it is easy to show that
\bb
(1-s')(E_R(\rho\otimes\omega_k)-E_R(\rho'\otimes\omega_k))\leq s' (E_R(\alpha_2\otimes\omega_k)-E_R(\alpha_1\otimes\omega_k))+ h_2(s')
\ee
and
\bb
(1-s')(E_R(\rho'\otimes\omega_k)-E_R(\rho\otimes\omega_k))\leq s' (E_R(\alpha_1\otimes\omega_k)-
E_R(\alpha_2\otimes\omega_k))+h_2(s')\, ,
\ee
with $h_2$ being the binary entropy~\eqref{h_2}. These inequalities imply that
\begin{equation}\label{one-er}
|E_R(\rho'\otimes\omega_k)-E_R(\rho\otimes\omega_k)|\leq s|E_R(\alpha_2\otimes\omega_k)-E_R(\alpha_1\otimes\omega_k)|+g(s)
\end{equation}
Assume that $E_R(\alpha_2\otimes\omega_k)\geq E_R(\alpha_1\otimes\omega_k)$. Then the subadditivity of $E_R$ implies 
\bb
E_R(\alpha_2\otimes\omega_k)\leq E_R(\alpha_2)+E_R(\omega_k),
\ee
while the monotonicity of $E_R$ under local partial traces shows that $E_R(\alpha_1\otimes\omega_k)\geq E_R(\omega_k)$ (cf.\ Ref.~\cite{tightuniform}). Hence,
\begin{equation}\label{RE-T-2-er}
|E_R(\alpha_2\otimes\omega_k)-E_R(\alpha_1\otimes\omega_k)|\leq \max\left\{E_R(\alpha_1),E_R(\alpha_2)\right\}.
\end{equation}
Since $\Tr H \alpha_i^A\leq E/s^2$ ($i=1,2$), it follows from inequality~\eqref{E-R-UB} together with property~(iii) in Section~\ref{energy_constraints_subsec} (cf.~\eqref{F_H}) that
\begin{equation}\label{one+er}
\max\{E_R(\alpha_1),E_R(\alpha_2)\}\leq G\!\left(E/s^2\right).
\end{equation}
Inequalities (\ref{one-er}), (\ref{RE-T-2-er}) and (\ref{one+er}) imply that
\begin{equation}\label{one-er+}
|E_R(\rho'\otimes\omega_k)-E_R(\rho\otimes\omega_k)|\leq s\,G\!\left(E/s^2\right)+g(s)
\end{equation}
Similarly, by using the second relation in~\eqref{2-r-er} and by noting that $\Tr H \beta_i^A \leq E/t^2$ ($i=1,2$), we obtain
\begin{equation}\label{two-er+}
|E_R(\sigma'\otimes\omega_k)-E_R(\sigma\otimes\omega_k)|\leq t\,G\!\left(E/t^2\right)+g(t).
\end{equation}
Since $s,t\leq y\coloneqq \sqrt{E/\gamma(d)}$ and the function $E\mapsto G(E)/\sqrt{E}$ is non-increasing by property~(v) in Section~\ref{energy_constraints_subsec}, for $x=s,t$ we have
\begin{equation}\label{three-er}
x\, G\!\left(E/x^2\right)\leq y\, G\!\left(E/y^2\right)=\sqrt{E/\gamma(d)}\,G\!\left(\gamma(d)\right)=
\sqrt{E/\gamma(d)}\ln d,
\end{equation}
where the last equality follows from the definition of $\gamma(d)$.

Now, thanks to the fact that $\rk \rho'_A\leq d$ and $\rk \sigma'_A\leq d$, the supports of both $\rho'_A$ and $\sigma'_A$ are contained in some $2d$-dimensional subspace of $\HH_A$. By the triangle inequality we have
\bb
\|\rho'-\sigma'\|_1\leq \|\rho'-\rho\|_1+\|\sigma'-\sigma\|_1+\|\rho-\sigma\|_1\leq 2\varepsilon+4\sqrt{E/\gamma(d)}\, .
\ee
So, the arguments from Ref.~\cite[proof of Corollary~8]{tightuniform} imply that
\begin{equation}\label{fcb-er}
 |E_R(\sigma'\otimes\omega_k)-E_R(\rho'\otimes\omega_k)|\leq \ln (2d)\left(2\sqrt{E/\gamma(d)}+\varepsilon\right)+g\!\left(2\sqrt{E/\gamma(d)}+\varepsilon\right).
\end{equation}
It follows from~\eqref{one-er+}--\eqref{fcb-er} and the monotonicity of the function $g$ that
\begin{equation}\label{m-cb-er}
|E_R(\rho\otimes\omega_k)-E_R(\sigma\otimes\omega_k)|\leq  \left(4\sqrt{E/\gamma(d)}+\varepsilon\right) \ln (2d)+g\!\left(2\sqrt{E/\gamma(d)}+\varepsilon\right)+2g\!\left(\sqrt{E/\gamma(d)}\right).
\end{equation}

We now conclude the proof of~\eqref{RE-T}. If $t\in(0,T]$ then there is a natural number $d_*>d_0$  such that $\gamma(d_*)>E/(\varepsilon t)^2$ but $\gamma(d_*-1)\leq E/(\varepsilon t)^2$. It follows that
\bb
\sqrt{E/\gamma(d_*)}\leq \varepsilon t\leq 1\qquad \textrm{and} \qquad
\ln (d_*-1) = G(\gamma(d_*-1))\leq G(E/(\varepsilon t)^2)\,.
\ee
Since $\ln d_*\leq\ln (d_*-1)+1/(d_*-1)\leq\ln (d_*-1)+1/d_0$, inequality~\eqref{m-cb-er} with $d=d_*$ and the monotonicity of the function $g$ imply the claimed relation~\eqref{RE-T}.

Assume now that the operator $H$ satisfies the condition of~\cite[Theorem~3]{Simon-Nila} and that the function $G$ satisfies property~(vi) in Section~\ref{energy_constraints_subsec}. Note first that the condition in Ref.~\cite[Eq.~(28)]{Shirokov-AFW-2} 
holds whenever $f$ is any of the functions $E_R$, $E^{\infty}_R$, $E_{R,\,\PPT}$, $E^{\infty}_{R,\,\PPT}$, $R$, $R^{\infty}$. This can be shown by using any purification $\widehat{\gamma}(E)$ of the Gibbs state $\gamma(E)\coloneqq \frac{e^{-\lambda H}}{\Tr e^{-\lambda H}}$ of system $A$, where $\lambda$ is determined by the equation $E\Tr e^{-\lambda H}=\Tr H e^{-\lambda H}$~\cite{Wehrl}, since $\Tr H \gamma(E)=E$ and
\bb \label{Gibbs-s-r}
f\left(\widehat{\gamma}(E)\right)=S\left(\gamma(E)\right)=F_{H}(E),\qquad  f=E_R,\,E^{\infty}_R,\,E_{R,\,\PPT},\,E^{\infty}_{R,\,\PPT},\,R,\, R^{\infty},\qquad\forall\ E>0\, .
\ee
Thus, the asymptotical tightness of the continuity bound in~\eqref{main-CB} for $f=E_R,E_{R,\,\PPT},R$ follows directly from the corresponding assertion of Ref.~\cite[Theorem~1]{Shirokov-AFW-2}. The asymptotical tightness of the continuity bound in~\eqref{main-CB} for $f=E^{\infty}_R,E^{\infty}_{R,\,\PPT}R^{\infty}$ can be shown by repeating the arguments from the proof of the aforementioned assertion of Ref.~\cite[Theorem~1]{Shirokov-AFW-2}.

If $H$ is the Hamiltonian~\eqref{canonical_Hamiltonian} of the $\ell$-mode quantum oscillator then the use of the function $G_{\ell,\omega}$ in~\eqref{F-ub+} in the role of $G$ allows to write the right-hand side of~\eqref{main-CB} in an explicit form. We refer the reader to Ref.~\cite[Section~3.2]{Shirokov-AFW-2} for details.
\end{proof}

If $H$ is the Hamiltonian of a quantum system $A$ then the positive operator on $\HH_A^{\otimes n}$ defined by the formula 
\begin{equation}\label{H-n}
H_n \coloneqq H\otimes \id \otimes \ldots \otimes \id + \cdots + \id \otimes \ldots \otimes \id \otimes H\, ,
\end{equation}
where $\id$ is the unit operator on each of the factors $\HH_{A_k}$, is the Hamiltonian of the system $A^n$ obtained by joining $n$ copies of $A$~\cite{HOLEVO-CHANNELS-2}. The continuity bounds in~\eqref{main-CB} then imply the following;

\begin{cor}\label{b-s-cont}
Let $A$ and $B$ be arbitrary quantum systems. If $H$ is a positive operator on $\HH_{A}$ satisfying condition~\eqref{H-cond+}, then
\begin{itemize}
    \item the functions $E_R$, $E^{\infty}_R$, $E_{R,\,\PPT}$, $E^{\infty}_{R,\,\PPT}$, $R$, and $R^{\infty}$ are uniformly continuous on the set of states $\rho$ in $\D(\HH_{AB})$ such that $\Tr H \rho_A\leq E$ for any $E>0$;
    \item the functions $E_R$, $E^{\infty}_R$, $E_{R,\,\PPT}$, $E^{\infty}_{R,\,\PPT}$, $R$, and $R^{\infty}$ are asymptotically continuous in the following sense~\cite{Eisert2002}:
if $(\rho_n)_{n\in \N}$ and $(\sigma_n)_{n\in \N}$ are sequences of states such that
\bb
\rho_n,\sigma_n \in\D(\HH_{AB}^{\otimes n}),\qquad \Tr H_n \rho_n^{A^n},\, \Tr H_n\sigma_n^{A^n} \leq nE,\quad \forall\ n\, ,\qquad \text{and} \qquad  \lim_{n\to+\infty} \|\rho_n-\sigma_n\|_1=0\, ,
\ee
where $H_n$ is the positive operator on $\HH_A^{\otimes n}$ defined in~\eqref{H-n} and $E>0$ is a finite positive number, then
\bb
\lim_{n\to+\infty}\frac{|f(\rho_n)-f(\sigma_n)|}{n}=0,\qquad f=E_R,\, E^{\infty}_R,\, E_{R,\,\PPT},\, E^{\infty}_{R,\,\PPT},\, R,\, R^{\infty}\, .
\ee
\end{itemize}
\end{cor}

\begin{proof}
The first assertion directly follows from the continuity bounds in~\eqref{main-CB} (since the right-hand side of~\eqref{main-CB} vanishes as $\varepsilon\to 0$). To prove the second assertion, note that $F_{H_n}(E)=nF_{H}(E/n)$ for each $n$~\cite[Lemma~2]{Shirokov-AFW-3}. So, if $G:[0,+\infty)\to \R$ is any function on $\R_+$ satisfying conditions~(i)--(v) in Section~\ref{energy_constraints_subsec} for the operator $H$, and $d_0$ is a positive integer such that $\ln d_0>G(0)$, then the function $G_n(E)\coloneqq nG(E/n)$ satisfies the same conditions for the operator $H_n$ and $d_n\coloneqq d^n_0$ is a positive integer such that $\ln d_n>G_n(0)$. Using this it is easy to obtain from Proposition~\ref{CB} that
\begin{equation}\label{ucb-n}
    \frac{|f(\rho_n)-f(\sigma_n)|}{n}\leq \varepsilon_n(1+4t) \left(G\!\left(\!\frac{E}{(\varepsilon_n t)^2}\!\right) + \frac{1}{nd^n_0} +\frac{\ln 2}{n}\right) + \frac{2 g(\varepsilon_n t) + g(\varepsilon_n(1+2t))}{n}\, ,
\end{equation}
where $f=E_R,\, E^{\infty}_R,\, E_{R,\,\PPT},\, E^{\infty}_{R,\,\PPT},\, R,\, R^{\infty}$, for any $t\in(0,T']$, $\varepsilon_n \coloneqq \frac12 \|\shs\rho_n-\sigma_n\|_1$, and $T'\coloneqq \min\left\{1,\,\sqrt{\frac{E}{G^{-1}\left(\ln d_0\right)}}\right\}$. Since the sequence $(\varepsilon_n)_{n\to \N}$ is vanishing by hypothesis and $G(E)=o\big(\sqrt{E}\big)$ as $E\to+\infty$, the right-hand side of~\eqref{ucb-n} tends to zero as $n\to+\infty$ for any fixed $t\in(0,T']$.
\end{proof}

\subsection{Uniform continuity on energy-constrained states: multipartite case}

In this subsection we obtain uniform continuity bounds for the relative entropy of $\pi$-entanglement of a state of $m$-partite system $A_1\ldots A_m$ defined in~\eqref{m-E-r} for any given (non-empty) set $\pi\subseteq P(m)$ of partitions of $\{1,\ldots, m\}$ under the energy constraint imposed either on the whole system $A_1\ldots A_m$ or on the subsystem $A_1 \ldots A_{m-1}$. Note first that
\begin{equation}\label{ER-UB}
E_{R,\pi}(\rho_{A_1\ldots A_m})\leq E_R(\rho_{A_1\ldots A_m})\leq\sum_{k=1}^{m-1}S(\rho_{A_{k}})\, .
\end{equation}
The first inequality follows from the definitions of $E_{R,\pi}$ and $E_R$, the latter being the relative entropy distance from the set of fully separable states. The second inequality is proved in Ref.~\cite{Plenio2001} in the finite-dimensional setting. Its validity in general case is established in Appendix~\ref{inequality_app}. Since the above inequalities hold with arbitrary $m-1$ subsystems of $A_1\ldots A_m$ (instead of $A_1\ldots A_{m-1}$), it is easy to show that
\begin{equation}\label{ER-UB+}
E_{R,\pi}(\rho_{A_1\ldots A_m})\leq E_R(\rho_{A_1\ldots A_m})\leq\,\frac{m-1}{m}\sum_{k=1}^{m}S(\rho_{A_k}).
\end{equation}

There is an important aspect in which $E_{R,\pi}$ differs from its bipartite counterpart, $E_R$. Namely, $E_R$ is sub-additive, in the sense that $E_R(\rho_{AB}\otimes \omega_{A'B'})\leq E_R(\rho_{AB}) + E_R(\omega_{A'B'})$ for all states $\rho_{AB},\omega_{A'B'}$, where the bipartition on the left-hand side is $AA':BB'$. No analogous inequality can be established for $E_{R,\pi}$ when $\pi$ contains more than one partition, because in that case the set of $\pi$-separable states is not closed under tensor products. For this reason, the limit in the regularisation $E_{R,\pi}^\infty (\rho_{AB}) \coloneqq \lim_{n\to\infty} \frac1n E_{R,\pi}\left( \rho_{AB}^{\otimes n}\right)$ of the relative entropy of $\pi$-entanglement is only guaranteed to exist when $\pi$ is composed of one partition only (thanks to Fekete's lemma~\cite{Fekete1923}). We will therefore consider the quantity $E_{R,\pi}^\infty$ only in this special case, which is however physically very relevant, as it includes e.g.\ the fully local scenario (corresponding to the finest partition $\pi=\{\{1\},\ldots,\{m\}\}$).

By using the upper bound~\eqref{ER-UB}, the non-negativity of $E_{R,\pi}$, the result in Ref.~\cite[Lemma~7]{tightuniform}, and the arguments from Ref.~\cite[proof of Corollary~8]{tightuniform}, one can show that
\begin{equation}\label{ER-CB}
|E_{R,\pi}^*(\rho)-E_{R,\pi}^*(\sigma)|\leq\varepsilon \ln \dim\HH_{A_1\ldots A_{m-1}} + g(\varepsilon)\, ,
\end{equation}
for any states $\rho$ and $\sigma$ in $\D(\HH_{A_1\ldots A_{m}})$ such that $\frac{1}{2}\|\rho-\sigma\|_1\leq\varepsilon$, where as usual $g(x)=(x+1)\ln(x+1)-x\ln x$, provided that the systems $A_1,\ldots, A_{m-1}$ are finite dimensional. In~\eqref{ER-CB}, as per the above discussion, we can set either $E^*_{R,\pi}=E_{R,\pi}$ and leave $\pi$ arbitrary, or else, if $\pi$ is composed of one partition only, consider also the case where $E_{R,\pi}^* = E_{R,\pi}^\infty$.

Assume that $A_1,\ldots, A_m$ are arbitrary infinite-dimensional quantum systems. If $H_{1},\ldots,H_{s}$ are the Hamiltonians of quantum systems $A_1,\ldots,A_s$ satisfying condition~\eqref{H-cond+}, where either $s=m-1$ or $s=m$, then the Hamiltonian
\begin{equation}\label{Hm}
\Hs \coloneqq H_1\otimes I_{A_2}\otimes\ldots \otimes I_{A_s} + \cdots + I_{A_1}\otimes\ldots \otimes I_{A_{s-1}} \otimes H_s
\end{equation}
of the system $\As \coloneqq A_1\ldots A_s$ satisfies condition~\eqref{H-cond+} thanks to Ref.~\cite[Lemma~2]{Shirokov-AFW-3}. It follows, leveraging the result in Ref.~\cite[Lemma~1]{Shirokov-AFW-1}, that
\begin{equation}\label{F-s}
F_{\Hs}(E)\coloneqq \sup_{\rho\in\D(\HH_{A_1\ldots A_s}):\, \Tr \rho \Hs\leq E} S(\rho) = o\big(\sqrt{E}\big)\qquad \textrm{as} \quad E\to+\infty\, .
\end{equation}
We will obtain continuity bounds for the function $E_{R,\pi}$ and its regularisation $E_{R,\pi}^{\infty}$ under two forms of energy constraint. They correspond to the cases $s=m-1$ and $s=m$ in the following proposition.

\begin{prop}\label{nREE-CB} Let $m\geq 2$ be an integer, and consider positive operators $H_{1},\ldots ,H_{s}$ on Hilbert spaces $\HH_{A_1},\ldots ,\HH_{A_s}$ that satisfy condition~\eqref{H-cond+}, where either $s=m-1$ or $s=m$. Let $\rho,\sigma\in\D(\HH_{A_1\ldots A_m})$ be such that
\bb\sum_{k=1}^{s}\Tr H_k\rho_{A_k},\,\sum_{k=1}^{s}\Tr H_k \sigma_{A_k}\leq sE\, ,
\ee
and $\frac{1}{2}\|\shs\rho-\sigma\|_1\leq\varepsilon\leq 1$. Let $\pi\subseteq P(m)$ be a non-empty set of partitions of $\{1,\ldots, m\}$. Then
\begin{equation}\label{nREE-CB-1}
|E^*_{R,\pi}(\rho)-E^*_{R,\pi}(\sigma)|\leq \frac{m-1}{s}\, \sqrt{2\varepsilon} \, F_{\Hs}\!\!\left(\frac{sE}{\varepsilon}\right) + g\big(\sqrt{2\varepsilon}\big)\, ,
\end{equation}
for either $E^*_{R,\pi}=E_{R,\pi}$ and $\pi$ arbitrary, or $E_{R,\pi}^* = E_{R,\pi}^\infty$ and $\pi$ composed of one partition only. Here, $F_{\Hs}$ is the function defined in~\eqref{F-s}. 

If all the operators $H_{1},\ldots,H_{s}$ are unitary equivalent to an operator $H$ on $\HH_A$ and $G:[0,+\infty)\to \R$ is any function on $\R_+$ satisfying conditions~(i)--(v) in Section~\ref{energy_constraints_subsec}, then
\begin{equation}\label{nREE-CB-2}
|E_{R,\pi}^{*}(\rho)-E_{R,\pi}^{*}(\sigma)| \leq 
\inf_{t\in (0,\,1/\varepsilon)} \left\{(m-1) \left(\left(\varepsilon+\varepsilon^2t^2\right) G\!\left(\frac{sE}{\varepsilon^2t^2}\right) + 2\sqrt{2\varepsilon t}\, G\!\left(\frac{E}{\varepsilon t}\right)\right)
+ g\!\left(\varepsilon+\varepsilon^2t^2\right)+2g\big(\sqrt{2\varepsilon t}\big)\right\},
\end{equation}
for either $E^*_{R,\pi}=E_{R,\pi}$ and $\pi$ arbitrary, or $E_{R,\pi}^* = E_{R,\pi}^\infty$ and $\pi$ composed of one partition only. In particular, if $H$ is the canonical Hamiltonian~\eqref{canonical_Hamiltonian} of an $\ell$-mode quantum oscillator then 
\eqref{nREE-CB-2} with the choice~\eqref{F-ub+} of $G$ becomes
\bb\label{exp-form-2}
\left|E_{R,\pi}^{*}(\rho)-E_{R,\pi}^{*}(\sigma)\right| &\leq \inf_{t\in (0,\,1/\varepsilon)}\bigg\{(m-1) \left(\varepsilon+\varepsilon^2t^2\right) \ell \, \ln\left(\frac{s E/(\varepsilon t)^2+2E_0}{e^{-1}\ell E_*}\right) \\
&\hspace{13ex} + (m-1)\, 2\sqrt{2\varepsilon t}\, \ell\, \ln\left(\frac{E/(\varepsilon t)+2E_0}{e^{-1}\ell E_*}\right)
+g\!\left(\varepsilon+\varepsilon^2t^2\right)+2g\big(\sqrt{2\varepsilon t}\big)\bigg\}\, ,
\ee
for either $E^*_{R,\pi}=E_{R,\pi}$ and $\pi$ arbitrary, or $E_{R,\pi}^* = E_{R,\pi}^\infty$ and $\pi$ composed of one partition only. Here, the parameters $E_0$ and $E_*$ are defined in~\eqref{F-ub+} via the frequencies of the oscillator. Both continuity bounds in~\eqref{exp-form-2} are asymptotically tight for large $E$ if $m=2$ and  $s=1,2$.
\end{prop} 

\begin{rem} The right-hand sides of~\eqref{nREE-CB-1} and~\eqref{nREE-CB-2} tend to zero as $\varepsilon\to 0$ for any given $E>0$, due to the condition $G(E)=o\big(\sqrt{E}\big)$ as $E\to+\infty$.
\end{rem}

\begin{proof}[Proof of Proposition~\ref{nREE-CB}]
The upper bounds~\eqref{ER-UB} and~\eqref{ER-UB+}, the non-negativity and convexity of $E_{R,\pi}$, together with the general inequality~\eqref{F-p-1} show that for any non-empty set of partitions $\pi$  the function $E_{R,\pi}$ belongs to the classes $L_m^{m-1}(1,1)$ and $L_m^{m}(1-1/m,1)$ defined in Ref.~\cite{Shirokov-AFW-3}. So, in both cases $s=m-1,m$, the continuity bounds~\eqref{nREE-CB-1} and~\eqref{nREE-CB-2} for $E^{*}_{R,\pi} = E_{R,\pi}$ follow directly from the results of Ref.~\cite[Theorems~1 and~2]{Shirokov-AFW-3}.

To prove the  continuity bound~\eqref{nREE-CB-1} for $E^{*}_{R,\pi}=E_{R,\pi}^{\infty}$, we will use the telescopic method from Ref.~\cite[proof of Corollary~8]{tightuniform}, with necessary modifications, combined with Ref.~\cite[proof of Theorem~1]{Shirokov-AFW-1}. We will consider the cases $s=m-1$ and $s=m$ simultaneously.

Let $\Hs$ be the operator defined in~\eqref{Hm}. Since $\Tr \left[\Hs(\rho_{A_1}\otimes\ldots \otimes\rho_{A_s}) \right] =\sum_{k=1}^s \Tr H_{k}\rho_{A_k}$, we have
\begin{equation}\label{u-ineq}
\sum_{k=1}^s S(\rho_{A_k})=S(\rho_{A_1}\otimes \ldots \otimes\rho_{A_s})\leq F_{\Hs}(sE)
\end{equation}
for any state $\rho\in\D(\HH_{A_1\ldots A_m})$ such that $\Tr \Hs\rho_{A_{[s]}}=\sum_{k=1}^s\Tr H_{k}\rho_{A_k}\leq sE$. Hence for any such state $\rho$  inequalities~\eqref{ER-UB} and~\eqref{ER-UB+} imply that
\begin{equation}\label{ER-UB++}
E_{R,\pi}(\rho) \leq \frac{m-1}{s} \, F_{\Hs}(sE).
\end{equation}

Since $E_{R,\pi}^{\infty}(\rho)\leq E_{R,\pi}(\rho)$ for any state $\rho$ and $F_{\Hs}$ is non-decreasing, inequality~\eqref{ER-UB++} shows that the continuity bound~\eqref{nREE-CB-1} for $E_{R,\pi}^*=E_{R,\pi}^{\infty}$ holds trivially if $\varepsilon\geq 1/2$. Hence, from now on we will assume that $\varepsilon<1/2$. For a given positive integer $u$ we have that~\cite{tightuniform}
\bb
\left| E_{R,\pi}(\rho^{\otimes u})-E_{R,\pi}(\sigma^{\otimes u}) \right| &\leq \sum_{v=1}^u \left|E_{R,\pi}\left(\rho^{\otimes v}\otimes\sigma^{\otimes(u-v)}\right)-E_{R,\pi}\left(\rho^{\otimes (v-1)}\otimes\sigma^{\otimes(u-v+1)}\right)\right| \\
&\leq \sum_{v=1}^u \left|E_{R,\pi}\left(\rho\otimes\omega_v\right)-E_{R,\pi}\left(\sigma\otimes\omega_v\right)\right|,
\ee
where $\omega_v=\rho^{\otimes (v-1)}\otimes\sigma^{\otimes(u-v)}$. The assumption $\Tr \Hs\rho_{A_{[s]}},\, \Tr \Hs \sigma_{\As}\leq sE$ together with inequality~\eqref{ER-UB++} for the system $A_1^{u}\ldots A_s^{u}$, where $A_k^u$ denotes $u$ copies of $A_k$, implies that all terms in the above inequality are finite. Thus, in order to prove the continuity bound~\eqref{nREE-CB-1} for $E_{R,\pi}^*=E_{R,\pi}^{\infty}$, it suffices to show that
\begin{equation}\label{RE-T-pi}
\left|E_{R,\pi}\left(\rho\otimes\omega_v\right)-E_{R,\pi}\left(\sigma\otimes\omega_v\right)\right|\leq \frac{m-1}{s} \sqrt{2\varepsilon}\, F_{\Hs}\!\left(\frac{sE}{\varepsilon}\right)+g\big(\sqrt{2\varepsilon}\big)\qquad \forall\ v\, .
\end{equation}
This can be done by using the arguments from Ref.~\cite[proof of Theorem~1]{Shirokov-AFW-1}, as we explain now.

Let $\hat{\rho}$ and $\hat{\sigma}$ denote purifications of the states $\rho$ and $\sigma$ with the property that $\delta\coloneqq \frac{1}{2}\|\hat{\rho}-\hat{\sigma}\|_1=\sqrt{2\varepsilon}$ (such purifications exist thanks to the Fuchs--van de Graaf inequalities~\cite{Fuchs1999} and Uhlmann's theorem~\cite{Uhlmann-fidelity}). Define $\hat{\omega}_v \coloneqq \hat{\rho}^{\otimes (v-1)}\otimes\hat{\sigma}^{\otimes(u-v)}$, and note that $\hat{\rho}'_v \coloneqq \hat{\rho}\otimes\hat{\omega}_v$ and $\hat{\sigma}'_v\coloneqq \hat{\sigma}\otimes\hat{\omega}_v$ are purifications of the states $\rho'_v\coloneqq \rho\otimes\omega_v$ and $\sigma'_v\coloneqq \sigma\otimes\omega_v$, respectively. Moreover, it holds that $\frac{1}{2}\|\hat{\rho}'_v-\hat{\sigma}'_v\|_1=\delta$.

Now, construct the pure states $\hat{\tau}_{\pm}=\delta^{-1}(\hat{\rho}-\hat{\sigma})_{\pm}$, where $X_\pm$ denote the positive and negative part of the self-adjoint operator $X$. Since these are states over a system comprising $A_1\ldots A_m$ as well as a purifying ancilla, we can consider the reduced states on $A_1\ldots A_m$, denoted by $\tau_{\pm}=(\hat{\tau}_{\pm})_{A_1\ldots A_m}$.

Since $\Tr \Hs \rho_{\As},\, \Tr \Hs\sigma_{\As}\leq sE$, the estimate in Ref.~\cite[proof of Theorem~1]{Shirokov-AFW-1} implies the key bound $\Tr \Hs (\tau_{\pm})_{\As}\leq sE/\varepsilon$ on the average (local) energy of $\tau_\pm$. In light of this, inequality~\eqref{ER-UB++} entails that
\begin{equation}\label{RE-T-3}
E_{R,\pi}(\tau_{\pm})\leq  \frac{m-1}{s}\, F_{\Hs}(sE/\varepsilon).
\end{equation}

By applying the main trick from Ref.~\cite[proof of Theorem~1]{Shirokov-AFW-1} to the states $\hat{\rho}'_v$, $\hat{\sigma}'_v$ and $\delta^{-1}(\hat{\rho}'_v-\hat{\sigma}'_v)_{\pm} = \hat{\tau}_{\pm}\otimes\hat{\omega}_v$ (instead of $\hat{\rho}$, $\hat{\sigma}$ and $\hat{\tau}_{\pm}$) and by using the convexity of $E_{R,\pi}$ and the validity of inequality~\eqref{F-p-1} for this function we obtain
\begin{equation}\label{RE-T-1}
\left|E_{R,\pi}(\rho'_v)-E_{R,\pi}(\sigma'_v)\right|\leq \delta\left|E_{R,\pi}(\tau_+\!\otimes\omega_v)-E_{R,\pi}(\tau_-\!\otimes\omega_v)\right|+g(\delta).
\end{equation}
Assume that $E_{R,\pi}(\tau_+\!\otimes\omega_v)\geq E_{R,\pi}(\tau_-\!\otimes\omega_v)$. By the subadditivity of $E_{R,\pi}$ we have $E_{R,\pi}(\tau_+\!\otimes\omega_v)\leq E_{R,\pi}(\tau_+)+E_{R,\pi}(\omega_v)$, while the definition of $E_{R,\pi}$ and the monotonicity of the relative entropy imply that $E_{R,\pi}(\tau_-\!\otimes\omega_v)\geq E_{R,\pi}(\omega_v)$ (cf.\ Ref.~\cite{tightuniform}). Hence,
\begin{equation}\label{RE-T-2}
|E_{R,\pi}(\tau_+\!\otimes\omega_v)-E_{R,\pi}(\tau_-\!\otimes\omega_v)|\leq \max\left\{E_{R,\pi}(\tau_-),E_{R,\pi}(\tau_+)\right\}.
\end{equation}
Inequalities~\eqref{RE-T-3},~\eqref{RE-T-1}) and~\eqref{RE-T-2} together imply~\eqref{RE-T-pi}.

By the reasoning in Ref.~\cite[Remark~6]{Shirokov-AFW-3}, the continuity bounds~\eqref{ER-CB} and~\eqref{nREE-CB-1} for $E_{R,\pi}^*=E_{R,\pi}^{\infty}$ allow us to obtain~\eqref{nREE-CB-2} for $E_{R,\pi}^*=E_{R,\pi}^{\infty}$ by using the arguments from Ref.~\cite[proof of Theorem~2]{Shirokov-AFW-3} with $f=E_{R,\pi}^{\infty}$.

Assume now that $H$ is the Hamiltonian~\eqref{canonical_Hamiltonian} of the $\ell$-mode quantum oscillator. In this case, we can take $G$ to be the function $G_{\ell,\omega}$ in~\eqref{F-ub+}; this allows us to write~\eqref{nREE-CB-2} in the explicit form~\eqref{exp-form-2}. To prove the last claim, note that the function $G_{\ell,\omega}$ satisfies condition~(vi) in Section~\ref{energy_constraints_subsec}~\cite[Section~3.2]{Shirokov-AFW-2}. Note also that if $m=2$ the condition from the last claim in Ref.~\cite[Theorem~2]{Shirokov-AFW-3} holds for the functions $E_{R}$ and $E^{\infty}_R$ in the cases $s=1$ and $s=2$. Indeed, in both cases the first relation in this condition  is proved by using a product state with appropriate marginal energies, while
the second relation is proved by using a pure state $\rho$ in $\D(\HH_{A_1A_2})$ such that $\rho_{A_k}$ is the Gibbs state $\gamma(E)\coloneqq \frac{e^{-\lambda H}}{\Tr e^{-\lambda H}}$ of system $A_k$, $k=1,2$, where $\lambda$ is determined by the equation $E\Tr e^{-\lambda H}=\Tr H e^{-\lambda H}$~\cite{Wehrl}, since $\Tr H \gamma(E)=E$ and
\bb
E_{R}(\rho)=E_{R}^{\infty}(\rho)=S(\gamma(E))=F_{H}(E).
\ee
Thus, the asymptotic tightness of the continuity bound~\eqref{nREE-CB-2} for $E_{R,\pi}^*=E_{R}$ in both cases $s=1,2$ follows directly from the last claim in Ref.~\cite[Theorem~2]{Shirokov-AFW-3}, while the asymptotic tightness of the continuity bound~\eqref{nREE-CB-2} for $E_{R,\pi}^*=E_{R}^{\infty}$ can be shown easily by using the arguments from the proof of the last assertion of Ref.~\cite[Theorem~2]{Shirokov-AFW-3}.
\end{proof}

The continuity bounds in~\eqref{nREE-CB-1} with $s=m-1$ imply the following

\begin{cor}\label{m-cont} Let $A_1,\ldots, A_{m}$ be arbitrary quantum systems, and let $\pi$ any non-empty set of partitions of $\{1,\ldots, m\}$. If $H_{1},\ldots, H_{m-1}$ are positive operators on the Hilbert spaces $\HH_{A_1},\ldots,\HH_{A_{m-1}}$ satisfying condition~\eqref{H-cond+}, then
\begin{itemize}
  \item the function $E_{R,\pi}$ is uniformly continuous on the set of states $\rho$ in $\D(\HH_{A_1\ldots A_m})$ such that  $\sum_{k=1}^{m-1}\mathrm{Tr} \rho_{A_k}H_{k}\leq E$ for any $E>0$;
  \item the function $E_{R,\pi}$ is asymptotically continuous in the following sense~\cite{Eisert2002}: if $(\rho_n)_{n\in \N}$ and $(\sigma_n)_{n\in \N}$ are any sequences of states such that
\bb
\rho_n,\sigma_n \in \D(\HH_{A_1\ldots A_m}^{\otimes n}),\qquad \sum_{k=1}^{m-1} \Tr H_{k,n}\rho^{A^n_k}_n ,\, \sum_{k=1}^{m-1}\Tr H_{k,n} \sigma^{A^n_k}_n\leq nE,\quad \forall\ n,\qquad \text{and}\qquad  \lim_{n\to+\infty}\|\rho_n-\sigma_n\|_1=0\, ,
\ee
where $A^{n}_k$ denotes $n$ copies of $A_k$, $H_{k,n}$ is the positive operator on $\HH_{A_k}^{\otimes n}$ defined in~\eqref{H-n} with $H=H_k$, and $E>0$ is a finite positive number, then
\bb
\lim_{n\to+\infty}\frac{|E_{R,\pi}(\rho_n)-E_{R,\pi}(\sigma_n)|}{n}=0.
\ee
\end{itemize}
The above properties are also valid for the function $E_{R,\pi}^{\infty}$ if $\pi$ composed of one partition only.
\end{cor}
\begin{proof}

The assertion about uniform continuity of the functions $E_{R,\pi}$ and $E_{R,\pi}^{\infty}$ follows directly from continuity bound~\eqref{nREE-CB-1} with $s=m-1$  (speaking about $E_{R,\pi}^{\infty}$ we assume that $\pi$ composed of one partition).

To prove of the asymptotic continuity of the functions $E_{R,\pi}$ and $E_{R,\pi}^{\infty}$ note that $F_{(H_{[m-1]})_n}(E)=n F_{H_{[m-1]}} (E/n)$ for each $n$, where $H_{[m-1]}$ is the operator on $\HH_{A_1\ldots A_{m-1}}$ defined in~\eqref{Hm} with $s=m-1$, and therefore $(H_{[m-1]})_n$ is the operator on $\HH_{A_1\ldots A_{m-1}}^{\otimes n}$ obtained by setting $H=H_{[m-1]}$ in~\eqref{H-n}. So, it follows from the continuity bound~\eqref{nREE-CB-1} with $s=m-1$ that
\begin{equation}\label{SE-ucb-k}
    \frac{|E_{R,\pi}^{*}(\rho_n)-E_{R,\pi}^{*}(\sigma_n)|}{n}\leq  \sqrt{2\varepsilon_n}\, F_{H_{[m-1]}}\!\left(\frac{E}{\varepsilon_n}\right)+\frac{g\big(\sqrt{2\varepsilon_n}\big)}{n},\quad E_{R,\pi}^{*}=E_{R,\pi},\, E_{R,\pi}^{\infty},
\end{equation}
where $\varepsilon_n=\frac{1}{2}\|\rho_n-\sigma_n\|_1$. Since $\lim_{n\to+\infty} \varepsilon_n = 0$ by hypothesis and $F_{H_{[m-1]}}(E)=o\big(\sqrt{E}\big)$ as $E\to+\infty$, by Ref.~\cite[Lemma~2]{Shirokov-AFW-3} and Ref.~\cite[Lemma~1]{Shirokov-AFW-1}, the right-hand side of~\eqref{SE-ucb-k} tends to zero as $n\to+\infty$.
\end{proof}

\section{Conclusions and outlook}

In this paper we established the surprising fact that the infimum defining the relative entropy of entanglement is always achieved, also in infinite-dimensional systems. This has been shown to be a consequence of a much more general result, stating that the relative entropy distance to a (convex) set of free states $\FF$, called the relative entropy of resource, is always achieved and moreover lower semi-continuous, provided that the cone generated by $\FF$ is closed in the weak*-topology (Theorem~\ref{achievable_relent_thm}). We employed this latter result to establish a dual variational formula by means of which the relative entropy of resource can be expressed as a maximisation instead of a minimisation (Theorem~\ref{variational_thm}). In doing so, we generalised several results of classic matrix analysis, most notably Lieb's three-matrix inequality, to the infinite-dimensional case (Appendix~\ref{Lieb_3_app}). The applications we envision for our dual formula are on the one hand computational, and on the other rest on the theoretical framework proposed in~\cite{Berta2017, nonclassicality}, where expressions of that kind are used to establish properties such as the super-additivity.

We further identified a general set of conditions implying the above topological property (Theorem~\ref{w*_closed_condition_thm}), and showed how to apply them to a variety of quantum resource theories, namely, that of multi-partite entanglement (Section~\ref{multipartite_entanglement_subsec}), NPT entanglement (Section~\ref{NPT_entanglement_subsec}), non-classicality, Wigner negativity and more generally $\lambda$-negativity (Section~\ref{lambda_negativity_subsec}), and finally non-Gaussianity (Section~\ref{non_Gaussianity_subsec}). Interestingly, the topological condition we have pinpointed is obeyed in almost all cases of practical interest, and can thus be regarded as a natural regularity assumption to impose on arbitrary infinite-dimensional quantum resource theories. For example, one could imagine to employ it to generalise the results of Ref.~\cite{Brandao-Gour}, which rest on a key identity between the smoothed regularised (generalised) robustness and the regularised relative entropy, to infinite-dimensional resources. Also, it would be interesting to extend the methods in this paper to address other resource quantifiers involving optimisations over non-compact sets, or else channel resource quantifiers~\cite{Gour-Winter}.

In the second part of our paper we focused our attention on the relative entropy of (NPT) entanglement, the Rains bound, regularisations thereof, and the corresponding multi-partite generalisations. We have established tight uniform continuity bounds for all those functions in the presence of an energy constraint. Conceptually, those bounds complement the general statement of lower semi-continuity, and prove that much stronger regularity properties can be obtained if one looks only at energy-bounded sets of states. \tcb{We speculate that even tighter constraints could possibly be derived by leveraging techniques recently proposed by Becker, Datta, and Jabbour~\cite{Becker2023}.}

\medskip
\textbf{Acknowledgements.} L.\ Lami\ is supported by the Alexander von Humboldt Foundation. He thanks Martin B.\ Plenio and Bartosz Regula for several interesting discussions about entanglement and infinite-dimensional resource theories. The work of M.\ Shirokov was performed at the Steklov International Mathematical Center and supported by the Ministry of Science and Higher Education of the Russian Federation (agreement no.\ 075-15-2019-1614). The authors are grateful to A.\ S.\ Holevo and G.\ G.\ Amosov for useful and motivating discussions. They also thank an anonymous referee at the `17th Conference on the Theory of Quantum Computation, Communication and Cryptography' (TQC 2022) for insightful comments.

\bibliographystyle{unsrt}
\bibliography{biblio}

\appendix

\section{On the differential of the operator logarithm in infinite dimensions} \label{app_differential_log}

Throughout this appendix we show how to rigorously derive~\eqref{inequality_differentiated} and~\eqref{Gamma_A} in the proof of Lemma~\ref{variational_technical_lemma}. We state the following:

\begin{lemma} \label{differential_log_lemma}
Let $\xi,\xi'\in \T_+(\HH)\cap B_1$ be two positive semi-definite trace class operators with trace at most $1$. Let $\rho\in \D(\HH)$ be a density operator such that $-\Tr \rho \ln \xi<\infty$, so that $-\Tr \rho \ln \left((1-\lambda)\xi + \lambda\xi'\right)<\infty$ for all $\lambda\in [0,1)$. Assume further that
\bb
\liminf_{\lambda\to 0^+} \frac{1}{\lambda}\left( \Tr \rho \ln \xi -\Tr \rho \ln \left((1-\lambda)\xi + \lambda\xi'\right) \right) \geq c > -\infty\, .
\label{assumption_further}
\ee
Then
\bb
\lim_{\lambda\to 0^+} \frac{1}{\lambda}\left( \Tr \rho \ln \xi -\Tr \rho \ln \left((1-\lambda)\xi + \lambda\xi'\right) \right) = \Tr \rho \,\Gamma_\xi(\xi - \xi') = 1 - \Tr \rho\, \Gamma_\xi (\xi') \geq c\, ,
\label{differential_log}
\ee
where as in~\eqref{Gamma_A} we convene to define $\Tr \rho\, \Gamma_\xi (X) \coloneqq \int_0^\infty ds\, \Tr \rho\, \frac{1}{\xi + s\id} X \frac{1}{\xi + s\id}$; the integral on the right-hand side is absolutely convergent when either $X=\xi-\xi'$ or $X=\xi'$.
\end{lemma}

Note that a quick application of Lemma~\ref{differential_log_lemma} with $c\coloneqq \Tr[\xi-\xi']$ yields precisely~\eqref{inequality_differentiated} and~\eqref{Gamma_A}.

\begin{proof}[Proof of Lemma~\ref{differential_log_lemma}]
The first claim is just~\eqref{boundedness_for_all_lambda}, so we proceed to prove~\eqref{differential_log}. For some $C\in \T_+(\HH)\cap B_1$ (see~\eqref{B_1}) with spectral decomposition $C=\sum_j c_j \ketbra{c_j}$, denoting with $\rho=\sum_i p_i \ketbra{e_i}$ the spectral decomposition of $\rho$, we see that
\bb
-\Tr \rho \ln C &= \sum_{i,j} p_i \left|\braket{e_i|c_j}\right|^2 (- \ln c_j) \\
&\eqt{(i)} \sum_{i,j} p_i \left|\braket{e_i|c_j}\right|^2 \int_0^\infty ds\, \left( \frac{1}{c_j+s} - \frac{1}{1+s} \right) \\
&\eqt{(ii)} \int_0^\infty ds \sum_{i,j} p_i \left|\braket{e_i|c_j}\right|^2 \left( \frac{1}{c_j+s} - \frac{1}{1+s} \right) \\
&= \int_0^\infty ds \Tr \rho\left(\frac{1}{C+s\id} - \frac{\id}{1+s}\right) ,
\ee
where in~(i) we used a well-known integral representation of the logarithm, and in~(ii) we applied Tonelli's theorem to exchange integral and sum --- this is possible because $c_j\leq 1$ for all $j$. Clearly, since $\frac{1}{C+s\id} \geq \frac{\id}{1+s}$ for all $s>0$, the function to be integrated is non-negative; thus, if $-\Tr \rho \ln C<\infty$ the integral is absolutely converging. We now apply this insight to manipulate the difference between two such integrals, obtained by setting $C=\xi$ and $C=(1-\lambda)\xi + \lambda \xi'$. We have that
\bb
\frac{1}{\lambda}\left( \Tr \rho \ln \xi -\Tr \rho \ln \left((1-\lambda)\xi + \lambda\xi'\right) \right) = \int_0^\infty ds\, f_\lambda(s)\, ,
\label{integral_f_lambda}
\ee
where
\bb
f_\lambda(s) \coloneqq \Tr \rho\, \frac{C_s(\lambda) - C_s(0)}{\lambda}\, ,\qquad C_s(\lambda) \coloneqq \frac{1}{(1-\lambda)\xi + \lambda\xi'+s\id}\, .
\label{f_lambda_C_lambda}
\ee
From how it was obtained it is clear that the integral in~\eqref{integral_f_lambda} is absolutely converging.

Now, fix $s>0$; since $C_s(\cdot)$ defined by~\eqref{f_lambda_C_lambda} is convex in the operator sense~\cite{SIMON-LOEWENER}, we see that $\frac{C_s(\lambda) - C_s(0)}{\lambda}$ is monotonically non-decreasing in $\lambda\in [0,1]$ as an operator. Consequently, $f_\lambda(s)$ is also monotonically non-decreasing in $\lambda$. We now claim that
\bb
\lim_{\lambda\to 0^+} \left\| \frac{C_s(\lambda) - C_s(0)}{\lambda} - \frac{1}{\xi+s\id}\, (\xi - \xi')\, \frac{1}{\xi+s\id}\right\|_\infty = 0\, ,
\label{operator_norm_convergence}
\ee
i.e.\ $\frac{C_s(\lambda) - C_s(0)}{\lambda}$ converges in operator norm to $\frac{1}{\xi+s\id}\, (\xi - \xi')\, \frac{1}{\xi+s\id}$. To establish~\eqref{operator_norm_convergence}, write first
\bb
\frac{C_s(\lambda) - C_s(0)}{\lambda} &= \frac{1}{\lambda} \left(\frac{1}{(1-\lambda)\xi + \lambda\xi'+s\id} - \frac{1}{\xi +s\id} \right) \\
&= \frac{1}{\lambda} \frac{1}{(1-\lambda)\xi + \lambda\xi'+s\id}\left(\xi + s\id - \left((1-\lambda)\xi + \lambda\xi'+s\id \right) \right)\frac{1}{\xi +s\id} \\
&= \frac{1}{(1-\lambda)\xi + \lambda\xi'+s\id}\left(\xi - \xi' \right)\frac{1}{\xi +s\id}\, ,
\ee
so that
\bb
\frac{C_s(\lambda) - C_s(0)}{\lambda} - \frac{1}{\xi+s\id}\, (\xi - \xi')\, \frac{1}{\xi+s\id} &= \left(\frac{1}{(1-\lambda)\xi + \lambda\xi'+s\id} - \frac{1}{\xi +s\id} \right) \left(\xi - \xi' \right) \frac{1}{\xi +s\id} \\
&= \lambda\, \frac{1}{(1-\lambda)\xi + \lambda\xi'+s\id}\left(\xi - \xi'\right) \frac{1}{\xi +s\id} \left(\xi - \xi' \right) \frac{1}{\xi +s\id}\, ,
\ee
and finally, remembering that $\|\xi-\xi'\|_\infty\leq 1$ because $\xi,\xi'\in \T_+(\HH)\cap B_1$,
\bb
\left\| \frac{C_s(\lambda) - C_s(0)}{\lambda} - \frac{1}{\xi+s\id}\, (\xi - \xi')\, \frac{1}{\xi+s\id}\right\|_\infty \leq \frac{\lambda}{s^3} \ctends{}{\lambda\to 0^+}{-0.1ex} 0\, .
\ee
This proves~\eqref{operator_norm_convergence}. Now, looking at the definition of $f_\lambda(s)$ in~\eqref{f_lambda_C_lambda} and putting all together, one sees that
\bb
\lim_{\lambda \to 0^+} f_\lambda(s) = \inf_{\lambda>0} f_\lambda(s) = \Tr \rho\, \frac{1}{\xi+s\id}\, (\xi - \xi')\, \frac{1}{\xi+s\id} \eqqcolon f(s)\, . 
\ee
We are now ready to use Beppo Levi's monotone convergence theorem (see e.g.~\cite[Theorem~11.1(ii)]{SCHILLING}), applicable thanks to~\eqref{assumption_further}, which yields the absolute integrability of $f$ and the identities
\bb
\lim_{\lambda\to 0^+} \int_0^\infty ds\, f_\lambda(s) = \inf_{\lambda>0} \int_0^\infty ds\, f_\lambda(s) = \int_0^\infty ds\, \inf_{\lambda>0} f_\lambda(s) = \int_0^\infty ds\, \lim_{\lambda\to 0^+} f_\lambda(s) = \int_0^\infty ds\, f(s) \geq c\, .
\ee
Again thanks to absolute integrability we can split the last integral, arriving at
\bb
c \leq \int_0^\infty ds\, f(s) = \int_0^\infty ds\, \Tr \rho\,\frac{\xi}{(\xi+s)^2} - \int_0^\infty ds\, \Tr \rho\, \frac{1}{\xi+s\id}\, \xi'\, \frac{1}{\xi+s\id} = 1 - \int_0^\infty ds\, \Tr \rho\, \frac{1}{\xi+s\id}\, \xi'\, \frac{1}{\xi+s\id}\, ,
\ee
which concludes the proof.
\end{proof}

\section{Lieb's three-operator inequality} \label{Lieb_3_app}

Junge and LaRacuente~\cite{Junge2021} (see also~\cite{Hollands2021}) have recently established a version of the multivariate Golden--Thompson inequality from~\cite{Sutter2017} that works in infinite dimensions as well. However, their result is expressed in a `unitarily rotated' form that is not prima facie equivalent to the generalised Lieb's three-matrix inequality that we need here. Namely, Junge and LaRacuente prove that~\cite[Theorem~1.2]{Junge2021}
\bb
\ln \left\| \exp\left[\frac{\ln A}{p} + \sumno_k X_k \right] \right\|_p \leq \int_{-\infty}^{+\infty} dt\, \beta_0(t)\, \ln \left\|\left( \prod\nolimits_k \exp\left[(1+it) X_k\right] \right) A^{1/p} \right\|_p
\label{Junge_LaRacuente}
\ee
for all trace class operators $0<A\in \T(\HH)$ and finite collections $\{X_k\}_k$ of bounded self-adjoint operators $X_k=X_k^\dag \in \B(\HH)$. Here, $\beta_0:\R \to \R_+$, given by $\beta_0(t)\coloneqq \frac{\pi}{2\left( \cosh(\pi t)+1\right)}$, is a fixed probability density function on $\R$. It is not obvious how to deduce an inequality of the form~\cite[Theorem~7]{lieb73c} from~\eqref{Junge_LaRacuente}. The way to do so is detailed in~\cite[Appendix~E]{Sutter2017} for the finite-dimensional case. The purpose of this appendix is to extend this derivation to the infinite-dimensional case as well.

\begin{lemma} \label{Lieb_3_matrix_lemma}
Let $A\in \T_+(\HH)$ be a positive semi-definite trace class operator, and for some $\delta>0$ let $\delta \id \leq B,C\in \B(\HH)$ be two positive semi-definite bounded operators. Then
\bb
\Tr e^{\ln A + \ln B - \ln C} \leq \int_0^\infty ds\, \Tr A\, \frac{1}{C+s\id}\, B\, \frac{1}{C+s\id} < \infty\, ,
\label{Lieb_3_matrix}
\ee
where the left-hand side is interpreted as in~\eqref{interpretation} when $A$ is not strictly positive definite, and the integral on the right-hand side is absolutely converging.
\end{lemma}

\begin{proof}
We start by arguing that it suffices to prove the claim when $A>0$ is strictly positive definite. Indeed, assume that we have addressed that case; given some $A\geq 0$ that is not strictly positive definite, we can pick any trace class $\Delta>0$ and for $\epsilon>0$ define $A_\epsilon \coloneqq A + \epsilon \Delta$. Clearly, $A_\epsilon>0$ and also $A_\epsilon \geq A$, so we would obtain that
\bb
\Tr e^{\ln A + \ln B - \ln C} &\leq \Tr e^{\ln A_\epsilon + \ln B - \ln C} \\
&\leq \int_0^\infty ds\, \Tr A_\epsilon\, \frac{1}{C+s\id}\, B\, \frac{1}{C+s\id} \\
&= \int_0^\infty ds\, \Tr A\, \frac{1}{C+s\id}\, B\, \frac{1}{C+s\id} + \epsilon \int_0^\infty ds\, \Tr \Delta\, \frac{1}{C+s\id}\, B\, \frac{1}{C+s\id}\, ,
\ee
where the first inequality comes from the monotonicity of the function in~\eqref{trace_perturbed_function}, as established by Lemma~\ref{Petz_amended_lemma}.\footnote{There is nothing circular here, as the result of the present Lemma~\ref{Lieb_3_matrix_lemma} is used in the proof of Lemma~\ref{variational_technical_lemma} but not in that of Lemma~\ref{Petz_amended_lemma}.} Since the term proportional to $\epsilon$ in the last line of the above inequality is finite (this will follow from the $A>0$ case of~\eqref{Lieb_3_matrix}), taking the limit $\epsilon \to 0^+$ yields the general case of~\eqref{Lieb_3_matrix}.

Therefore, in what follows we assume that $A>0$. Let $(P_n)_{n\in \N}$ be a sequence of orthogonal projectors $P_n:\HH\to V_n\simeq \C^n$, where (a)~for all $n$ the subspace $V_n$ is invariant under $A$ --- for example, it may be the linear span of $n$ eigenvectors; and (b)~$\Pi_n \coloneqq P_n^\dag P_n:\HH\to \HH$ converges strongly to the identity, which we write $\Pi_n \tendsn{s} \id$. Set
\bb
C'_n \coloneqq \delta\id + P_n (C-\delta \id) P_n^\dag \oplus 0 = P_n C P_n^\dag \oplus \big(\delta Q_n Q_n^\dag\big) \eqqcolon C_n \oplus \big(\delta Q_n Q_n^\dag\big)\, ,
\label{approximating_C}
\ee
where the direct sum is with respect to the decomposition $\HH=V_n\oplus V_n^\perp$, we denoted with $Q_n$ the orthogonal projector onto $V_n^\perp$, and we set $C_n \coloneqq P_n C P^\dag_n$. Then clearly $C'_n\geq \delta \id$ for all $n$, and moreover $C'_n\tendsn{s} C$.
Putting all together:
\bb
\Tr e^{\ln A + \ln B - \ln C} &= 
\left\| \exp\left[ \frac12 \ln A + \frac12 \ln B - \frac12 \ln C \right] \right\|_2^2 \\
&\leqt{(i)} \exp \left[ 2\int_{-\infty}^{+\infty} dt\, \beta_0(t)\, \ln \left\| B^{\frac{1+it}{2}} C^{-\frac{1+it}{2}} A^{1/2} \right\|_2 \right] \\
&\leqt{(ii)} \int_{-\infty}^{+\infty} dt\, \beta_0(t) \left\| B^{\frac{1+it}{2}} C^{-\frac{1+it}{2}} A^{1/2} \right\|_2^2 \\
&= \int_{-\infty}^{+\infty} dt\, \beta_0(t) \Tr B^{\frac{1+it}{2}} C^{-\frac{1+it}{2}} A\, C^{-\frac{1-it}{2}} B^{\frac{1-it}{2}} \\
&\eqt{(iii)} \int_{-\infty}^{+\infty} dt\, \beta_0(t) \Tr B\, C^{-\frac{1+it}{2}} A\, C^{-\frac{1-it}{2}} \\
&\eqt{(iv)} \lim_{n\to\infty} \int_{-\infty}^{+\infty} dt\, \beta_0(t) \Tr B\, (C'_n)^{-\frac{1+it}{2}} A\, (C'_n)^{-\frac{1-it}{2}} \\
&\eqt{(v)} \lim_{n\to\infty} \left( \int_{-\infty}^{+\infty} dt\, \beta_0(t) \Tr B_n\, C_n^{-\frac{1+it}{2}} A_n\, C_n^{-\frac{1-it}{2}} + \frac1\delta \Tr B(\id-\Pi_n) A (\id-\Pi_n) \right) \\
&\eqt{(vi)} \lim_{n\to\infty} \left( \int_{0}^{+\infty} ds\, \Tr B_n\, \frac{1}{C_n+s\id_n} \, A_n\, \frac{1}{C_n+s\id_n} + \frac1\delta \Tr B(\id-\Pi_n) A (\id-\Pi_n) \right) \\
&\eqt{(vii)} \lim_{n\to\infty} \int_{0}^{+\infty} ds\, \Tr B\, \frac{1}{C'_n+s\id} \, A\, \frac{1}{C'_n+s\id} \\
&\eqt{(viii)} \int_{0}^{+\infty} ds\, \Tr B\, \frac{1}{C +s\id}\, A\, \frac{1}{C+s\id}
\label{long_chain_C_prime}
\ee
Here, (i)~is an application of Junge and LaRacuente's result~\eqref{Junge_LaRacuente} with $p=2$, $k=1,2$, $X_1=\frac12 \ln B$, and $X_2=-\frac12 \ln C$; (ii)~follows from the convexity of the exponential $\exp:\R\to \R$ (remember that $\beta_0$ is a probability density function); and in~(iii) we leveraged the cyclicality of the trace. The justification of~(iv) is slightly more complex. Since $C'_n \tendsn{s} C$ and $C'_n, C\geq \delta\id$, thanks to~\cite[Propositions~10.1.9 and~10.1.13(a)]{OLIVEIRA} we see that $(C'_n)^{-\frac{1\pm it}{2}} \tendsn{s} C^{-\frac{1\pm it}{2}}$; as we did in~\eqref{convergence_MM_n}, this can be shown to imply that
\bb
(C'_n)^{-\frac{1+it}{2}} A\, (C'_n)^{-\frac{1-it}{2}} \tendsn{tn} C^{-\frac{1+it}{2}} A\, C^{-\frac{1-it}{2}}\, ,
\ee
in turn guaranteeing that
\bb
\Tr B\, (C'_n)^{-\frac{1+it}{2}} A\, (C'_n)^{-\frac{1-it}{2}} \tends{}{n\to +\infty} \Tr B\, C^{-\frac{1+it}{2}} A\, C^{-\frac{1-it}{2}}
\label{pointwise_convergence_C_prime}
\ee
because $B$ is bounded. Now, given that
\bb
\beta_0(t) \left| \Tr B\, (C'_n)^{-\frac{1+it}{2}} A\, (C'_n)^{-\frac{1-it}{2}} \right| \leq \beta_0(t) \|B\|_\infty \left\|(C'_n)^{-\frac{1+it}{2}} \right\|_\infty^2 \|A\|_1 \leq \frac{\|B\|_\infty \|A\|_1}{\delta}\, \beta_0(t)
\ee
and moreover $\int_{-\infty}^{+\infty} dt\, |\beta_0(t)| = 1$, the identity~(iv) follows from~\eqref{pointwise_convergence_C_prime} thanks to Lebesgue's dominated convergence theorem.
Continuing with the justification of the derivation in~\eqref{long_chain_C_prime}: in~(v) we decomposed the trace exploiting the fact that $V_n$ is invariant under the action of $A$, introducing the operators $B_n \coloneqq P_n B P_n^\dag$ and $A_n\coloneqq P_n A P_n^\dag$ on $V_n$; in~(vi) we massaged the first term, which is the trace of an $n\times n$ matrix, by means of the identity~\cite[Eq.~(96)]{Sutter2017}
\bb
\int_{-\infty}^{+\infty} dt\, \beta_0(t)\, x^{-\frac{1+it}{2}} y^{-\frac{1-it}{2}} = \int_0^\infty ds\, \frac{1}{(x+s)(y+s)}\, ,
\ee
valid for all $x,y>0$ (this step is the same as in~\cite[Appendix~E]{Sutter2017}, where it is explained in more detail); (vii)~descends from the chain of equalities
\bb
&\int_{0}^{+\infty} ds\, \Tr B\, \frac{1}{C'_n+s\id} \, A\, \frac{1}{C'_n+s\id} \\
&\qquad = \int_{0}^{+\infty} ds\, \Tr B_n\, \frac{1}{C_n+s\id_n} \, A_n\, \frac{1}{C_n+s\id_n} + \int_{0}^{+\infty} ds\, \frac{1}{(\delta+s)^2}\, \Tr B\, (\id-\Pi_n)\, A\, (\id-\Pi_n) \\
&\qquad = \int_{0}^{+\infty} ds\, \Tr B_n\, \frac{1}{C_n+s\id_n} \, A_n\, \frac{1}{C_n+s\id_n} + \frac1\delta \Tr B\, (\id-\Pi_n)\, A\, (\id-\Pi_n) \, ,
\ee
where $\id_n$ denotes the identity operator on $V_n$ (essentially, the identity matrix of size $n$); finally, (viii)~can be deduced once again thanks to Lebesgue's dominated convergence theorem, because (a)~due to $C'_n\geq \delta \id$ we have
\bb
\left| \Tr B\, \frac{1}{C'_n +s\id}\, A\, \frac{1}{C'_n+s\id} \right| \leq \|B\|_\infty \|A\|_1 \left\|\frac{1}{C'_n+s\id}\right\|_\infty^2 \leq \frac{\|B\|_\infty \|A\|_1}{(\delta+s)^2}\, ,
\ee
(b)~$\int_0^\infty \frac{1}{(\delta+s)^2}=\frac{1}{\delta}<
\infty$, and (c)~thanks to~\cite[Propositions~10.1.9 and~10.1.13(a)]{OLIVEIRA}, from $C'_n\tendsn{s} C$ it follows that $\frac{1}{C'_n+s\id}\tendsn{s} \frac{1}{C+s\id}$ and in turn that $\frac{1}{C'_n+s\id}\, A\, \frac{1}{C'_n+s\id} \tendsn{tn} \frac{1}{C+s\id}\, A\, \frac{1}{C+s\id}$, from which we infer that
\bb
\Tr B\, \frac{1}{C'_n+s\id}\, A\, \frac{1}{C'_n+s\id} \tendsn{} \Tr B\, \frac{1}{C+s\id}\, A\, \frac{1}{C+s\id}\, ,
\ee
precisely as in~\eqref{pointwise_convergence_C_prime}. This concludes the proof of the first inequality in~\eqref{Lieb_3_matrix}. As for the second, we see as in~(a) and~(b) above that the integral in~\eqref{Lieb_3_matrix} is upper bounded by $1/\delta$.
\end{proof}

\section{On the proof of an inequality involving the multi-partite relative entropy of entanglement} \label{inequality_app}

To prove the second inequality in~\eqref{ER-UB} in the infinite-dimensional setting, assume that $\rho$ is a state in $\D(\HH_{A_1\ldots A_m})$ with spectral decomposition $\rho=\sum_{i=1}^{\infty} p_i \varphi_i$, where each $\varphi_i = \varphi_i^{A_1\ldots A_m} = \ketbra{\varphi_i}_{A_1\ldots A_m}$ is a pure state. By Lemma~\ref{Omega} below, for each $i$ there is a fully separable state $\omega_i$ in $\D(\HH_{A_1\ldots A_m})$ such that $D(\varphi_i\|\omega_i)\leq \sum_{k=1}^{m-1}S\big(\varphi_i^{A_k}\big)$. 
Let $\sigma\coloneqq \sum_{i=1}^{\infty} p_i \omega_i$, $\rho_n\coloneqq c_n^{-1} \sum_{i=1}^{n} p_i\varphi_i$, and $\sigma_n \coloneqq c^{-1}_n \sum_{i=1}^{n} p_i \omega_i$, where $c_n\coloneqq \sum_{i=1}^{n} p_i$, for each positive integer $n$. By the joint convexity of the relative entropy (see~\eqref{joint_convexity}) we have
\bb
D(\rho_n\|\sigma_n)\leq c_n^{-1}\sum_{i=1}^n p_i D(\varphi_i\|\omega_i)\qquad \forall\ n\, .
\ee
By using the lower semi-continuity of the relative entropy we obtain
\bb
E_R(\rho) \leq D(\rho\|\sigma) \leq \liminf_{n\to+\infty} D(\rho_n\|\shs\sigma_n) \leq \sum_{i=1}^{\infty} p_i D(\varphi_i\|\omega_i)\leq\sum_{k=1}^{m-1}\sum_{i=1}^{\infty} p_i S\big(\varphi_i^{A_k}\big) \leq \sum_{k=1}^{m-1} S(\rho_{A_k})\, ,
\ee
where the first inequality follows from the full separability of $\sigma$, while the last one is due to the concavity 
of the entropy. This concludes the proof of the second inequality in~\eqref{ER-UB}.

The following lemma is a $m$-partite generalisation of the observation in Ref.~\cite{Vedral1998}.

\begin{lemma}\label{Omega}
For any pure state $\ket{\Psi}$ on $\HH_{A_1\ldots A_m}$ there is a countably decomposable separable state $\sigma$ in $\D(\HH_{A_1\ldots A_m})$ such that $\sigma_{A_k}=\Psi_{A_k}$ for $k=1,\ldots, m$ and
\bb
D(\Psi\|\sigma) \leq\sum_{k=1}^{m-1} S(\Psi_{A_k}).
\ee
\end{lemma}

\begin{proof}
The Schmidt decomposition of $\ket{\Psi}$ with respect to the bi-partition $A_1:A_2\ldots A_m$ implies that
\bb
\ket{\Psi} =\sum_{i_1}\sqrt{p^1_{i_1}} \ket{\varphi^1_{i_1}} \otimes \ket{\psi^1_{i_1}}\, ,
\ee
where $\left(\ket{\varphi^1_{i_1}}\right)_{i_1}$ and $\left(\ket{\psi^1_{i_1}}\right)_{i_1}$ are orthogonal sets of unit vectors in $\HH_{A_1}$ and $\HH_{A_2\ldots A_m}$, respectively, and $\left(p^1_{i_1}\right)_{i_1}$ is a probability distribution. Applying the Schmidt decomposition with respect to the bi-partition $A_2:A_3\ldots A_m$ to any of the vectors $\psi^1_{i_1}$, we obtain
\bb
\ket{\Psi} = \sum_{i_1,i_2}\sqrt{p^1_{i_1}p^2_{i_1i_2}} \ket{\varphi^1_{i_1}} \otimes \ket{\varphi^2_{i_1i_2}} \otimes \ket{\psi^2_{i_1i_2}}\, ,
\ee
where $\left(\ket{\varphi^2_{i_1i_2}}\right)_{i_2}$ and $\left(\ket{\psi^2_{i_1i_2}}\right)_{i_2}$ are orthogonal sets of unit vectors in $\HH_{A_2}$ and $\HH_{A_3\ldots A_m}$, respectively, and $\left(p^2_{i_1i_2}\right)_{i_2}$ is a probability distribution for any given $i_1$.

By repeating this process we get
\begin{equation*}
\ket{\Psi} = \sum_{i_1,i_2,\ldots ,i_{m-1}}\sqrt{p^1_{i_1}p^2_{i_1i_2}\ldots p^{m-1}_{i_1i_2\ldots i_{m-1}}} \ket{\varphi^1_{i_1}} \otimes\ldots \otimes \ket{\varphi^{m-1}_{i_1i_2\ldots i_{m-1}}} \otimes \ket{\psi^{m-1}_{i_1i_2\ldots i_{m-1}}}\, ,
\end{equation*}
where $\left(p^{s}_{i_1i_2\ldots i_{s}}\right)_{i_{s}}$ and $\left(\ket{\varphi^s_{i_1i_2\ldots i_{s}}}\right)_{i_{s}}$ ($s=1,\ldots, m-1$) are, respectively, a probability distribution and an orthogonal set of unit vectors in $\HH_{A_s}$ for any fixed $i_1,i_2,\ldots ,i_{s-1}$, and $\left( \ket{\psi^{m-1}_{i_1i_2\ldots i_{m-1}}}\right)_{i_{m-1}}$ is an orthogonal set of unit vectors in $\HH_{A_m}$ for any fixed $i_1,i_2,\ldots ,i_{m-2}$.

Consider the countably decomposable separable state
\begin{equation}\label{sigma-def}
\sigma=\sum_{i_1,i_2,\ldots ,i_{m-1}} p^1_{i_1}p^2_{i_1i_2}\ldots p^{m-1}_{i_1i_2\ldots i_{m-1}} \varphi_{i_1}^1 \otimes \ldots \otimes \varphi^{m-1}_{i_1i_2\ldots i_{m-1}} \otimes \psi^{m-1}_{i_1i_2\ldots i_{m-1}}
\end{equation}
in $\D(\HH_{A_1\ldots A_m})$, 
By noting that all the addends 
in~\eqref{sigma-def} are mutually orthogonal pure states one can show that the state $\sigma$ has the required properties.
\end{proof}

\end{document}